\documentclass[sn-mathphys-num]{sn-jnl}

\usepackage{graphicx}%
\usepackage{multirow}%
\usepackage{amsmath,amssymb,amsfonts}%
\usepackage{amsthm}%
\usepackage{mathrsfs}%
\usepackage[title]{appendix}%
\usepackage{xcolor}%
\usepackage{textcomp}%
\usepackage{manyfoot}%
\usepackage{booktabs}%
\usepackage{algorithm}%
\usepackage{algorithmicx}%
\usepackage{algpseudocode}%
\usepackage{listings}%

\usepackage{soul,hyperref}
\usepackage[textwidth=2.2cm]{todonotes}
\usepackage{dsfont}
\usepackage{leftidx}
\usepackage{comment}
\usepackage{geometry}

\usepackage{tikz}
\usetikzlibrary{matrix,arrows,patterns,cd}

\usepackage{scalerel}
\usetikzlibrary{svg.path}

\definecolor{orcidlogocol}{HTML}{A6CE39}
\tikzset{
  orcidlogo/.pic={
    \fill[orcidlogocol] svg{M256,128c0,70.7-57.3,128-128,128C57.3,256,0,198.7,0,128C0,57.3,57.3,0,128,0C198.7,0,256,57.3,256,128z};
    \fill[white] svg{M86.3,186.2H70.9V79.1h15.4v48.4V186.2z}
                 svg{M108.9,79.1h41.6c39.6,0,57,28.3,57,53.6c0,27.5-21.5,53.6-56.8,53.6h-41.8V79.1z M124.3,172.4h24.5c34.9,0,42.9-26.5,42.9-39.7c0-21.5-13.7-39.7-43.7-39.7h-23.7V172.4z}
                 svg{M88.7,56.8c0,5.5-4.5,10.1-10.1,10.1c-5.6,0-10.1-4.6-10.1-10.1c0-5.6,4.5-10.1,10.1-10.1C84.2,46.7,88.7,51.3,88.7,56.8z};
  }
}

\newcommand\orcidicon[1]{\href{https://orcid.org/#1}{\mbox{\scalerel*{
\begin{tikzpicture}[yscale=-1,transform shape]
\pic{orcidlogo};
\end{tikzpicture}
}{|}}}}

\makeatletter
\newcommand{\raisemath}[1]{\mathpalette{\raisem@th{#1}}}
\newcommand{\raisem@th}[3]{\raisebox{#1}{$#2#3$}}
\makeatother
\newcommand{\act}[2]{\leftidx{^{\raisemath{-0.5pt}{#1}}}{#2}{}}

\let\Im\relax
\DeclareMathOperator{\Im}{Im}
\let\Re\relax
\DeclareMathOperator{\Re}{Re}
\newcommand{\CC}{\mathbb{C}}
\newcommand{\RR}{\mathbb{R}}

\newcommand{\HH}{\mathcal{H}}
\newcommand{\KK}{\mathcal{K}}
\newcommand{\MM}{\mathcal{M}}

\newcommand{\Af}{\mathscr{A}}
\newcommand{\Bf}{\mathscr{B}}
\newcommand{\Cf}{\mathscr{C}}
\newcommand{\Tf}{\mathscr{T}}
\newcommand{\Uf}{\mathscr{U}}

\newcommand{\zen}{%
    \mathscr{Z}
}

\DeclareMathOperator{\Bor}{Bor}
\DeclareMathOperator{\Ind}{Ind}
\DeclareMathOperator{\id}{id}

\DeclareMathOperator{\Ad}{Ad}
\DeclareMathOperator{\Aut}{Aut}
\DeclareMathOperator{\Reg}{Reg}
\newcommand{\Alg}{\mathrm{Alg}}
\newcommand{\Loc}{\mathrm{Loc}}

\newcommand{\diff}{\textup{d}}

\numberwithin{equation}{section}

\theoremstyle{thmstyleone}%
\newtheorem{theorem}{Theorem}[section]
\newtheorem{proposition}[theorem]{Proposition}%
\newtheorem{lemma}[theorem]{Lemma}
\newtheorem{corollary}[theorem]{Corollary}

\theoremstyle{thmstyletwo}%
\newtheorem{remark}[theorem]{Remark}%

\theoremstyle{thmstylethree}%
\newtheorem{definition}[theorem]{Definition}%

\begin{document}

\title[Article Title]{Quantum reference frames, measurement schemes and the type of local algebras in quantum field theory}


\author[1,2]{\fnm{Christopher J.} \sur{Fewster} \orcidicon{0000-0001-8915-5321}}\email{chris.fewster@york.ac.uk}

\author[1]{\fnm{Daan W.} \sur{Janssen} \orcidicon{000-0001-7809-5044}}\email{daan.janssen@york.ac.uk}

\author[3]{\fnm{Leon Deryck} \sur{Loveridge} \orcidicon{0000-0002-2650-9214}}\email{leon.d.loveridge@usn.no}

\author*[1,2]{\fnm{Kasia} \sur{Rejzner} \orcidicon{0000-0001-7101-5806}}\email{kasia.rejzner@york.ac.uk}

\author[4]{\fnm{James} \sur{Waldron}}\email{james.waldron@ncl.ac.uk}

\affil[1]{\orgdiv{Department of Mathematics}, \orgname{University of York}, \orgaddress{
\city{Heslington, York}, \postcode{YO10 5DD}, 
\country{United Kingdom}}}

\affil[2]{\orgdiv{York Centre for Quantum Technologies}, \orgname{University of York}, \orgaddress{
\city{Heslington, York}, \postcode{YO10 5DD}, 
\country{United Kingdom}}}

\affil[3]{\orgdiv{Department of Science and Industry Systems}, \orgname{University of South-Eastern Norway}, \orgaddress{
\city{Kongsberg}, \postcode{3616},
 \country{Norway}}}
 
\affil[4]{\orgdiv{School of Mathematics, Statistics and Physics}, \orgname{Newcastle University}, \orgaddress{
\city{Newcastle-Upon-Tyne}, \postcode{NE1 7RU}, 
\country{United Kingdom}}}


\abstract{We develop an operational framework, combining relativistic quantum measurement theory with quantum reference frames (QRFs), in which local measurements of a quantum field on a background with symmetries are performed relative to a QRF. This yields a joint algebra of quantum-field and reference-frame observables that is invariant under the natural action of the group of spacetime isometries. For the appropriate class of quantum reference frames, this algebra is parameterised in terms of crossed products. Provided that the quantum field has good thermal properties (expressed by the existence of a KMS state at some nonzero temperature), one can use modular theory to show that the invariant algebra admits a semifinite trace. If furthermore the quantum reference frame has good thermal behaviour (expressed in terms of the properties of a KMS weight) at the same temperature, this trace is finite. We give precise conditions for the invariant algebra of physical observables to be a type $\textnormal{II}_1$ factor. Our results build upon recent work of Chandrasekaran, Longo, Penington and Witten [JHEP {\bf 2023}, 82 (2023)], providing both a significant mathematical generalisation of these findings and a refined operational understanding of their model.}

\keywords{quantum reference frames, measurement schemes, quantum field theory, von Neumann type reduction}



\maketitle
\newpage 

\section{Introduction}

A remarkable feature of quantum field theory (QFT) as described in the algebraic approach (also known as AQFT)~\cite{Haag:book,HK,Araki81} is that all the von Neumann algebras describing local observables in well-behaved theories turn out to be isomorphic. That is, every QFT from the free scalar field to -- in principle -- the standard model, has local algebras given by the unique hyperfinite factor of type $\textnormal{III}_1$ \cite{fredenhagen1985modular,Haa87} (see also Appendix~\ref{appx:factor_types} for a brief summary of the classification of von Neumann factors). The distinction between different QFTs therefore lies in the variety of possible relations between algebras associated to different spacetime regions.

A recent and radical departure from this paradigm has been proposed in~\cite{chandrasekaran2023algebra}. There, it was argued on physical grounds that, in certain quantum gravitational scenarios, the algebras of physical observables are actually von Neumann algebras of type $\textnormal{II}_1$. The significance of this observation is that type $\textnormal{II}_1$ algebras admit finite trace functionals which can be used to compute physically relevant entropies. By contrast, type $\textnormal{III}_1$ algebras lack any finite traces, and naive attempts to compute entropies result in divergent quantities.

A key component in the arguments of~\cite{chandrasekaran2023algebra}
is the role played by an observer, described by a simple quantum mechanical system, which has the effect of replacing the type $\textnormal{III}_1$ algebra of the QFT by one of type $\textnormal{II}_1$. In this paper, we develop and explore a new perspective on the question of such type reduction, taking seriously the idea that an operational theory of measurement is required to place the discussion on a firm foundation. We start from a recently developed framework for describing measurement in QFT~\cite{FewVer_QFLM:2018,BostelmannFewsterRuep:2020,FewsterJubbRuep:2022,FewsterVerch_encyc:2023}. In this setting, both the system to be measured and the probe that measures it are QFTs, which are coupled together in a compact spacetime region so that the coupled theory is also a QFT. Measurement of the probe in a natural `out' region of the spacetime then provides information about system observables. By comparison, while~\cite{chandrasekaran2023algebra} invokes an observer, it is not specified how the observer is coupled to the system QFT.

From the perspective of quantum gravity (or, for that matter, general relativity), the singling-out of a specific spacetime region in which coupling occurs raises problems, as it appears to break diffeomorphism invariance. In situations where gravity is weakly coupled, one is effectively working with a fixed background spacetime and the relevant symmetry group is the identity-connected component of the spacetime isometry group. For instance~\cite{chandrasekaran2023algebra} study the static patch of de Sitter spacetime, which may be regarded as the maximal experimental  domain for an observer on a complete timelike geodesic in full de Sitter, taking into account the need to both initiate and receive results from experiments. The static patch has an isometry group with identity connected component $\RR\times\textnormal{SO}(n-1)$ describing translations along, and rotations about, the geodesic.

As mentioned, a specific coupling of a system and probe theory will not respect the symmetry group $\RR\times\textnormal{SO}(n-1)$. Instead, the group acts on the family of all measurement couplings, restoring symmetry at that level. But then one is left with the problem of how a particular measurement can be singled out -- indeed we emphasise that the same problem is present in any generally covariant QFT on a fixed background. One can overcome this problem by introducing a reference frame,
which itself can be modelled as a quantum system. 
In this way, the discussion of measurement of QFTs on symmetric background spacetimes naturally leads to the introduction of a quantum reference frame (QRF)
---a concept that has been well studied from a number of angles in non-relativistic quantum theory and quantum information (e.g., \cite{Aharonov:1967zza,Aharonov:1984zz,Bartlett2007,Vanrietvelde:2018pgb,giacomini2019quantum,de2020quantum,Loveridge2017a,carette2023operational,lake2023quantum})---but which requires further development for use in relativistic settings. Our contention is that the `observer' in~\cite{chandrasekaran2023algebra} is a QRF, and we will show that 
the phenomenon of type reduction has a natural description in terms of the theory of QRFs, as given in our present work. 

Put simply, a QRF is a quantum system equipped with an observable that is covariant with respect to a given symmetry group $G$. More technically, the observable is a positive operator-valued measure (POVM) on a measure space, to be called the \emph{value space}. 
Suppose one is interested in measurement of another quantum system admitting the same group of symmetries. One adopts the view that typical observables of the QRF or system have no intrinsic meaning, and that the only physically relevant observables are joint observables of the system and reference frame that are invariant under the joint action of $G$. A simple example in quantum mechanics arises when considering a free particle on the line. Owing to translation invariance, there is no preferred origin and therefore no absolute position observable; however, if a second particle is used as a QRF, their relative position is physically meaningful.

The problem of identifying physically relevant observables is thus reduced to the mathematical problem of finding fixed points under the group action in the tensor product $\MM_S\otimes B(\HH_R)$ where $\MM_S$ is the von Neumann algebra of the QFT observables and $\HH_R$ is the Hilbert space of the QRF. Under certain assumptions on the QRF, we will show (in Theorem~\ref{thm:invariants_crossed_prod}) that the invariant subalgebra is parameterised by the crossed product algebra $(\MM_S\rtimes_\alpha G)\otimes B(\KK)$ where $\alpha$ is the $G$-action on $\MM_S$ and $\KK$ is an auxiliary Hilbert space determined by the QRF. (See Section~\ref{sec:inv_op_crossed} and particularly Definition~\ref{def:crossed_prod} for the definition of a crossed product.) This is a generalisation of the parameterisation employed in~\cite{chandrasekaran2023algebra}, which corresponds to the case $G=\RR$ and $\KK=\CC$. The next task is to analyse the type structure of this cross-product algebra, which we will accomplish for $G=\RR\times H$ with $H$ compact (and further technical assumptions described below) and $\RR$ acting by modular automorphisms on $\MM_S$. Note that this includes the isometry group of the de Sitter static patch, and when $\MM_S$ acts in the GNS representation of a thermal (KMS) state at some inverse temperature $\beta\in (0,\infty)$. As in~\cite{chandrasekaran2023algebra}, the analysis of the type structure relies on deep results from modular theory~\cite{vanDaele:1978, Takesaki1973duality,takesaki1970tomita}.
However, we go beyond~\cite{chandrasekaran2023algebra} because we employ very general QRFs and also consider the larger symmetry groups $\RR\times H$ as well as $\RR$. 

Based on our treatment of the crossed product algebra, our main result (Theorem~\ref{thm:type_change}) furnishes a sufficient condition 
for the crossed product to be a finite von Neumann algebra. This condition can be phrased in terms of a spectral multiplicity function measuring the degrees of freedom available to the QRF at particular energy scales (see Definition~\ref{def:spec_mult}), and has a very natural physical interpretation in terms of the QRF's thermodynamic properties. Namely, a naturally defined KMS weight on the QRF, at the same inverse temperature $\beta$ as the KMS state on $\MM_S$, is finite on a specific dense subalgebra of $B(\HH_R)$. See Theorem~\ref{thm:KMS_weight} for the precise statement. Here, we recall that a weight is, informally, a generalisation of the concept of a state, when we drop the requirement that it be everywhere defined and normalised. Weights arise naturally in constructions of QRFs of the sort needed here (see, for instance \cite{BFH09}). The precise definition of a weight will be recalled in Appendix~\ref{appx:factor_types}; we will also explain why the weight of interest to us is not a state.

Our work requires and extends techniques from AQFT, quantum measurement theory and quantum reference frames. The main developments are summarised at the end of this introduction, after first describing the structure of the paper. In Section~\ref{sec: measurement in QFT}, we begin with a review of measurement theory in QFT based on~\cite{FewVer_QFLM:2018}, providing necessary background on AQFT. In particular, the notion of a measurement scheme is recalled: namely the way in which a probe theory, locally coupled to a system QFT, can be used to measure induced system observables. A particular focus is on the situation where the background spacetime admits a isometry group. 

As mentioned above, the coupling between the system and probe breaks isometry invariance, but we will show that the isometries act in a natural way on the collection of measurement schemes, so that the induced observables transform covariantly (Theorem~\ref{thm:measurement_cov}).
As discussed above, two measurement schemes related by an isometry are not intrinsically distinguishable, and 
we employ a quantum reference frame to identify the true physical observables. The basic properties of QRFs are summarised in Section~\ref{Sec:QRFs}, along with some specific examples. As we recall in Section~\ref{subsec:relativisation1}, previous work on QRFs (e.g., \cite{miyadera2016approximating,Loveridge2017a,glowacki2023quantum}) has provided a specific construction of invariant observables for various classes of QRF through a \emph{relativisation map} which parameterises a subclass of the invariants using the system observables. Later in the present paper, this notion is extended to the setting in which the frame observable's value space is a certain homogeneous space of the symmetry group $G$ rather than $G$ itself. Previously, such an extension has only been constructed for finite groups \cite{glowacki2023quantum}.

Next, in Section~\ref{sec:qrf and measurement}, we combine the theory of QRFs with measurement theory, and single out a preferred class of \emph{compactly stabilised} QRFs that is particularly suitable for the discussion of QFT on the de Sitter static patch. In Theorem~\ref{thm:qref_embed}, we provide a full characterisation of all such QRFs using results of Cattaneo~\cite{Cattaneo:1979}
concerning covariant Naimark dilations and systems of imprimitivity, building on Mackey's theory~\cite{mackey1949imprimitivity,mackey1952induced} as described in Appendix~\ref{appx:imprimitvity_thm}. This has the consequence that all compactly stabilised QRFs may be expressed as compressions of well-understood systems of imprimitivity. Using this characterisation, we then parameterise the invariant observables in the joint algebra of the QFT and the QRF in terms of a crossed product algebra. In turn, this parameterisation is used to provide a relativisation map to the compactly stabilised QRFs (see Definition~\ref{def:relativisation_compact_stab} and Proposition~\ref{prop:yen_properties}). For our analysis of type reduction, we assume that the isometry group $G$ contains a subgroup of time translations; for simplicity we assume it to be of the form $G=\RR\times H$ for $H$ compact. We discuss in Section~\ref{sec:RHcase} how our results on the characterisation of quantum reference frames simplify in this case. In particular, we show that the (energy) spectrum associated with the unitarily implemented action of time translations on compactly stabilised quantum reference frame observables is continuous and that this spectrum admits a Borel-measurable function that can be interpreted as encoding the amount of QRF degrees of freedom available at a given energy, the spectral multiplicity function (see Proposition~\ref{prop:spect_decomp} and Definition~\ref{def:spec_mult}).

Section \ref{sec:type_change} discusses our general results on type reduction. We start by recalling key facts from Tomita-Takesaki theory focussing on results involving crossed products, in particular, that the crossed product of a von Neumann algebra by its modular action (the modular crossed product) is semifinite. Combining this with our results on the parameterisation of invariant observables, we show that this algebra of observables is also semifinite under natural conditions (Theorem~\ref{thm:type_change}). We also provide a condition on the quantum reference frame that guarantees finiteness of the algebra and may be interpreted as a growth condition on the degrees of freedom accessible to the QRF as a function of energy (Equation~\eqref{eq:type_change_condition}). We show that this condition implies good thermal properties of the QRF -- specifically the existence of a KMS weight taking finite values on a dense subalgebra of QRF observables -- at inverse temperature $\beta$ (Theorem~\ref{thm:KMS_weight}). This demonstrates that our type reduction condition for QRFs plays a similar role to nuclearity in QFT \cite{Buchholz:1988fk}.

Lastly, in Section \ref{sec: deSitter}, we discuss the concrete example of the invariant algebra associated with observables in the static patch of de Sitter spacetime. We show that our results agree with and generalise those of \cite{chandrasekaran2023algebra}. We discuss some natural quantum reference frames that can be used to distinguish measurement schemes on the static patch to further illustrate the role of our type reduction condition and its thermal interpretation.

To conclude this introduction, we briefly summarise the main achievements of this paper. First, it brings 
together measurement schemes and QRFs, which seems not to have been done systematically before now, despite their clear mutual relevance (with the exception of \cite{Loveridge2020a}, which is in the context of the Wigner--Araki--Yanase (WAY) theorem \cite{E.Wigner1952,Busch2010,Araki1960,Loveridge2011}). Second, we have identified and characterised the new class of compactly stabilised QRFs, which seems to be of independent utility. Third, we have developed the theory of QRFs, by giving a more detailed understanding of the relationship between the invariant algebra and the relativisation map, and also extending the theory 
so as to be applicable to the case of QFTs on de Sitter and general backgrounds with symmetry.
Fourth, we have given precise conditions sufficient for type reduction from type $\textnormal{III}_1$ to semifinite (e.g., type $\textnormal{II}$) and finite factors (e.g., type $\textnormal{II}_1$) in terms of physical properties of the QRF. Finally, specialising to the de Sitter static patch, we have a significant generalisation of the results obtained in~\cite{chandrasekaran2023algebra}. Further remarks are given in the Conclusion.

\section{Measurement schemes for quantum fields}\label{sec: measurement in QFT}

The discussion of measurement in QFT, or more generally in relativistic quantum information theory, has been a long-standing source of puzzles and apparent paradoxes, leading Sorkin,
for instance, to write that QFT is left `with no definite measurement theory'~\cite{sorkin1993impossible}. Relatively recently, this serious gap has been addressed in \cite{FewVer_QFLM:2018}, where a framework for measurement in QFT, understood in the algebraic formulation (AQFT), has been proposed, leading to a series of developments and new insights -- see e.g.,~\cite{BostelmannFewsterRuep:2020,FewsterJubbRuep:2022}. A recent review of the current state of the art has been provided in \cite{FewsterVerch_encyc:2023}. For completeness, we review the most important features of this framework here, after a brief introduction into AQFT.

\subsection{Background on algebraic quantum field theory}\label{sec:AQFTbackground}

Let $M$ be a fixed globally hyperbolic spacetime (we suppress the metric and time-orientation in the notation). Algebraic quantum field theory (AQFT) describes a QFT on $M$ in terms of its algebras of local observables and their relationships~\cite{Haag:book,AdvAQFT}. In more detail, a \emph{region} of $M$ is any open subset $N$ that is causally convex, i.e., $N$ contains any causal curve whose endpoints lie in $N$, and we denote the set of all regions in $M$ by $\Reg(M)$. For each region $N$ in $M$, one specifies a unital $*$-algebra $\Af(M;N)$ that contains the observables of the theory that are associated with $N$. In the first instance, this interpretation should be regarded informally, and it is the task of a measurement framework to determine what is or is not actually observable.
A typical element of $\Af(M;N)$ might be a quantum field smeared with a test function that vanishes outside $N$. However local algebras also often contain elements that might not be directly observable such as smeared fermion fields or non-gauge-invariant objects. In general any element of $\Af(M;N)$ is said to be \emph{localisable} in $N$, and the algebra of all local observables is $\Af(M)=\Af(M;M)$. One could demand additionally that the algebras $\Af(M;N)$ are $C^*$-algebras or von Neumann algebras; we will not do that at this stage, because it introduces additional technical complications. Later on, we will pass to von Neumann algebras formed in specific Hilbert space representation of $\Af(M)$.  

The local algebras must satisfy a number of relations. First, if $N_1$ and $N_2$ are regions with $N_1\subset N_2$, then $\Af(M;N_1)\subset \Af(M;N_2)$ (this condition is called \emph{isotony}). In particular, $\Af(M;N)\subset \Af(M)$ for all regions $N$; it is assumed that all $\Af(M;N)$'s share a common unit with $\Af(M)$. Furthermore, the \emph{timeslice axiom} requires that if $N_1\subset N_2$ contains a Cauchy surface of $N_2$ then $\Af(M;N_1)=\Af(M;N_2)$. Next, if $N_1$ and $N_2$ are causally disjoint, \emph{Einstein causality} asserts that every element of $\Af(M;N_1)$ commutes with every element of $\Af(M;N_2)$. Finally, we assume the \emph{Haag property}~\cite{FewVer_QFLM:2018} that if $K$ is a compact subset contained in a connected region $L$, and  $A\in \Af(M)$ commutes with every element of $\Af(M;K^\perp)$, then $A\in \Af(M;L)$. Here, $K^\perp=M\setminus (J^+(K)\cup J^-(K))$ is the causal complement of $K$. For brevity the theory as a whole will be denoted $\Af$. In fact, $\Af$ can be regarded as a functor, and
many other objects in this section are either functors or natural transformations. This is spelled out further in Appendix~\ref{appx:functorial}. A final general definition that will be needed is that a \emph{global gauge symmetry} 
of $\Af$ is any automorphism $\zeta\in\Aut(\Af(M))$ that preserves the local algebras, i.e., $\zeta(\Af(M;N))=\Af(M;N)$ for all $N\in\Reg(M)$. Under composition, the global gauge symmetries of $\Af$ form a group, denoted $\Aut(\Af)$. By restriction, each $\zeta\in\Aut(\Af)$ determines automorphisms $\zeta_N\in\Aut(\Af(M;N))$ so that $\zeta_{N_1}$ and
$\zeta_{N_2}$ agree on $\Af(M;N_1)$ whenever $N_1\subset N_2$ are nested regions.
 
As mentioned, our discussion in this section focusses on local algebras given as unital $*$-algebras. However, when we come to consider quantum reference frames, it will be necessary to work with local von Neumann algebras. Let us briefly recall how these may be obtained from $*$-algebras. 
First, a state $\omega$ on $\Af(M)$ is a normalised, positive, linear functional, i.e., $\omega:\Af(M)\to\CC$ is a linear map obeying $\omega(1_{\Af(M)})=1$ and $\omega(A^*A)\ge 0$ for all $A\in\Af(M)$. We furthermore say that $\omega$ is faithful if $\omega(A^*A)=0$ implies $A=0$. The role of a state is to assign expectation values to observables.  
 Given a (faithful) state $\omega:\Af(M)\to \CC$, one can consider the GNS representation $(\HH,\pi,\Omega,\mathscr{D})$ associated with this state, which is in particular an (injective) unital *-homomorphism $\pi$ mapping elements of $\Af(M)$ into (possibly unbounded) operators on a Hilbert space $\HH$ with dense domain $\mathscr{D}$, such that there is a $\Omega\in \mathscr{D}$ with the property that for all $A\in \Af(M)$ 
\begin{equation}
    \omega(A)=\langle \Omega,\pi(A)\Omega\rangle,
\end{equation}
and that $\pi(\Af(M))\Omega$ is dense in $\HH$, see e.g.\ \cite[Ch.~5]{AdvAQFT}. This allows us to define for open causally convex $N\subset M$ a von Neumann algebra given in terms of the double commutant (as a subalgebra of $B(\HH)$)
\begin{equation}
    \mathfrak{A}(M;N)=\pi(\Af(M;N))''.
\end{equation}
If the algebra $\Af(M)$ is a $C^*$-algebra, meaning that it can be faithfully represented as a norm closed algebra of bounded operators on a Hilbert space, one sees that $\Af(M;N)$ can be embedded in the von Neumann algebra $\mathfrak{A}(M;N)$ associated with a particular faithful representation. $C^*$-algebraic models of quantum field theory include the Weyl algebra description of the free field and the dynamical $C^*$-algebra description of interacting field theories \cite{Buchholz_2020}.

For a generic quantum field theory $\Af$, state $\omega$ and sufficiently regular bounded region $N$, the von Neumann algebras $\mathfrak{A}(M;N)$ are expected to be of type III${}_1$, i.e.\ to be decomposable into type III${}_1$ factors, see e.g.\ \cite{fredenhagen1985modular,Buchholz1987Universal,Verch:1996wv,YNGVASON2005135}.  Some aspects 
 of the type classification of von Neumann algebras are recalled in Appendix \ref{appx:factor_types}. 

\subsection{Review of the measurement framework}

Turning to measurement, the viewpoint adopted in~\cite{FewVer_QFLM:2018} is that measurements on one QFT (the ``system'') are made through interactions with another QFT (the ``probe''). Denoting the system QFT by $\Af$ and the probe by $\Bf$, the two theories can be combined by taking a tensor product $\Uf=\Af\otimes\Bf$ so that the local algebras are $\Uf(M;N)=\Af(M;N)\otimes\Bf(M;N)$. The tensor product theory $\Uf$ contains $\Af$ and $\Bf$ as independent subsystems, and is called their \emph{uncoupled combination}.
A measurement interaction is modelled by a \emph{coupled combination} $\Cf$, which is another QFT that agrees with $\Uf$ outside a compact \emph{coupling zone} $K\subset M$. 
Specifically, one assumes that there are $*$-algebra isomorphisms $\chi_N:\Uf(M;N)\to\Cf(M;N)$ labelled by regions $N$ outside the causal hull $J^+(K)\cap J^-(K)$ of $K$ so that
for all such regions $N_1$ and $N_2$ with $N_1\subset N_2$, the maps $\chi_{N_2}$ and $\chi_{N_1}$ agree on $\Uf(M;N_1)$. The precise nature of the interaction is unspecified, beyond the assumption that $\Cf$ is itself a QFT obeying the conditions set out above. We will denote the system of isomorphisms as a whole by $\chi$, and denote the coupled combination formally by the pair $(\Cf,\chi)$.

As $K$ is compact and $M$ is globally hyperbolic, the sets $M^\pm=M\setminus J^\mp(K)$ are regions and there are isomorphisms $\chi_{M^\pm}:\Uf(M;M^\pm)\to \Cf(M;M^\pm)$. As $M^\pm$ contain Cauchy surfaces of $M$, $\chi_{M^\pm}$ induce isomorphisms $\tau^\pm:\Uf(M)\to\Cf(M)$ called the retarded ($+$) and advanced ($-$) response maps, and $\Theta=(\tau^-)^{-1}\circ\tau^+$ is an automorphism of $\Uf(M)$ called the \emph{scattering map}.

Now suppose that the system and probe are prepared in algebraic states $\omega$ and $\sigma$ at early times, corresponding to the state $\omega\otimes\sigma$ on $\Uf(M)$ or $((\tau^-)^{-1})^*(\omega\otimes\sigma)$ on $\Cf(M)$ -- here, the use of $\tau^-$ reflects the `early time' nature of the preparation. Suppose also that a measurement is made, at late times, of a probe observable $B\in\Bf(M)$, corresponding to $1\otimes B\in \Uf(M)$ or $\tau^+(1\otimes B)\in\Cf(M)$. Here, the use of $\tau^+$ reflects the `late time' nature of the measurement. The expectation value of the experiment that measures $B$ in the given state is 
\begin{equation}
    (((\tau^-)^{-1})^*(\omega\otimes\sigma))(\tau^+(1\otimes B)) = 
    (\omega\otimes\sigma)(\Theta 1\otimes B).
\end{equation}
We can describe the measurement at the level of the system alone, as the measurement of an observable $A\in\Af(M)$ in state $\omega$, where $A$ is chosen so that its expectation value $\omega(A)$ matches the expectation value of the actual experiment on the coupled system, i.e., $\omega(A)=(\omega\otimes\sigma)(\Theta 1\otimes B)$. The observable $A$ is given explicitly by 
\begin{equation}\label{eq:induct}
    A = \varepsilon_\sigma(B):=\eta_\sigma (\Theta (1\otimes B)),
\end{equation}
where the linear map $\eta_\sigma:\Uf(M)\to\Af(M)$ traces out the probe, 
so that $\eta_\sigma (A\otimes B) =\sigma(B)A$. We refer to 
$\varepsilon_\sigma(B)$ as the \emph{induced system observable} corresponding to probe observable $B$. Taken together, $\Bf$, $\Cf$, $\chi$, $\sigma$, and $B$ provide a \emph{measurement scheme} for $\varepsilon_\sigma(B)$, so that one has $\omega(\varepsilon_\sigma(B))=(\omega\otimes\sigma)(\Theta 1\otimes B)$. 

Although it will not be needed here, one may derive rules for updating states after selective or nonselective measurements~\cite{FewVer_QFLM:2018} which are fully compatible with relativity and avoid pathologies such as those described in~\cite{sorkin1993impossible} -- see~\cite{BostelmannFewsterRuep:2020} and \cite{FewsterVerch_encyc:2023} for further discussion. Furthermore, at least in some models, it has been shown that every local observable can be measured at least asymptotically by a sequence of measurement schemes~\cite{FewsterJubbRuep:2022}.

Interaction terms in QFT are usually required to be gauge-invariant, and a similar condition can be imposed on the coupled combination of two theories. In detail, suppose that there are homomorphisms $\varphi_\Af:H\to\Aut(\Af)$ and $\varphi_\Bf:H\to\Aut(\Bf)$ for some group $H$. Then there is a homomorphism
$\varphi_\Uf:H\to\Aut(\Uf)$ given by $\varphi_\Uf(h)=\varphi_\Af(h)\otimes\varphi_\Bf(h)$. Given a coupled combination $(\Cf,\chi)$ of $\Af$ and $\Bf$, we will describe the coupling as gauge-invariant if there is also a homomorphism $\varphi_\Cf:H\to \Aut(\Cf)$ with 
\begin{equation}\label{eq:CfUf}
    \varphi_\Cf(h)_N\circ\chi_N=\chi_N\circ\varphi_\Uf(h)_N
\end{equation} 
for all $h\in H$ and regions $N$ outside the causal hull of $K$,
where $\varphi_\Uf(h)_N:\Uf(M;N)\to\Uf(M;N)$ is the isomorphism
obtained by restriction of $\varphi_\Uf(h)$, and $\varphi_\Cf(h)_N$ is defined similarly. 
Under these circumstances the induced observable maps transform equivariantly, as shown in Theorem~3.7 of~\cite{FewVer_QFLM:2018}, namely one has
\begin{equation}
    \varphi_\Af(h)(\varepsilon_\sigma(B)) = \varepsilon_{\varphi_\Bf(h^{-1})^*\sigma} (\varphi_\Bf(h)(B))
\end{equation}
for all $h\in H$, $B\in\Bf(M)$ and probe preparation states $\sigma$. Consequently, if the probe preparation state is gauge-invariant, i.e., $\varphi_\Bf(h^{-1})^*\sigma=\sigma$ for all $h$, then gauge-invariant probe observables induce gauge-invariant system observables. 

\subsection{Measurement on backgrounds with symmetry} \label{sec: measurement with symm}

Now suppose that the background spacetime $M$ admits a group $G$ of isometries that preserve orientation and time-orientation,  
with identity element $1_G$. (The restriction to (time)-orientation preserving isometries is partly made for convenience, and for direct compatibility with the framework of locally covariant QFT~\cite{BrFrVe03}. However, we note that measurement schemes naturally break time-reversal symmetry, because one prepares in the past and measures in the future; it also seems reasonable that an observer has a fixed notion of orientation.) An element $g\in G$ acts on $M$ and on the set of regions $\Reg(M)$, with actions denoted $(g,x)\mapsto g.x$ and $(g,N)\mapsto g.N = \{g.x:x\in N\}$ for $g\in G$, $x\in M$ and $N\in\Reg(M)$. 

Given a QFT $\Af$, each $g\in G$ determines a modified theory $\act{g}{\Af}$
with local algebras $\act{g}{\Af}(M;N)=\Af(M;g^{-1}.N)$, so that $G$ acts from the left:  
\begin{equation}
\act{g}{(\act{h}{\Af})}(M;N) = \act{h}{\Af}(M;g^{-1}.N)=\Af(M; h^{-1}g^{-1}.N)=\act{gh}{\Af}(M;N).
\end{equation}
It is straightforward to show that $\act{g}{\Af}$ inherits isotony, the timeslice property, Einstein causality and Haag property from $\Af$, and therefore defines a QFT in its own right.
We will say that $\Af$ is \emph{$G$-covariant} if there are automorphisms $\alpha(g)$ of $\Af(M)$ labelled by $g\in G$ such that $\alpha(g) (\act{g}{\Af}(M;N)) =  \Af(M;N)$ for all $N\in\Reg(M)$ and $\alpha(1_G)=\id_{\Af(M)}$, and, if so, the $G$-covariance is said to be implemented by $\alpha$. It is natural to assume that (a) $\alpha$ is a group homomorphism of $G$ to $\Aut(\Af(M))$ and (b) in the spirit of the Coleman--Mandula theorem (\cite{ColeMand:1967,Fewster:2018}), the action of the spacetime symmetries should commute with that of
the global gauge symmetries in $\Aut(\Af)$. Thus we require $\zeta\circ\alpha(g)=\alpha(g)\circ \zeta$ for
all $\zeta\in\Aut(\Af)$, $g\in G$. Restricting to local regions, this gives the identity $ \zeta_N \circ\alpha(g)_N = \alpha(g)_{N} \circ\zeta_{g^{-1}.N}$, or equivalently
\begin{equation}\label{eq:local_b}
    \zeta_N = \alpha(g)_{N} \circ\zeta_{g^{-1}.N}\circ\alpha(g)_N^{-1} ,
\end{equation}
where $\alpha(g)_N:\act{g}{\Af}(M;N)\to \Af(M;N)$ is defined by restriction of $\alpha(g)$ to $\act{g}{\Af}(M;N)=\Af(M;g^{-1}.N)$. The notion of $G$-covariance can be expressed functorially, as explained in Appendix~\ref{appx:functorial}, where the motivation for assumptions (a) and (b) is explained in greater detail.

Returning to the discussion of measurements, let us now suppose that both the system and probe theories are $G$-covariant, with implementations $\alpha$ and $\beta$ obeying both assumptions (a) and (b).
Then $\Uf$ is clearly $G$-covariant, with implementation $\upsilon(g) = \alpha(g)\otimes\beta(g)$ that obeys assumption (a), and a partial form of assumption (b). Namely, $\upsilon(g)$ commutes with
gauge symmetries of $\Uf$ that are of the form $\zeta\otimes\xi$, for
$\zeta\in\Aut(\Af)$, $\xi\in\Aut(\Bf)$. It is possible that $\Uf$ may
have other gauge symmetries not of this form,\footnote{For instance, if $\Bf=\Af$ then $\Uf$ has an additional `flip' symmetry $A\otimes B\mapsto B\otimes A$.} and so condition~(b) is not necessarily guaranteed to hold in full. 

As described above, the measurement framework defines a coupled combination 
in terms of a theory $\Cf$ that agrees with $\Uf$ outside a compact coupling zone $K\subset M$.
Assuming that the dynamics of $\Cf$ differs nontrivially from that of $\Uf$, the theory $\Cf$ 
will not respect $G$-covariance. In fact, if $\chi$ is the system of isomorphisms making $\Cf$
a coupled combination of $\Af$ and $\Bf$, then $\act{g}{\chi}$ makes $\act{g}{\Cf}$ another coupled combination, where $\act{g}{\chi}$ is the system of isomorphisms
\begin{equation}\label{eq:gchi}
    (\act{g}{\chi})_N = \chi_{g^{-1}.N} \circ \upsilon(g)_N^{-1}  : \Uf(M;N)\to \Cf(M;g^{-1}.N)=\act{g}{\Cf}(M;N)
\end{equation}
defined for all $N\in\Reg(M)$ such that $g^{-1}.N$ lies outside the causal hull of $K$. See Figure~\ref{fig:cd1} for a commutative diagram illustrating this definition and equations~\eqref{eq:gtaupm} and~\eqref{eq:gTheta} below.
\begin{figure}
    \centering 
    \begin{tikzcd}
    \Uf(M;N) \arrow[d,"\upsilon(g)^{-1}_N"']\arrow[r,"(\act{g}{\chi})_N"] &\act{g}{\Cf}(M;N)\arrow[equal,d] \\
    \Uf(M;g^{-1}.N) \arrow[r," \chi_{g^{-1}.N}"] &\Cf(M;g^{-1}.N)
    \end{tikzcd}\hfil
    \begin{tikzcd}
    \Uf(M) \arrow[d,"\upsilon(g)^{-1}"']\arrow[r,"\act{g}{\tau}^+"]
    \arrow[rr,bend left=40,"\act{g}{\Theta}"] &\act{g}{\Cf}(M)\arrow[equal,d] \arrow[r, "(\act{g}{\tau}^-)^{-1}"] & \Uf(M) \\
    \Uf(M) \arrow[r," \tau^+"]\arrow[rr,bend right=40,"\Theta"] &\Cf(M) \arrow[r, "(\tau^-)^{-1}"] & \Uf(M)\arrow[u,"\upsilon(g)"']
    \end{tikzcd}
    \caption{Commuting diagrams illustrating equations~\eqref{eq:gchi},~\eqref{eq:gtaupm} and~\eqref{eq:gTheta}.}
    \label{fig:cd1}
\end{figure}
Here, $\upsilon(g)_N:\Uf(M;g^{-1}.N)\to \Uf(M;N)$ denotes the isomorphism obtained by restriction of $\upsilon(g)$. Thus $\act{g}{\Cf}$ has coupling zone $g.K$, 
`in' and `out' regions $g.M^\pm$, and response maps induced by 
$(\act{g}{\chi})_{g.M^\pm}=\chi_{M^\pm}\circ \upsilon(g)^{-1}_{g.M^\pm}$  that satisfy
\begin{equation}\label{eq:gtaupm}
    \act{g}{\tau}^\pm = \tau^\pm \circ \upsilon(g)^{-1}. 
\end{equation}
Consequently, the scattering map of $\act{g}{\Cf}$ with respect to $\Uf$ is
\begin{equation}\label{eq:gTheta}
    \act{g}{\Theta} = \upsilon(g)\circ \Theta\circ \upsilon(g)^{-1}.
\end{equation}

Meanwhile, if $\Cf$ represents a gauge-invariant coupling of $\Af$ and $\Bf$, then -- as we now describe -- the same is true of $\act{g}{\Cf}$, due to condition (b) above. As in the previous subsection, suppose that there are
homomorphisms $\varphi_\Af:H\to\Aut(\Af)$, $\varphi_\Bf:H\to\Aut(\Bf)$, and
$\varphi_\Cf:H\to\Aut(\Cf)$ obeying~\eqref{eq:CfUf} for some group $H$, where $\varphi_\Uf(h)=\varphi_\Af(h)\otimes \varphi_\Bf(h)$. As the gauge symmetries $\varphi_\Uf(h)$ all take the product form, they commute with the implementation of the $G$-covariance $\upsilon(g)$ due to assumption (b) on both $\Af$ and $\Bf$ -- using~\eqref{eq:local_b} one sees that
\begin{equation}
    \varphi_\Uf(h)_N=\upsilon(g)_{N}\circ \varphi_{\Uf}(h)_{g^{-1}.N}\circ \upsilon(g)^{-1}_N
\end{equation} 
holds for $g\in G$, $h\in H$ and $N\in\Reg(M)$. This identity can be used to show that $\act{g}{\Cf}$ is a gauge-invariant coupling of $\Af$ and $\Bf$, with respect to the homomorphism  $\varphi_{\act{g}{\Cf}}:H\to\Aut(\act{g}{\Cf})=\Aut(\Cf)$ given by 
$\varphi_{\act{g}{\Cf}}(h)  = \varphi_\Cf(h)$. To see this, one notes that
$\varphi_{\act{g}{\Cf}}(h)_N = \varphi_\Cf(h)_{g^{-1}.N}$ and consequently
\begin{align}
    \varphi_{\act{g}{\Cf}}(h)_N\circ\act{g}{\chi}_N 
    &=\varphi_{\Cf}(h)_{g^{-1}.N}\circ\chi_{g^{-1}.N} \circ \upsilon(g)^{-1}_N \nonumber \\
   &= 
    \chi_{g^{-1}.N} \circ\varphi_{\Uf}(h)_{g^{-1}.N}\circ \upsilon(g)^{-1}_N  
   \nonumber \\
   &= \act{g}{\chi}_N\circ \upsilon(g)_{N}\circ \varphi_{\Uf}(h)_{g^{-1}.N}\circ \upsilon(g)^{-1}_N  
    \nonumber \\
    &=\act{g}{\chi}_N\circ\varphi_\Uf(h)_N
\end{align}
holds for all $N\in\Reg(M)$ outside the causal hull of $g.K$, and all $g\in G$, $h\in H$, which is the required property. 

The induced observable map for $\act{g}{\Cf}$ is easily computed, noting that
the partial trace maps $\eta_\sigma:\Uf(M)\to\Af(M)$ satisfy
the identity 
\begin{equation}
    \eta_\sigma\circ\upsilon(g) = \alpha(g)\circ \eta_{\beta(g)^*\sigma}
\end{equation}
as is seen from the calculation
\begin{equation}
    \eta_\sigma(\upsilon(g) (A\otimes B)) = \eta_\sigma(\alpha(g)A\otimes\beta(g)B) = \sigma(\beta(g)B)\alpha(g)A = \alpha(g)\eta_{\beta(g)^*\sigma} (A\otimes B)
\end{equation}
for all $A\in\Af(M)$ and $B\in\Bf(M)$. Consequently, the induced observable map $\act{g}{\varepsilon}_\sigma$ for measurements using the coupled combination $\act{g}{\Cf}$ is
\begin{align}
\act{g}{\varepsilon}_\sigma (\beta(g)B) &= 
\eta_\sigma(\act{g}{\Theta} 1\otimes \beta(g)B) = 
\eta_\sigma(\upsilon(g) \Theta 1\otimes B) = 
\alpha(g) \eta_{\beta(g)^*\sigma}(\Theta 1\otimes B) \nonumber \\&= 
\alpha(g) \varepsilon_{\beta(g)^*\sigma}(B).
\end{align}
for all $B\in\Bf(M)$ -- a formula that was stated in~\cite{FewVer_QFLM:2018} but left to the reader.
We see immediately that if $A=\varepsilon_\sigma(B)$
is the system observable measured using $\Bf$, $\Cf$, $\chi$, $\sigma$ and $B$, then the transformed measurement scheme $\Bf$, $\act{g}{\Cf}$, $\act{g}{\chi}$, $\act{g}{\sigma}:=\beta(g^{-1})^*\sigma$ and $\beta(g)B$ measures $\alpha(g)A$. 
Note that $(g,\sigma)\mapsto \act{g}{\sigma}$ defines a left action of $G$ on the states of $\Bf(M)$, so that $\act{g}{\sigma}(\beta(g)B)=\sigma(B)$. At first sight, it might appear strange that $\Bf$, rather than $\act{g}{\Bf}$ should appear in the transformed measurement scheme, but the point is that we wish to translate the interaction region in the coupled theory while keeping the uncoupled combination $\Uf$ unchanged.

We may summarise the above as follows.
\begin{theorem} 
\label{thm:measurement_cov}
Let $G$ be a group of time-orientation preserving isometries of the background spacetime $M$.
Suppose $\Af$ and $\Bf$ are $G$-covariant QFTs with implementations $\alpha$ and $\beta$
obeying assumptions (a) and (b), and suppose that $(\Cf,\chi)$ is a (gauge-invariant) coupled combination of $\Af$ and $\Bf$ with coupling zone $K$ and scattering map $\Theta$. Then, for all $g\in G$,  $(\act{g}{\Cf},\act{g}{\chi})$ is a (gauge-invariant) coupled combination of $\Af$ and $\Bf$ with coupling zone $g.K$ and scattering map $\act{g}{\Theta} = \upsilon(g) \circ\Theta \circ\upsilon(g)^{-1}$, where $\upsilon(g) = \alpha(g)\otimes\beta(g)$. Furthermore, if $(\Bf,\Cf,K,\chi,\sigma, B)$ is a measurement scheme for 
$A\in\Af(M)$, then $(\Bf, \act{g}{\Cf},g.K, \act{g}{\chi}, \act{g}{\sigma},\beta(g)B)$ is a measurement scheme for $\alpha(g)A$. 
\end{theorem}

The conclusion to be drawn from this discussion is that, while
any individual measurement scheme necessarily breaks $G$-covariance by introducing a specific coupling zone, the action of the symmetry group is operationally meaningful, when considering the transformation of the measurement scheme as a whole. 

From the perspective of general relativity, individual points in manifolds have no physical meaning, and therefore the specification of a coupling zone as a subset of $M$ is likewise physically meaningless. As we consider a background spacetime with a given metric, we need only consider transformations by isometries rather than arbitrary diffeomorphisms. A more detailed discussion on this viewpoint will be given in \cite{type_reductionII}. 
By this reasoning, the coupled theory
$\Cf$ can only be distinguished from $\act{g}{\Cf}$ by specifying a physical system to act as a reference frame that itself transforms under an action of $G$. Taking the further step that the reference frame is a quantum system, we may conclude that
(a) measurement of quantum fields in spacetimes with symmetry requires the specification of a quantum reference frame, with a $G$-covariant observable, whereupon (b) physical observables should be elements of the combined observable algebra of the system and reference frames that are invariant under the combined actions of $G$. To express these ideas precisely we now review and extend the theory of quantum reference frames.

\section{Quantum reference frames}\label{Sec:QRFs}

Quantum reference frames have been studied since the 1960's \cite{Aharonov:1967zza}, were refined in the 80's \cite{Aharonov:1984zz}, and the subject has developed in a number of directions in the intervening decades. Recently, ideas surrounding QRFs have found new domains of application, particularly in (quantum) gravitational and spacetime settings (e.g., \cite{goeller2022diffeomorphism,carrozza2022edge,kabel2024identification,de2023quantum1,de2022quantum}). Various distinct QRF frameworks have now been established (e.g., \cite{Bartlett2007,Vanrietvelde:2018pgb,giacomini2019quantum}---the \emph{quantum information, perspective neutral} and \emph{purely perspectival} approaches); here we follow and develop further the \emph{operational} framework described in e.g. \cite{Loveridge2017a,carette2023operational}. The understanding of quantum reference frames promulgated in the operational approach (and elsewhere) is that describing quantum systems in a manner that respects some given symmetry requires that a second, `reference' system/frame be included explicitly as part of the description, and that the symmetry is imposed by demanding invariance of observable quantities, as motivated in the previous section. In addition, the operational approach utilises a notion of relative, or relativised, observables, which will be discussed shortly.
 
 The algebra of truly observable quantities in the operational approach is given as the commutant of the symmetry group, which differs from e.g. the perspective-neutral approach \cite{Vanrietvelde:2018pgb,de2021perspective,vanrietvelde2023switching,ahmad2022quantum,Hoehn:2019owq} in which the observables are given on a Hilbert space described by the kernel of a quantised constraint. The demand that what is observable must be invariant arises also as a theorem in the quantum theory of measurement -- if a system has a conserved quantity then all observables of that system must be invariant under the symmetry it generates \cite{Loveridge2017a}, constituting a `strong' version of the famous Wigner-Araki-Yanase theorem \cite{E.Wigner1952,Busch2010,Araki1960,Loveridge2011}.

Before giving a definition of quantum reference frames suitable for the rest of this work, we briefly introduce some elements of operational quantum mechanics on which many of the ideas are based (see e.g. \cite{busch1997operational, Busch2016a}). Here the notion of observable is relaxed from self-adjoint operator to positive operator valued measure (POVM), which is motivated through the utilisation of the full probabilistic structure of the Hilbert space framework. Precisely, a normalised POVM is a mapping
 $F:\mathcal{F} \to B(\mathcal{H})$, where $\mathcal{H}$ is a Hilbert space and $\mathcal{F}$ is a $\sigma$-algebra of subsets of the \emph{value space} $\Sigma$ of $F$; for each $X \in \mathcal{F}$, $F(X)$ is positive, $F(\Sigma)=1_\HH$, and $F$ is $\sigma$-additive on disjoint elements of $\mathcal{F}$, i.e., $F(\bigcup_{i\in I}X_i) = \sum_{i\in I}F(X_i)$ for $I$ countable and $X_i\cap X_j=\emptyset$ when $i\neq j$, where in the case of an infinite sum the convergence is understood in the weak operator topology. All POVMs will henceforth satisfy the normalisation $F(\Sigma)=1_\HH$, and therefore we will omit explicit statement of this fact. For any normal state $\rho$, identified here with a positive, trace class operator with unit trace,
 the mapping $X \mapsto \mathrm{tr}[\rho F(X)]$ is a probability measure on $\Sigma$. In the case that $F(X)$ is a projection for each $X \in \mathcal{F}$, which is equivalent to $F(X \cap Y) = F(X)F(Y)$ for any pair $X$ and $Y$ in $\mathcal{F}$, $F$ is called a projection valued measure (PVM).
 If $\mathcal{F}=\Bor(\mathbb{R})$, a projection valued measure is often called a spectral measure, which corresponds to a (generally unbounded) self-adjoint operator through the spectral theorem, from which the more standard text-book description of observables is recovered. In the literature on operational quantum mechanics, POVMs which are not PVMs are often called `unsharp'. Finally, operators in the unit operator interval $[0,1_{\HH}]\subset B(\HH)$ are called \emph{effects}, which can be alternatively described as POVMs on the space $\{0,1\}$, i.e., the effects are binary observables, whose expectation value gives the probability of a binary measurement (`yes'/`no') resulting in success. In particular, if $E$ is a POVM with value space $\Sigma$ then each $E(X)$ with $X\in\mathcal{F}$ is an effect, representing a test that the value of the measured observable lies in $X\subset\Sigma$. In what follows, the value space $\Sigma$ will typically be chosen to be a topological space and $\mathcal{F}=\Bor(\Sigma)$, its Borel $\sigma$-algebra.
 
Turning back to quantum reference frames, we recall that two requirements must be satisfied - the frame system must possess a $G$-covariant POVM, and observables of the combined system and frame must be $G$-invariant. We return to the invariance in Subsec. \ref{subsec:relativisation1}, and first delineate the covariance requirement:

\begin{definition} \label{def:qrf1} Consider a locally compact topological space $\Sigma$ with Borel $\sigma$-algebra $\Bor(\Sigma)$, and a complex, separable Hilbert space $\mathcal{H}$. Let $G$ be a locally compact, second countable Hausdorff group $G$ for which $\Sigma$ is a homogeneous left $G$-space with continuous left action written $(g,x) \mapsto g.x$, and consider a strongly continuous unitary representation $U$ of $G$ in $\mathcal{H}$. A positive operator valued measure $E:\Bor(\Sigma)\to B(\mathcal{H})$ is called covariant if for each $g \in G$ and $X \in \Bor(\Sigma)$, $E(g.X) = U(g)E(X)U(g)^*$. The triple $(U,E,\mathcal{H})$ is called a
quantum reference frame.
 \end{definition}

Quantum reference frames are therefore \emph{systems of covariance}, which generalise (by allowing the operator valued measures to be unsharp) the notion of systems of imprimitivity familiar from, inter alia, the representation theory of groups (see e.g.~\cite{mackey1949imprimitivity}). We will often write $(U_R,E,\mathcal{H}_R)$ to emphasise that the given system of covariance is understood as a QRF. We note that in general $\Sigma\cong G/H$ for some closed $H\subset G$, where as a $G$-space isomorphism the left action on $G/H$ is left multiplication.

 In keeping with standard usage in the operational framework of quantum mechanics described above, in which observables are understood as POVMs, $E$ is therefore understood as the frame observable. The properties and nomenclature of different POVMs $E$, and their value spaces, are reflected in those of the reference frame (see also \cite{carette2023operational}).
 
 \begin{definition}\label{def:QRF types} A quantum reference frame $(U_R,E,\mathcal{H_R})$ is called 
 \begin{itemize}
     \item Sharp if $E$ is a PVM (and unsharp otherwise)
     \item Principal if $\Sigma$ is a principal homogeneous space
     \item Ideal if it is both principal and sharp
     \item Compactly stabilised if $\Sigma \cong G/H$ with $H \subset G$ compact
     \item Localisable if $E$ satisfies the norm-$1$ property, that is, for each $X \in \Bor(\Sigma)$ for which $E(X)\neq 0$, it holds that $||E(X)||=1$
     \item Complete if there is no (non-trivial) subgroup $H_0 \subset G$ acting trivially on all the effects of $E$ (and incomplete otherwise)
 \end{itemize}
 \end{definition}

     Compactly stabilised quantum reference frames are particularly tractable and sufficiently general for many physical applications; indeed, they arise naturally for measurement schemes with compact coupling regions in the static patch of de Sitter; see
      Sec. \ref{sec:qrf and measurement} for further motivation and commentary. A similar compactness condition arises at the level of the representation acting on the frame effects; compact $H_0$ plays an important role in \cite{de2021perspective} for POVMs defined by systems of coherent states. The norm-$1$ property can be stated in two equivalent forms: (i) $E$ satisfies the norm-$1$ property if and only if for each non-zero $E(X)$, there is a sequence of unit vectors $(\varphi_n)\subset \mathcal{H}$ for which 
     $\lim_{n \to \infty}\langle \varphi_n, E(X) \varphi_n \rangle = 1$, and (ii) for any non-zero $E(X)$ and any $\epsilon > 0$, there is a unit vector $\varphi \in  \mathcal{H}$ for which 
     $\langle \varphi, E(X) \varphi  \rangle > 1-\epsilon$. Both (i) and (ii) capture the idea that POVMs with the norm-$1$ property admit probability distributions which arbitrarily well concentrated inside any set which may be obtained in a measurement of $E$ - a property which is satisfied by all PVMs, but by no means all POVMs (for example, covariant POVMs in low dimensional Hilbert spaces typically do not have this property \cite{Loveridge2017a,carette2023operational}). Finally, for any normal state $\omega_R$ on $B(\HH_R)$, $\Bor(\Sigma) \ni X \mapsto \omega_R(E(X))$ is a probability measure, which can be interpreted as a probabilistic value specification of the orientation/configuration of the frame, and picks out, to arbitrary precision, a particular $s \in \Sigma$ when the probability defined through $\omega_R$ and $E$ is highly localised around such a $s$.

 If the algebra of observables of a quantum system is given by the von Neumann algebra $\mathcal{M}_S$ with $\sigma$-weakly continuous action $\alpha_S$ of $G$ on $\mathcal{M}_S$, the invariant algebra $(\mathcal{M}_S \otimes B(\mathcal{H}_R))^{\alpha_S \otimes \Ad U_R}$ represents the properties/observables that are defined independently of an external frame \cite{Loveridge2017a,carette2023operational}. In some cases, discussed below, certain elements of the invariant algebra can be understood as \emph{relative observables}, which can be obtained through a `relativisation' procedure. For example, if $\mathcal{M}_S$ and $B(\mathcal{H}_R)$ are both given as $B(L^2(\mathbb{R}))$, and $G = \mathbb{R}$ acts by translations on both factors, the (spectral measure of the) self-adjoint position operator $Q_S$ relativises to the (spectral measure of the) self-adjoint relative position $Q_S \otimes \id  - \id \otimes Q_R$, which is translation-invariant. We return to this, giving the relativisation map and some examples in Subsec. \ref{subsec:relativisation1}.

\subsection{Examples of quantum reference frames}\label{sec:qrf examples}

\paragraph{Example 1.} Fix $G$ and $\Sigma$ as in Def. \ref{def:qrf1}, equipped with $\sigma$-finite invariant measure $\mu$, and write $(\Sigma,\mu)$ for the measure space. Set $\mathcal{H}_R = L^2(\Sigma, \mu)$, $(U_R(g)f)(x):= f(g^{-1}.x)$, and $(E(X)f)(x):=\chi_X(x)f(x)$. Then $(U_R,E,\mathcal{H_R})$ a sharp frame, and is principal if $\Sigma$ can be identified with $G$, and in this case therefore also ideal.

\paragraph{Example 2.} Let $(\Sigma_1, \mathcal{F}_1)$ and $(\Sigma_2, \mathcal{F}_2)$ be measurable spaces and $p: \mathcal{F}_2 \times \Sigma_1  \to [0,1]$ a Markov kernel, i.e., a mapping for which (i) for each $X\in \mathcal{F}_2$,  $p(X,\cdot):\Sigma_1 \to [0,1]$ is a measurable function, and (ii) for each $x \in \Sigma_1$, $p(\cdot, x):\mathcal{F}_2 \to [0,1]$ is a probability measure. If $E_1:\mathcal{F}_1 \to B(\mathcal{H})$ is a POVM, then 
\begin{equation}\label{eq:markov}
    X \mapsto E(X):=\int_{\Sigma_1}p(X,x)dE_1(x)
\end{equation}
is also a POVM in $B(\mathcal{H})$, understood as a `smearing' of $E_1$, written $E=p * E_1$. If $E_1$ is a PVM, $E$ is commutative, and  conversely, any commutative POVM can be obtained as a smearing of a PVM realised through Eq. \eqref{eq:markov} \cite{holevo2003statistical}. Setting $B(\mathcal{H}_R) \equiv B(\mathcal{H}) = L^2(\mathbb{R})$, $E_1$ the spectral measure of the sharp `position' (i.e., multiplication by $\chi$) covariant under shifts on $\mathbb{R}$ given through the standard $U_R$, and $\mu$ any positive, normalised Borel measure on $\mathbb{R}$, we may fix the Markov kernel $p(X,y)=\mu(X-y)$. With $E$ defined as above, $(U_R,E,L^2(\mathbb{R}))$ yields a system of covariance, defining an unsharp (or smeared) position observable, with the unsharpness quantified by the spread of $\mu$. Smeared momentum observables can be generated directly from the analogous construction, or via the Fourier-Plancherel transform on $L^2(\mathbb{R})$. Indeed, setting $M$ as the canonical spectral measure on $L^2(\mathbb{R})$ (defined by $M(X)f = \chi_X f$ for $X \in \Bor(\mathbb{R})$), and $E:\Bor(\mathbb{R})\to L^2(\mathbb{R})$ any POVM, if $E$ is covariant under shifts (conjugation by $U$ defined by $(U(q)f)(x) = f(x-q)$) and invariant under `boosts' (defined by $(V(p)f)(x)=e^{ipx}f(x)$), then $E=\mu * M$ for some positive normalised Borel measure $\mu$ \cite{carmeli2004position}. The analogue holds for momentum observables, which are invariant under shifts and covariant under boosts.

\paragraph{Example 3.}
For physical reasons it is often demanded that the (self-adjoint) Hamiltonian operator $H$ generating the time evolution $U(t)=\exp(iHt)$ is bounded from below. In this setting, the famous Pauli theorem (e.g. \cite{pauli2012general}) dictates that there is no spectral measure $P:\Bor(\mathbb{R}) \to B(\mathcal{H})$ for which $e^{itH}P(X)e^{-itH}=P(X+t)$, thereby ruling out the existence of a self-adjoint time observable conjugate to energy.
(Throughout this paper, the POVM will transform in the same fashion as Heisenberg-picture quantum fields, which is why we use $U(t)=e^{iHt}$.)
Nevertheless, there are many examples of POVMs satisfying the covariance requirement (e.g. \cite{busch1994time,PhysRevA.66.044101,Busch2016a}), understood as unsharp (non-projection valued) time observables. An example arises in \cite{chandrasekaran2023algebra}, in which the `observer' Hamiltonian is given as the position operator $q_+$ in $L^2(\mathbb{R}_+)$ (i.e., multiplication by the argument, in a suitable dense domain). Here, one should think of  $L^2(\mathbb{R}_+)$ as the spectral representation of the observer Hamiltonian, with $q_+$ having dimensions of energy -- it is not necessary to assume that $q_+$ represents a physical position.
It is then possible to construct a covariant time observable $E^T: \Bor(\mathbb{R})\to B(L^2(\mathbb{R}_+))$, i.e., 
\begin{equation}\label{eq:unshrpp}
    e^{itq_+}E^T(X)e^{-itq_+}=E^T(X+t),
\end{equation}
which is not projection valued. 

We sketch here the construction of such a covariant POVM, which is inspired by the Naimark dilation theorem \cite{riesz2012functional,Busch2016a}. 
\begin{theorem}
\label{thm:naimark}
    Let $(\Omega,\mathcal{F})$ be a measurable space. Given a POVM $F:\mathcal{F} \to B(\HH)$, there exists a Hilbert space $\mathcal{K} \supset \HH$, an isometry $W:\HH\to \mathcal{K}$ and a PVM $P:\mathcal{F} \to B(\mathcal{K})$ such that for all $X \in \mathcal{F}$, 
    \begin{equation}
        F(X)=W^*P(X)W.
    \end{equation}
\end{theorem} 
The triple $(\mathcal{K},P,W)$ is called the Naimark dilation of $F$, and the projection $WW^*:\mathcal{K} \to \HH$ is called the Naimark projection.
We note that such a $P$ is not unique, but there is, up to unitary equivalence, a unique minimal dilation, given by setting $\mathcal{K}$ as the closure of $\rm{span}\{P(X)\varphi ~|~\varphi \in \HH,~X \in \mathcal{F}\}$.

We may now apply this to the setting of \cite{chandrasekaran2023algebra}. 
The observer Hamiltonian $q_+$ has spectral measure $P_+:\Bor(\mathbb{R}_+) \to L^2(\mathbb{R}_+)$ defined by $(P_+(X)f)(x)=\chi^+_{X}(x)f(x)$. A POVM $E^T:\Bor(\mathbb{R})\to B(L^2(\mathbb{R}_+))$ which satisfies Eq. \eqref{eq:unshrpp} can be 
 given as the restriction of the spectral measure $P^p$ of the
standard self-adjoint momentum observable $p$ acting in $L^2(\mathbb{R})$, which is the Fourier transform of the spectral measure of the position operator $q$ acting in $L^2(\mathbb{R})$.  Take the canonical isometric embedding $W:L^2(\mathbb{R}_+) \to L^2(\mathbb{R})$, then $E^T(X)=W^*P^p(X)W$ is a covariant POVM with dilation $(L^2(\mathbb{R}),P^p,W)$. We note that the dilation $P^p$ is also covariant under the extension of the unitary representation $U_+(t)= e^{itq_+}$ in $L^2(\mathbb{R}_+)$ to $U(t)=e^{itq}$ in $L^2(\mathbb{R})$. This extension follows the treatment in \cite{chandrasekaran2023algebra} of `ignoring' the positivity of the Hamiltonian for the sake of the analysis, but appears here systematically through the dilation theory of the unsharp conjugate time/momentum observable. We also note that covariant dilations are always possible and will play an important role throughout this paper; see also \ref{appsec:catt} for further details.

\paragraph{Example 4.}
Fix a complex, separable Hilbert space $\mathcal{H}$ with orthonormal Hilbert basis $\{e_n\}_{n \in N_0}$, with $N_0$ a non-empty subset of $\mathbb{Z}$, and set $N = \sum_{n\in N_0}n P[e_n]$ acting in $\mathcal{D}(N):=\{\varphi \in \mathcal{H}~|~\sum_{n \in N_0} n^2 |\langle e_n,\varphi \rangle|^2 < \infty\}$. Then $N$ generates a strongly continuous unitary representation of the circle group through $\theta \mapsto e^{i \theta N}$; $\theta \in [0, 2\pi)$. $N$ represents a number observable -- if $N_0 = \mathbb{N}$, $N$ is the standard number observable for an optical field, or harmonic oscillator. As in example 3., the positivity of $N$ precludes the existence of a canonically conjugate self-adjoint operator; nevertheless, there are phase-shift covariant POVMs, i.e., those which satisfy $ e^{i \theta N}E(X) e^{-i \theta N} = E(X + \theta)$ (addition understood modulo $2 \pi$), $X \in \Bor((0,2 \pi])$. These are known to take the form \cite{lahti1999covariant}
\begin{equation}
    E^{\rm{phase}}(X)= \sum_{n,m=0}^{\infty}c_{n,m}\int_X e^{i(n-m)\theta}|{n}\rangle \langle{m} | d \theta~,~X \in \Bor((0,2\pi]),
\end{equation}
where $(c_{n,m})$ is a positive matrix with $c_{n,n}=1$ for all $n$. $E^{\rm{phase}}$ is not projection-valued, but if $(c_{n,m})=1$ for all $n,m \in \mathbb{N}$, $E^{\rm{phase}}$ is called the canonical phase observable \cite{lahti2000characterizations} which enjoys various optimality properties, including the norm-$1$ property, thereby providing a localisable quantum reference frame if used in that capacity. The minimal Naimark dilation of the canonical phase is the spectral measure of the self-adjoint azimuthal angle operator, which is conjugate to the $z$ component of the angular momentum observable, whose spectrum is of course all of $\mathbb{Z}$.

\subsection{Relativisation for principal frames}\label{subsec:relativisation1}
In previous work (see e.g. \cite{miyadera2016approximating,Loveridge2017a,carette2023operational}), a class of invariants in $B(\mathcal{H}_S \otimes \mathcal{H}_R)$ was given through a \emph{relativisation} map, giving rise to observables which are interpreted as being relative to a given principal quantum reference frame (with an extension to finite homogeneous spaces given in \cite{glowacki2023quantum}; see also \cite{glowacki2024relativization} for a recent categorical development). For any $A \in B(\HH_S)$ and any 
principal quantum reference frame $(U_R, E, \mathcal{H}_R)$, there is a `relativised version' of $A$ given by 
\begin{equation}\label{eq:relprinc}
    \yen (A):= \int_G U_S(g)AU_S(g)^* \otimes dE(g).
    \end{equation}
See \cite{loveridge2017relativity} for a construction. Note that this can be extended to any POVM $E_S: \Bor(\Sigma)\to B(\mathcal{H}_S)$ effect-wise, i.e., $(\yen \circ E_S) (X) = \yen(E_S(X))$ for each $X \in \Sigma$. In particular, this also shows that unbounded operators can be relativised through their spectral measures.
The invariance of the right hand side under conjugation by $(U_S \otimes U_R)(g)$ follows from the covariance of $E$; the map $\yen$ is completely positive and normal, preserves effects (hence POVMs), the identity and the adjoint, and is multiplicative exactly when $E$ is projection valued. In simple settings $\yen$ gives rise to quantities typically viewed as `relative', for instance relative position, angle, phase and time arise through $\yen$ for appropriate choices of quantum reference frame. For example, the (spectral measure of) the position observable in $\mathcal{H}_S$ relativises to the (spectral measure of) relative position when $G = \mathbb{R}$ in example 1. above. The generalisation of \eqref{eq:relprinc} in the setting in which the value space of the quantum reference frame observable $E$ is a homogeneous left $G$-space is given in Subsec. \ref{subsec:relativisation}.

 Given that the standard 
framework of quantum theory, based on the description of single 
systems, functions adequately with regards to observation, 
it must be the case that, at least in certain regimes, the 
relativised description and the standard description are in approximate agreement, at least at a probabilistic level. Put more simply: 
the probabilities generated by e.g. the position operator are in agreement with what is observed in experiment, yet, as we have 
argued, the true observable is actually the relative position of 
system relative to a frame, raising the question of the empirical agreement of the two different descriptions. In order to compare the standard and relativised description, we may use the map $\eta_{\sigma}:B(\mathcal{H}_S \otimes \mathcal{H}_R)\to B(\mathcal{H}_S)$ defined in \eqref{eq:induct}, which we recall is given on tensor products as $\eta_{\sigma}(A \otimes B)= A \sigma(B)$, extended linearly and continuously to the whole space. $\eta_{\sigma}$ is a completely
positive normal conditional expectation (for a normal state $\sigma$) (e.g. \cite{Loveridge2017a}; in the QRF literature $\eta_{\sigma}$ is called a \emph{restriction} map, often denoted $\Gamma_{\sigma}$). Interpreting $\sigma$ as the state of the frame, we can compare the probabilities which arise from a POVM $E_S:\Bor(G)\to B(\mathcal{H}_S)$ with those of the POVM $\eta_{\sigma} \circ \yen \circ E_S :\Bor(G)\to B(\mathcal{H}_S)$. It is known \cite{Loveridge2017a,carette2023operational} that if $E$ satisfies the norm-$1$ property, for a system POVM $E_S$ and any $X \in \Bor(G)$, there is a sequence of normal states $(\sigma_n)$ for which $\lim_{n \to \infty} (\eta_{\sigma_n} \circ \yen)(E_S(X))=E_S(X)$, where the limit is in the weak/$\sigma$-weak topology, i.e. the topology of pointwise convergence of expectation values, and the above holds for any bounded operator in place of $E_S(X)$. The proof follows from noting that if $G$ is metrisable, for any $x \in G$ one can fix the given 
sequence $(\sigma_n)$ so that the sequence of scalar measures 
$\mathrm{tr} [E(\cdot)\sigma_n]$ converges to the Dirac measure at $x$; 
setting $x=e$ then gives the result. Conversely \cite{miyadera2016approximating}, for a badly localisable frame, there is typically a lower bound between a system observable and its restricted relativised version.

\section{Quantum reference frames for relativistic measurement schemes}
\label{sec:qrf and measurement}

\subsection{General considerations}
\label{sec:qrf_measurement_general}

As described in Sec.\ \ref{sec: measurement with symm}, on a spacetime $M$ with a group $G$ of (time-)orientation preserving isometries, one requires a physical system to act as a reference frame to distinguish between measurement schemes 
\begin{equation}
    m=(\Bf,\Cf,K,\chi,\sigma, B),
\end{equation} and 
\begin{equation}
    g.m:=(\Bf,\allowbreak{}^g\Cf,g.K,\allowbreak{}^g\chi,\allowbreak\beta(g^{-1})^*\sigma,\allowbreak\beta(g)B),
\end{equation} for any $g\in G$, where for both schemes the system QFT $\Af$ and probe QFT $\Bf$ are $G$-covariant fields on $M$, on which the group $G$ acts via the group homomorphisms $\alpha:G\to \Aut(\Af(M))$ and $\beta:G\to \Aut(\Bf(M))$ respectively as in Thm.~\ref{thm:measurement_cov}. Such a frame then also allows one to distinguish between measured/induced system observables $A=\varepsilon_\sigma(B)\in\Af(M)$ and $\alpha(g)A$. Having recalled the notion of a quantum reference frame in Sec.\ \ref{Sec:QRFs}, we here give a more operational description of the idea of selecting measurement schemes based on the configuration of a quantum reference frame.
In what follows, we assume $G$ to be a topological group that is locally compact, second countable and Hausdorff.

An experimenter may have access to some set of measurement schemes $\mathfrak{M}$, which due to symmetry considerations is taken to be closed under the action of $G$. The space of orbits of $\mathfrak{M}$ under the action of $G$, denoted by $\mathfrak{M}/G$, gives a partition of measurement schemes into sets of schemes that are, a priori, physically indistinguishable owing to the symmetries of the classical background. In order for an experimenter to decide on performing any particular measurement in some orbit, they 
employ an additional physical system as a reference, and as the system under consideration is a QFT, it seems most natural to 
use a quantum reference frame for this purpose (this also avoids the need to specify a hybrid quantum-classical dynamical system). For such a frame, we require that the value space $\Sigma$ (a homogeneous $G$-space) has the following property: 
\begin{definition}\label{def:suf_res}
    We say that a homogeneous $G$-space $\Sigma$ \textup{resolves} $\mathfrak{M}$ if for any $[m]\in \mathfrak{M}/G$ there exists an $s\in \Sigma$ and an $m\in[m]$ such that the stabiliser subgroup $G_s\subset G$ of $s$ is contained in the stabiliser subgroup $G_m\subset G$ of $m$, i.e., $G_s\subset G_m$.  
\end{definition}
Here the stabiliser subgroup $G_s$ of $s\in\Sigma$ is defined by
\begin{equation}
    G_s=\{g\in G:g.s=s\},
\end{equation}
and similarly for the stabiliser $G_m$ of $m\in\mathfrak{M}$. As $G$-spaces, one has 
\begin{equation}
    \Sigma\cong G/G_s,
\end{equation}
for all $s\in\Sigma$, where $G$ acts on $G/G_s$ by left multiplication,  $g.(g'G_s):=gg'G_s$  for all $g,g'\in G$. For fixed $s\in \Sigma$, this isomorphism is given by $g.s\mapsto gG_s$, using homogeneity to write any $s'\in \Sigma$ as $s'=g.s$ for $g\in G$. A similar mapping can be used to give an alternative characterisation of $\Sigma$ resolving $\mathfrak{M}$. 
\begin{proposition}
\label{prop:qrf_measurement_assignment}
A homogeneous $G$-space $\Sigma$ resolves $\mathfrak{M}$ if and only if there exists at least one map
\begin{equation}
    \mathfrak{f}:\Sigma\times \mathfrak{M}/G\to \mathfrak{M},
\end{equation}
such that
\begin{equation}
    \label{eq:qrf_measurement_assignment}
    \mathfrak{f}(s,[m])\in [m],\qquad\mathfrak{f}(g.s,[m])=g.\mathfrak{f}(s,[m])
\end{equation}
 for any $[m]\in \mathfrak{M}/G$, $s\in \Sigma$, $g\in G$.
\end{proposition}
\begin{proof} Given the existence of such a map $\mathfrak{f}$, one has for each $s\in\Sigma$ and $[m]\in\mathfrak{M}$ that
\begin{equation}
    G_s\subset G_{\mathfrak{f}(s,[m])}.
\end{equation}
In particular, this means $\Sigma$ resolves $\mathfrak{M}$. 

Conversely, if $\Sigma$ resolves $\mathfrak{M}$, we may construct a map $\mathfrak{f}$ as follows: for each $[m]\in\mathfrak{M}/G$, choose $s\in \Sigma$ and $m\in[m]$ so that $G_s\subset G_m$, and set $\mathfrak{f}(s,[m])=m$; then extend $\mathfrak{f}$ to $\Sigma\times \mathfrak{M}/G$ by imposing the relation $\mathfrak{f}(g.s,[m])=g.m$. This assignment is well-defined, as whenever $g.s=g'.s$ for $g,g'\in G$, this means $g^{-1}g'\in G_s\subset G_m$, and hence $g.m=g'.m$. Due to homogeneity of $\Sigma$, i.e.~by the fact that $G$ acts transitively on $\Sigma$, this defines the function $\mathfrak{f}:\Sigma\times \mathfrak{M}/G\to\mathfrak{M}$ completely.
\end{proof}
Condition \eqref{eq:qrf_measurement_assignment} expresses that $\mathfrak{f}$ defines a covariant parameterisation of measurement schemes by the value space $\Sigma$. One can therefore read Prop.~\ref{prop:qrf_measurement_assignment} as stating that $\Sigma$ resolving $\mathfrak{M}$ is equivalent to $\Sigma$ being `sufficiently large' to parameterise the elements of measurement scheme orbits. Generally this parameterisation is not unique. In the construction of a map $\mathfrak{f}$ in the proof above, this is reflected by the freedom of the initial choice of $s\in \Sigma$ and $m\in[m]$ for each $[m]\in\mathfrak{M}/G$.

A quantum reference frame $(U_R,E,\HH_R)$ with POVM $E:\Bor(\Sigma)\to B(\HH_R)$ and a state $\omega_R$ on $B(\HH_R)$ define a probability distribution $\mu_R:=\omega_R\circ E:\Bor(\Sigma)\to [0,1]$ on $\Sigma$ giving the outcome distribution for the POVM $E$ in state $\omega_R$. Let us now describe how such probability distributions, together with an equivalence class $[m]$ of measurement schemes and a map $\mathfrak{f}$ as above, determine an invariant induced observable relative to the QRF. 

Given a $\mu_R$-measurable function $f:\Sigma\to\CC$, the probability distribution $\mu_R$ defines a complex-valued random variable $f(\mathbf{S})$ on the sample space $\Sigma$ with for $X\in \Bor(\CC)$ the probability
\begin{equation}
    \mathds{P}[f(\mathbf{S})\in X]=\mu_R(f^{-1}(X)).
\end{equation}
Here $\mathbf{S}$ is randomly distributed in $\Sigma$ according to $\mu_R$. For $f\in L^\infty(\Sigma,\mu_R)$, the random variable $f(\mathbf{S})$ has expectation value
\begin{equation}
    \mathds{E}[f(\mathbf{S})]=\int_\Sigma f(s)\diff\mu_R(s).
\end{equation}

Recall that for each measurement scheme $m\in\mathfrak{M}$, there is an induced observable $A_m\in\Af(M)$ such that, given a state $\omega_S$ on $\Af(M)$, $\omega_S(A_m)$ yields the expected outcome of the measurement $m$, see Eq.~\eqref{eq:induct}. By choosing a map $\mathfrak{f}:\Sigma\times \mathfrak{M}/G\to \mathfrak{M}$ satisfying Eq.\ \eqref{eq:qrf_measurement_assignment}, 
we may use the statistics of the reference frame outcomes to determine which measurement scheme is actually performed within a given class $[m]\in\mathfrak{M}/G$. As the reference outcome is random, so is the selected measurement $m=\mathfrak{f}(\mathbf{S},[m])$. This results in a randomly chosen observable $A_{\mathfrak{f}(\mathbf{S},[m])}$ which depends on the reference state $\omega_R$
via the distribution of $\mathbf{S}$. For system states $\omega_S$ such that $\Sigma\ni s\mapsto \omega_S(A_{\mathfrak{f}(s,[m])})$ is
in $L^\infty(\Sigma,\mu_R)$, the complex-valued random variable 
$\omega_S(A_{\mathfrak{f}(\mathbf{S},[m])})$ on $\Sigma$ has expectation value
\begin{equation}
    \label{eq:rel_exp}
    \mathds{E}[\omega_S(A_{\mathfrak{f}(\mathbf{S},[m])})]=\int_\Sigma \omega_S(A_{\mathfrak{f}(s,[m])})\diff(\omega_R\circ E)(s),
\end{equation}
which gives the expected outcome of a measurement from class $[m]$
relative to the QRF in state $\omega_R$ for system state $\omega_S$.

Due to Thm.\ \ref{thm:measurement_cov}, the natural action of $G$ on $\mathfrak{M}$ gives $A_{g.m}=\alpha(g)A_m$, hence also $A_{\mathfrak{f}(g.s,[m])}=\alpha(g)A_{\mathfrak{f}(s,[m])}$. For mathematical convenience we shall assume that all measurable observables $A_m$ belong to some von Neumann algebra $\MM_S$. This is certainly true if $\Af(M)$ is a $C^*$-algebra and more generally this is a natural mathematical idealisation, as similarly to as in Sec.~\ref{subsec:relativisation1} one may consider the following discussion to be at the level of spectral measures in the case that $A_m$ can only be represented as an unbounded operator. We furthermore assume that the *-automorphisms $\alpha(g)$ extend to $\MM_S$ to define a weakly continuous action of *-automorphisms,   i.e.~for each $x\in \MM_S$, the map $g\mapsto \alpha(g)x$ is a continuous map from $G$ to $\MM_S$ w.r.t.~the weak operator topology on $\MM_S$.  Under these conditions, and a further technical assumption on the value space $\Sigma$ which we will introduce shortly, one can show for normal states $\omega_S$ on $\MM_S$ and $\omega_R$ on $B(\HH_R)$ that the integral \eqref{eq:rel_exp} is indeed well defined. Furthermore, one can in fact construct a $G$-invariant operator $A_{[m]}\in (\MM_S\otimes B(\HH_R))^{\alpha\otimes\Ad U_R}$ for which 
\begin{equation}
    \label{eq:relative_measured_exp}
    (\omega_S\otimes \omega_R)(A_{[m]})=\mathds{E}[\omega_S(A_{\mathfrak{f}(\mathbf{S},[m])})].
\end{equation}
A (formal) expression for this operator is given by
\begin{equation}
\label{eq:relative_measured}
    A_{[m]}=\int_\Sigma  A_{\mathfrak{f}(s,[m])}\otimes \diff E(s),
\end{equation}
to be understood in the same way as the integral defining the relativisation map \eqref{eq:relprinc}. We shall however avoid working directly with operator valued integrals of the form above by giving a construction of the relativisation map via a relation between quantum reference frames and crossed product algebras, see Sec.\ \ref{subsec:relativisation}. We can then apply this construction to define the operator $A_{[m]}$ solving Eq.~\eqref{eq:relative_measured_exp}. See Eq.\ \eqref{eq:yen_measured} and Cor.~\ref{cor:yen_measured}.

Eq.~\eqref{eq:relative_measured_exp} asserts that the
expectation values of $A_{[m]}$ in states of form $\omega_S\otimes\omega_R$ give the expected outcome (in state $\omega_S$) of a random measurement sampled from $[m]$, where the probability measure on the orbit of measurement schemes is pushed forward from a probability measure on $\Sigma$ (determined by $\omega_R$) by the map $\mathfrak{f}$. Roughly speaking, one could imagine that, in order to select a measurement $m\in[m]$, an experimenter first samples (or more loosely, makes a measurement of) the quantum reference frame, where the outcome of this reference measurement is some $s\in \Sigma$ and the probability distribution of potential outcomes is given by the measure $\mu_R=\omega_R\circ E$. Based on this outcome $s\in \Sigma$, they perform the measurement scheme $\mathfrak{f}(s,[m])$. Overall, this procedure yields a probability distribution on the potential measurement schemes as described above. 

Two remarks are in order. First, in general it will not be possible to perform a measurement of the covariant observable that would yield an outcome $s\in\Sigma$ with perfect precision, especially when $\Sigma$ is a continuum. Moreover any realistic probing of a reference frame observable would introduce additional noise into the procedure. Hence we view the description given above in terms of random variables as a mathematical idealisation that may only be realised in a limiting sense. Second, ideally, one would model the whole procedure of probing the reference observable and the quantum field theory in a single extended measurement scheme. However, as one typically cannot view such a scheme as localised in any compact region, this would require a significant generalisation of the framework described in Sec.\ \ref{sec: measurement in QFT}. We shall leave such efforts to future discussion.

In what follows we shall restrict ourselves to a particular subclass of quantum reference frames, namely those that are \textit{compactly stabilised} (see Def.\ \ref{def:QRF types}). Besides their technical advantages, which shall become apparent in the remainder of this section, compact stabilisation is a necessary requirement for QRFs to resolve measurement schemes on certain relevant spacetimes. This relies on the fact that non-trivial measurement schemes described in Sec.~\ref{sec: measurement in QFT} have (non-empty) compact coupling zones.
\begin{proposition}
    \label{prop:compact_stab}Given the set-up of Thm.~\ref{thm:measurement_cov}, assume $G$ to be locally compact, second countable and Hausdorff, and $M$ such that for each non-empty compact $K\subset M$ the stabiliser subgroup 
    \begin{equation}
        G_K:=\{g\in G:g.K=K\}
    \end{equation}
    is compact. Let $\mathfrak{M}$ a set of measurement schemes for system field $\Af$ by probe field $\Bf$ coupled in a non-empty compact coupling zone, such that $\mathfrak{M}$ is closed under the action of $G$ as given in Thm.~\ref{thm:measurement_cov}. Let $\mathcal{R}=(U_R,E,\HH_R)$ be a QRF for $G$ with value space $\Sigma$ that resolves $\mathfrak{M}$. Then $\mathcal{R}$ is compactly stabilised.
\end{proposition}
\begin{proof}
For any $m\in\mathfrak{M}$ invariant under the action of a group element $g\in G_m$, this means that $g.K=K$ for $K\subset M$ the coupling zone of $m$. Therefore $G_m\subset G_K$, the stabiliser subgroup associated with the compact set $K$. If $\Sigma$ resolves $\mathfrak{M}$, then there exist an $s\in \Sigma$ with $G_s\subset G_m$ for some $m\in\mathfrak{M}$, thus there is a non-empty compact $K\subset M$ such that $G_s\subset G_K$. Hence if $G_K\subset G$ is compact, this means $G_s$ (a closed subgroup) is compact. As $\Sigma\cong G/G_s$, 
this means that $\mathcal{R}$ is compactly stabilised.
\end{proof}

In the case of the background $M$ being the de Sitter static patch (such as considered in Sec.~\ref{sec: deSitter}), where we take $G=\RR\times \textnormal{SO}(3)$, one can see that the compactness assumption on $G_K$ as above is indeed satisfied. This is also true for the exterior of a Schwarzschild black hole (with group of isometries isomorphic to $\RR\times \textnormal{SO}(3)$). In fact, on these spacetimes the isometry groups act properly (see \cite[Thm.~1.2.9]{palaisExistenceSlicesActions1961} for various equivalent definitions), which implies the compactness assumption above. Note however that for the full de Sitter or Schwarzschild spacetimes, which in particular have bifurcate Killing horizons, there exist compact sets $K\subset M$ with noncompact stabilisers (namely any bifurcation surface). In these latter examples, the compact sets $K$ with non-compact stabilisers have empty interiors. It should be pointed out that in concrete implementations of the abstract measurement schemes described in Sec.~\ref{sec: measurement in QFT}, the coupling zones have non-empty interiors (see e.g.~\cite{FewVer_QFLM:2018}). This is due to the fact that in QFT, couplings between fields are often required to satisfy certain regularity properties (see e.g.~\cite{hawkinsNovelClassFunctionals2024} for such requirements in the context of perturbative QFT). With that in mind, one can modify Prop.~\ref{prop:compact_stab} accordingly; by restricting to measurement schemes with compact coupling zones with non-empty interiors, the conclusion of Prop.~\ref{prop:compact_stab} holds on a more general class of spacetimes $M$, namely those on which each compact $K\subset M$ \emph{with non-empty interior} have compact stabilisers $G_K$. Such backgrounds now include maximally extended Schwarzschild and de Sitter spacetimes, though are not expected to include arbitrary globally hyperbolic backgrounds. Nevertheless, for the purposes of this paper, it will be sufficient to restrict ourselves to backgrounds on which the assumptions of Prop.~\ref{prop:compact_stab} are satisfied.

In the remainder of this section, we shall give a complete description of all compactly stabilised quantum reference frames, and see what this implies for the von Neumann algebra of invariants $(\MM_S\otimes B(\HH_R))^{\alpha\otimes \Ad U_R}$ in which we will construct the operators $A_{[m]}$ obeying Eq.~\eqref{eq:relative_measured_exp}. Following this fairly general treatment, we specialise to the case where the symmetry group is given by $G=\RR\times H$, for $H$ some compact second countable Hausdorff group, as this brings us to setting of Sec.\ \ref{sec:type_change}. There we prove that, under certain conditions on the action $\alpha$ on $\MM_S$ and of $U_R$ on $\HH_R$, the von Neumann algebra $(\MM_S\otimes B(\HH_R))^{\alpha\otimes \Ad U_R}$ is finite.

\subsection{Compactly stabilised quantum reference frames}

Let $\Sigma$ be the value space of a compactly stabilised quantum reference frame $(U_R,E,\HH_R)$ for a locally compact second countable Hausdorff group $G$. Without loss of generality, we can assume $\Sigma=G/H$ for some compact subgroup $H\subset G$ (see also the discussion following Def.~\ref{def:suf_res}). For simplicity, we shall assume that $G/H$ admits a nonzero left $G$-invariant Radon measure $\mu_{G/H}$ (see Appendix \ref{appx:imprimitvity_thm}). This is sufficiently general for the applications that are of interest to us. Below we give a description of compactly stabilised quantum reference frames in terms of a subrepresentation of well-understood unitary representations of $G$. This result relies on the Naimark dilation theorem, see Thm.\ \ref{thm:naimark}, and Mackey's imprimitivity theorem, see e.g.\ \cite{kaniuth2012induced}, as applied to systems of covariance as in \cite{Cattaneo:1979}. Some aspects of these results are described in Appendix \ref{appx:imprimitvity_thm}. Section \ref{appx:compact_stab_qrf_charact} concludes with a proof of the following theorem.
\begin{theorem}
        \label{thm:qref_embed}
        Let $G$ be a locally compact second countable Hausdorff group, $H\subset G$ a compact subgroup such that $G/H$ admits a nonzero left $G$-invariant measure, $\lambda$ the natural action of left translations of $G$ on $L^2(G)$ as defined in Appendix \ref{appx:G_integrals} and $P:\Bor(G/H)\to B(L^2(G))$ the PVM given by $P(X)=T_{\chi_{q^{-1}(X)}}$ with $q:G\to G/H$ the canonical quotient map and $T_{\chi_{U}}\in B(L^2(G))$ denoting the multiplication operator by the characteristic function $\chi_U\in L^\infty(G)$ for any $U\subset G$. Then the following statements hold:
        \begin{enumerate}
            \item Let $\mathcal{R}=(U,E,\HH)$ be a compactly stabilised quantum reference frame for $G$ with value space $G/H$. Then there exists a separable Hilbert space $\KK$ and an isometry $W:\HH\to\KK\otimes L^2(G)$ that dilates $\mathcal{R}$ to the QRF $(1_\KK\otimes \lambda, 1_\KK\otimes P, \KK\otimes L^2(G))$. In particular, one has for all $g\in G$, $X\in\Bor(G/H)$
        \begin{equation}
            \label{eq:compact_QRF_compr}
            WU(g)=(1_\KK\otimes \lambda(g))W,\qquad E(X)=W^*(1_\KK\otimes P(X))W.
        \end{equation}
            \item Conversely, any isometry $W:\HH\to\KK\otimes L^2(G)$ for $\HH$ and $\KK$ separable Hilbert spaces satisfying for all $g\in G$
        \begin{equation}
            \label{eq:proj_invar}
            (1_\KK\otimes \lambda(g))WW^*=WW^*(1_\KK\otimes \lambda(g)),
        \end{equation}
        defines a compactly stabilised quantum reference frame $(W^*(1_\KK\otimes \lambda)W,W^*(1_\KK\otimes P)W,\HH)$ for $G$ with value space $G/H$.
        \end{enumerate}
    \end{theorem}
As quantum reference frames admit a natural notion of unitary equivalence (see Def.~\ref{def:qrf_equiv}), one can use the result above to give a complete description of all equivalence classes of compactly stabilised quantum reference frames.
    \begin{corollary}
    \label{cor:QRF_compact_subrep}
        Let $H\subset G$ as in Thm.\ \ref{thm:qref_embed}, then up to unitary equivalence, any QRF $\mathcal{R}$ of $G$ with value space $G/H$ is of the form
        \begin{equation}
            \mathcal{R}=((1_\KK\otimes \lambda)\restriction_{\HH_R},p(1_\KK\otimes P)\restriction_{\HH_R},\HH_R),
        \end{equation}
        for $\HH_R=p(\KK\otimes L^2(G))$ with $\KK$ some separable Hilbert space, $P(X)=T_{\chi_{q^{-1}(X)}}$ for $X\in\Bor(G/H)$ and $p\in(1_\KK\otimes \lambda(G))'$ a projection onto a subrepresentation of $(1_\KK\otimes \lambda,\KK\otimes L^2(G))$.
    \end{corollary}
    This result follows from the theorem above by setting $p=WW^*$. We shall now use the above description of quantum reference frames to characterise invariant observables under the action of $G$ in a joint system/reference frame algebra.

    \subsection{Invariant operators and the crossed product}
    \label{sec:inv_op_crossed}
    We consider a von Neumann algebra $\MM_S\subset B(\HH_S)$ of system observables acting on a separable Hilbert space $\HH_S$. 
    We let a locally compact second countable Hausdorff group $G$ act on $\MM_S$ via a group homomorphism $\alpha_S:G\to\Aut(\MM_S)$, such that the action $G\times \MM_S\ni(g,x)\mapsto \alpha_S(g)x$ is weakly continuous in $g$ for fixed $x$. Note that this weak continuity assumption is equivalent to assuming continuity in the $\sigma$-weak operator topology, as these agree on bounded sets, see \cite[Lem.~II.2.5]{Takesaki2001}, as well as in the strong operator topology, see the introductory remarks of \cite[Ch.~13.2]{kadison1997fundamentals}. We shall sometimes denote $\alpha_S(g)x$ by $\alpha_S(g,x)$, i.e.~we use $\alpha_S$ to denote both a group homomorphism as well as the group action on $\MM_S$ that it defines. We can summarise the assumptions above by saying that $\alpha_S:G\times\MM_S\to \MM_S$ is a \textit{continuous action of unital *-automorphisms}. 
    Let $\MM=\MM_S\otimes B(\HH_R)$ be a joint system/reference frame algebra, where $(U_R,E,\HH_R)$ is some compactly stabilised quantum reference frame for $G$ such that its value space admits a nonzero left G-invariant Radon measure. 
    We are interested in parameterizing all observables in $\MM^{\alpha_S\otimes\Ad  U_R}$, i.e. the invariant operators in the joint system/reference frame algebra. Note that, due to Cor.\ \ref{cor:QRF_compact_subrep}, one can assume without loss of generality that for some compact subgroup $H\subset G$, one has \begin{equation}
        \HH_R=p(\KK\otimes L^2(G)),\qquad U_R(g)=(1_\KK\otimes\lambda(g)),\qquad E(X)=p(1_\KK\otimes T_{\chi_{q^{-1}(X)}})\restriction_{\HH_R},
    \end{equation}
    for each $g\in G$ and$X\in\Bor(G/H)$ given some invariant projection $p\in (1_\KK\otimes \lambda(G))'$. We can now employ the following lemma
    \begin{lemma}
        Let $\MM\subset B(\HH)$ be a von Neumann algebra closed under an action $\alpha:G\times\MM\to\MM$ where for each $g\in G$ the map $\alpha(g)\in \Aut(\MM)$. Let $p\in \MM^{\alpha}$ be a G-invariant projection, then $(p \MM p)^{\alpha}=p\MM^{\alpha}p$.
    \end{lemma}
    \begin{proof}
        Since $(p\MM p)\subset \MM$, clearly one has $(p\MM p)^{\alpha}=p(p\MM p)^{\alpha}p\subset p\MM^{\alpha}p$. Now let $pap\in p\MM^{\alpha}p$ be such that $a\in \MM^{\alpha}$. One now sees that for $g\in G$ \begin{equation*}
            \alpha(g,pap)=\alpha(g,p)\alpha(g,a)\alpha(g,p)=pap,
        \end{equation*}
        hence $pap\in (p\MM p)^{\alpha}$, so $p\MM^{\alpha}p\subset (p\MM p)^{\alpha}$. 
    \end{proof}
    Applying this to the setting described above, where $p\in B(\KK\otimes L^2(G))^{\Ad 1_\KK\otimes\lambda}$ defines our quantum reference frame, yields 
    \begin{equation}
        (\MM_S\otimes B(\HH_R))^{\alpha_S\otimes\Ad U_R}=(1\otimes p)(\MM_S\otimes B(\KK\otimes L^2(G)))^{\alpha_S\otimes\Ad (1_\KK\otimes \lambda)}\restriction_{\HH_S\otimes\HH_R},
    \end{equation}
    hence by characterizing the operator algebras of the form $(\MM_S\otimes B(\KK\otimes L^2(G)))^{\alpha_S\otimes\Ad (1_\KK\otimes \lambda)}$, we can characterise invariant operators for joint system/reference frame algebras. Note that, under the obvious unitary equivalence $\HH_S\otimes \KK\otimes L^2(G)\cong \HH_S\otimes L^2(G)\otimes \KK$, one has
    \begin{equation}
        (\MM_S\otimes B(\KK\otimes L^2(G)))^{\alpha_S\otimes\Ad (1_\KK\otimes \lambda)}\cong (\MM_S\otimes B(L^2(G)))^{\alpha_S\otimes\Ad \lambda}\otimes B(\KK).
    \end{equation} It is a well established fact that fixed point algebras of the form $(\MM_S\otimes B(L^2(G)))^{\alpha_S\otimes\Ad \lambda}$ are given by a crossed product, see e.g.\  \cite{vanDaele:1978}.
    \begin{definition}
        \label{def:crossed_prod}
        Let $\MM\subset B(\HH)$ be a von Neumann algebra for some (separable) Hilbert space $\HH$, and let $G$ be a locally compact second countable Hausdorff group acting on $\MM$ via a continuous action of unital *-automorphisms $\alpha:G\times\MM\to \MM$. Define for each $a\in \MM$ an operator $\pi(a)\in B(L^2(G,\HH))$ such that for each $\varphi\in L^2(G,\HH)$, $g\in G$
        \begin{equation}
        \label{eq:pi_map}
            (\pi(a)\varphi)(g)=\alpha(g,a)\varphi(g).
        \end{equation}
        Furthermore, consider the right translation action $\tilde{\rho}:G\to \textbf{U}(L^2(G,\HH))$ as usual such that for each $\varphi\in L^2(G,\HH)$, $g,g'\in G$
        \begin{equation}
            (\tilde{\rho}(g)\varphi)(g')=\Delta_G(g)\varphi(g'g),
        \end{equation}
        where $\Delta_G$ is the modular function as defined by Eq.\ \eqref{eq:modular_func} in Appendix \ref{appx:imprimitvity_thm}. Then the \textup{crossed product} $\MM\rtimes_\alpha G$ is the von Neumann algebra given by
        \begin{equation}
            \MM\rtimes_\alpha G:=(\pi(\MM)\cup\tilde{\rho}(G))''\subset B(L^2(G,\HH)).
        \end{equation}
    \end{definition}
    Here we have used the notation $\tilde{\rho}$ for the right translation action on $L^2(G,\HH)$ to distinguish it from the usual right translation action $\rho$ on $L^2(G)$ as given in Appendix \ref{appx:G_integrals}. Naturally, under the usual isomorphism $L^2(G,\HH)\cong \HH\otimes L^2(G)$, $\tilde{\rho}$ corresponds to $1_\HH\otimes\rho$. As mentioned above, crossed products are fixed point algebras.
    \begin{proposition}[Commutation theorem for the crossed product (\cite{vanDaele:1978} Thm.\ 3.11)]
    \label{prop:comm_thm}
        Let $(\MM,G,\alpha)$ be as in Def.\ \ref{def:crossed_prod}, consider $\MM\otimes B(L^2(G))$ as a von Neumann algebra acting naturally on $L^2(G,\HH)\cong \HH\otimes L^2(G)$. Then
        \begin{equation}
            \MM\rtimes_\alpha G=(\MM\otimes B(L^2(G)))^{\alpha\otimes\Ad\lambda},
        \end{equation}
        with $\lambda$ the left translation action of $G$ on $L^2(G)$.
    \end{proposition}
    We can now directly apply this to our established set-up to find the following key result.
    \begin{theorem}
    \label{thm:invariants_crossed_prod}
        Let $\MM_S\subset B(\HH_S)$ be a von Neumann algebra and $G$ a locally compact second countable Hausdorff group, let $\alpha_S:G\times\MM_S\to\MM_S$ be a continuous action of unital *-automorphisms. Furthermore let $\mathcal{R}=(U_R,E,\HH_R)$ be a quantum reference frame for $G$ as in Thm.\ \ref{thm:qref_embed}. Then there exists a separable Hilbert space $\KK$ and a projection $p\in \lambda(G)'\otimes B(\KK)\subset B(L^2(G)\otimes \KK)$ such that, when viewing $\MM_S\rtimes_{\alpha_S} G$ as a von Neumann algebra on $\HH_S\otimes L^2(G)$, one has
        \begin{equation}
            \label{eq:invariants_crossed_prod}
            (\MM_S\otimes B(\HH_R))^{\alpha_S\otimes\Ad U_R}\cong (1_{\HH_S}\otimes p)((\MM_S\rtimes_{\alpha_S}G)\otimes B(\KK))(1_{\HH_S}\otimes p),
        \end{equation}
        via a spatial isomorphism induced by a unitary equivalence $\HH_R\cong p(L^2(G)\otimes \KK)$.
    \end{theorem}
    A crossed product $\MM_S\rtimes_{\alpha_S} G$ has an alternative description in the case that $\alpha_S$ is unitarily implementable on $\HH_S$, i.e.\ when there is a strongly continuous unitary representation $U_S:G\to\mathbf{U}(\HH_S)$ such that $\alpha=\Ad U_S$. This is due to the following result.
    \begin{proposition}[\cite{vanDaele:1978} Prop.\ 2.12]
    \label{prop:crossed_prod_unitary}
        Let $\MM$, $G$ be as in Def.\ \ref{def:crossed_prod}, let $U:G\to\mathbf{U}(\HH)$ a strongly continuous unitary representation, then let $V\in \mathbf{U}(L^2(G,\HH))$ be the unitary operator such that for $\varphi\in L^2(G,\HH)$, $g\in G$
        \begin{equation}
            (V\varphi)(g)=U(g)\varphi(g).
        \end{equation}
        One then has
        \begin{equation}
            \MM\rtimes_{\Ad U}G=V((\MM\otimes 1_{L^2(G)})\cup (U\otimes \rho)(G))''V^*,
        \end{equation}
        where for each $a\in \MM$, $g\in G$
        \begin{equation}
            \label{eq:crossed_prod_alt}
            V(a\otimes 1)V^*=\pi(a),\;V(U(g)\otimes \rho(g))V^*=\tilde{\rho}(g).
        \end{equation}
    \end{proposition}
    This allows us to give the following parameterisation of invariant observables.
    \begin{corollary}
    \label{cor:inv_unitary_param}
        Let $\MM_S,\HH_S,G,\mathcal{R}=(U_R,E,\HH_R)$ be as in Thm.\ \ref{thm:invariants_crossed_prod}, let $U_S:G\to \mathbf{U}(\HH_S)$ be a strongly continuous unitary representation of $G$ such that $\MM_S$ is closed under $\Ad U_S$ and let $W:\HH_R\to L^2(G)\otimes\KK$ for some separable Hilbert space $\KK$ be as in Thm. \ref{thm:qref_embed} (up to a swapping of tensor products). Let $V\in \mathbf{U}(\HH_S\otimes L^2(G))$ be as in Prop.\ \ref{prop:crossed_prod_unitary}, such that for $\varphi\in \HH_S\otimes L^2(G)\cong L^2(G,\HH_S)$, $g\in G$
        \begin{equation}
            (V\varphi)(g)=U_S(g)\varphi(g),
        \end{equation}
        and define
        \begin{equation}
            \tilde{\MM}_S^G=\{xU_S(g)\otimes\rho(g):x\in \MM_S,g\in G\}''\subset B(\HH_S\otimes L^2(G)).
        \end{equation}
        Then there exists a surjective completely positive unital normal linear map 
        \begin{equation}
            \zen:\tilde{\MM}_S^G\otimes B(\KK)\to (\MM_S\otimes B(\HH_R))^{\Ad U_S\otimes U_R},
        \end{equation}
        given by $\zen(\,\cdot\,)=\tilde{W}^*(\,\cdot\,)\tilde{W}$, where $\tilde{W}:\HH_S\otimes\HH_R\to \HH_S\otimes L^2(G)\otimes\KK$ an isometric embedding given by $\tilde{W}=(V^*\otimes 1_\KK)(1_{\HH_S}\otimes W)$.
    \end{corollary}
    \begin{proof}
        By Prop.\ \ref{prop:crossed_prod_unitary}, we have that 
        \begin{equation}
            \MM_S\rtimes_{\Ad U_S} G= V\tilde{\MM}_S^GV^*,
        \end{equation}
        hence by Thm.\ \ref{thm:invariants_crossed_prod} we have
        \begin{equation}
            (\MM_S\otimes B(\HH_R))^{\Ad U_S\otimes U_R}=(1\otimes W^*)((\MM\rtimes_{\Ad U_S} G)\otimes B(\KK))(1\otimes W)= \tilde{W}^*(\tilde{\MM}_S^G\otimes B(\KK))\tilde{W},
        \end{equation}
        where we note that the projection $p$ in Thm.\ \ref{thm:invariants_crossed_prod} is given by $p=WW^*$ and unitary isomorphism $\HH_R\cong p(L^2(G)\otimes \KK)$ yielding the spatial isomorphism of Eq.\ \eqref{eq:invariants_crossed_prod} is given by $W$. It directly follows that the map $\zen$ is surjective.
        
        Since $\tilde{W}$ is an isometry, the map $\zen(\,\cdot\,)=\tilde{W}^*(\,\cdot\,)\tilde{W}$ is a unital normal linear map and it is completely positive by \cite[Thm.\ IV.3.6]{Takesaki2001}. 
    \end{proof}
    The map $\zen$ parameterises the full algebra of invariant operators in a joint system/ref{\-}erence frame algebra. Compare this to the relativisation map $\yen$ considered in Sec.\ \ref{subsec:relativisation1} for principal reference frame, which generates a specific subset of all invariant operators. In the next subsection we shall generalise the definition of the relativisation map $\yen$ to compactly stabilised QRFs, and we will see that in the setting of Cor.\ \ref{cor:inv_unitary_param}, this relativisation map corresponds to a restriction of the map $\zen$. For this reason, we shall refer to $\zen$ as the \textit{extended relativisation map}.
    
    \subsection{Relativisation for compactly stabilised QRFs}\label{subsec:relativisation} 
    We have given a description of all invariant operators in a joint algebra of system and reference frame observables based on the theory of crossed products, provided the reference frame is compactly stabilised and that its value space admits a $G$-invariant measure. As discussed in Sec.\  \ref{subsec:relativisation1}, for principle reference frames a specific subset of these has a concrete interpretation as being system observables defined relative to the reference frame. In particular these are defined as the image of a relativisation map $\yen$ \cite{loveridge2017relativity}. Here we shall give a generalisation of this map to the class of compactly stabilised quantum reference frames. We shall furthermore show that the observables of the form of Eq.\ \eqref{eq:relative_measured}, which arose in the context of measurement theory, are well defined and given in terms of a relativisation map.

    Recall that for a principal reference frame $\mathcal{R}=(U_R,E,\HH_R)$ for a group $G$, the relativisation map $\yen:\MM_S\to (\MM_S\otimes B(\HH_R))^{\alpha_S\otimes\Ad U_R}$  is constructed explicitly in terms of the POVM $E$ and the orbits in $\MM_S$ under the action of $\alpha_S$ of $G$. One may therefore wonder, in our more general setting, which invariant operators may be constructed explicitly from system observables and the POVM, and to what extent these operators may be seen as relativised observables. More precisely, for a compactly stabilised quantum reference frame of the form $(U_R,E,\HH_R)$ with POVM $E:\Bor(G/H)\rightarrow B(\HH_R)$, one may ask which operators make up
    \begin{equation}
        \label{eq:inv_alg}
       (\MM_S\otimes E(\Bor(G/H))'')^{\alpha_S\otimes\Ad U_R} 
    \end{equation}
    with $\alpha_S$ a strongly continuous $G$-action of unital *-automorphisms as before. Under some assumptions on $E$, we can describe this fixed point algebra precisely. Firstly, we consider the following special case.
    \begin{lemma}
    \label{lem:pi_range}
        Let $\MM_S,G,\alpha_S$ be as in Thm.\ \ref{thm:invariants_crossed_prod}, let $(\lambda,P,L^2(G))$ be a sharp quantum reference frame (i.e.\ a system of imprimitivity) for $G$ acting on $G/H$ with $H$ compact, $\lambda$ the left translation action, $P(X)=T_{\chi_{q^{-1}X}}$ for any $X\in \Bor(G/H)$ with $q:G\to G/H$ the usual projection. Then
        \begin{equation}
            (\MM_S\otimes P(\Bor(G/H))'')^{\alpha_S\otimes\Ad \lambda}\cong \pi(\MM_S^{\alpha_S\restriction_H}),
        \end{equation}
        under the standard isomorphism given by $\HH_S\otimes L^2(G)\cong L^2(G,\HH_S)$, with $\pi(x)\in \MM_S\rtimes_{\alpha_S}G$ for $x\in \MM_S$ as defined by Eq.\ \eqref{eq:pi_map} of Def.\ \ref{def:crossed_prod}.
    \end{lemma}
    The proof of this lemma is given in Appendix \ref{appx:relativisation_proofs}.
    
    Observe that for any $x\in \MM_S^{\Ad U_S\restriction_H}$ and $\psi_i\in \HH_S$, $\varphi_i\in L^2(G)$ for $i\in \{1,2\}$, and regarding $\tilde{\varphi}_i:=\psi_i\otimes \varphi_i$ as an element of $ L^2(G,\HH_S)$, one finds
    \begin{align}
        \langle\tilde{\varphi}_1, \pi(x)\tilde{\varphi}_2\rangle=&\int_G\overline{\varphi_1(g)}\varphi_2(g)\langle \psi_1, \alpha_S(g,x)\psi_2\rangle_{\HH_S}\,\diff \mu_G(g)\nonumber\\
        =&\int_{G/H}\langle \psi_1, \alpha_S(g,x)\psi_2\rangle_{\HH_S}\int_H\overline{\varphi_1(gh)}\varphi_2(gh)\,\diff \mu_H(h)\,\diff \mu_{G/H}(gH),
    \end{align}
    where the integrand on second line is independent of the choice of representative $g\in gH$. The PVM $P$ defines for $\varphi_1,\varphi_2$ a complex valued Radon measure $\nu_{\varphi_1,\varphi_2}:\Bor(G/H)\to \CC$ with
    \begin{align}
    \label{eq:proj_measure}
        \nu_{\varphi_1,\varphi_2}(X)=&\langle \varphi_1,P(X)\varphi_2\rangle_{L^2(G)}\nonumber\\
        =&\int_{q^{-1}(X)}\overline{\varphi_1(g)}\varphi_2(g)\,\diff\mu_G(g)\nonumber\\
        =&\int_{X}\left(\int_{H}\overline{\varphi_1(gh)}\varphi_2(gh)\,\diff\mu_H(h)\right)\,\diff\mu_{G/H}(gH).
    \end{align}
    Hence we see
    \begin{equation}
        \label{eq:pi_matrixelt}
        \langle\tilde{\varphi}_1, \pi(x)\tilde{\varphi}_2\rangle=\int_{G/H}\langle \psi_1, \alpha_S(g,x)\psi_2\rangle_{\HH_S}\diff \nu_{\varphi_1,\varphi_2}(gH).
    \end{equation}
    This means we can formally write
    \begin{equation}
        \pi(x)=\int_{G/H}\alpha_S(g,x)\otimes\diff P(gH).
    \end{equation}
    For $H$ trivial, we recognise this expression as the relativisation map with respect to the PVM $P$, such as given in Sec.\ \ref{subsec:relativisation1}. Note furthermore that for $E:\Bor(G/H)\to B(\HH_R)$ given by $E(X)=p(P(X)\otimes 1_\KK)\restriction_{\HH_R}$, where $\HH_R\subset L^2(G)\otimes\KK$ is some $G$-invariant closed subspace and $p$ is the projection onto $\HH_R$, and $x\in M_S^{\alpha_S\restriction_H}$, that
    \begin{equation}
        (1_{\HH_S}\otimes p)(\pi(x)\otimes 1_\KK)\restriction_{\HH_S\otimes\HH_R}=\int_{G/H}\alpha_S(g,x)\otimes\diff E(gH).
    \end{equation}
    Here we naturally view $\HH_S\otimes\HH_R$ as a $G$-invariant subspace of $L^2(G,\HH_S)\otimes\KK$. Noting that by Cor.\ \ref{cor:QRF_compact_subrep} each compactly stabilised quantum reference frame can be brought into the form given above, we can formulate the following definition. 
    \begin{definition}
        \label{def:relativisation_compact_stab}
        Let $\MM_S,G,\alpha_S$ be as in Thm.\ \ref{thm:invariants_crossed_prod}, $H\subset G$ compact as in Thm.\ \ref{thm:qref_embed} and $\mathcal{R}=(U_R,E,\HH_R)$ a QRF for $G$ with value space $G/H$. Let $W:\HH_R\to L^2(G)\otimes \KK$ an isometric embedding for some $\KK$ separable such that $WW^*\in \lambda(G)'\otimes B(\KK)$ as in Thm.\ \ref{thm:qref_embed}, such that
        \begin{equation}
            U_R(g)=W^*(\lambda(g)\otimes 1_\KK) W,\qquad E(X)=W^*(T_{\chi_{q^{-1}(X)}}\otimes 1_\KK)W,
        \end{equation}
        with $q:G\to G/H$ the canonical quotient map, we define the \textup{relativisation map}
        \begin{equation}
\yen:\MM_S^{\alpha_S\restriction_H}\to (\MM_S\otimes E(\Bor(G/H))'')^{\alpha_S\otimes\Ad U_R},
        \end{equation}
        by
        \begin{equation}
            \label{eq:yenproj}
            \yen(x)=(1_{\HH_S}\otimes W^*)(\pi(x)\otimes1_\KK)(1_{\HH_S}\otimes W).
        \end{equation}
    \end{definition}

    The relativisation map $\yen$ provides us with a concrete way of generating invariant operators in the algebra \eqref{eq:inv_alg}. It has the following properties.
    \begin{proposition}
    \label{prop:yen_properties}
        For $\MM_S,G,H,\alpha_S$ as in Def.\ \ref{def:relativisation_compact_stab}, $\mathcal{R}=(U_R,E,\HH_R)$ a (general) quantum reference frame of $G$ with value space $G/H$, the map
        $$\yen:\MM_S^{\alpha_S\restriction_H}\to (\MM_S\otimes E(\Bor(G/H))'')^{\alpha_S\otimes\Ad U_R}$$
        is uniquely defined, meaning that the definition does not depend on a choice of isometry $W:\HH_R\to L^2(G)\otimes\KK$, and is a completely positive, unital, normal linear map. If $E$ is projection valued, then $\yen$ is a *-isomorphism.
    \end{proposition}
    A proof of this result can also be found in Appendix \ref{appx:relativisation_proofs}.
    
    Comparing the relativisation map $\yen$ to the extended relativisation map $\zen$ in the case that the group action $\alpha_S$ is unitarily implemented, one finds that, for $x\in \MM_S^{\alpha_S\restriction_H}$,
    \begin{align}
        \zen(x\otimes 1_{L^2(G)}\otimes 1_{\KK})=&(1_{\HH_S}\otimes W^*)(V(x\otimes 1_{L^2(G)})V^*\otimes 1_{\KK})(1_{\HH_S}\otimes W)\nonumber\\
        =&(1_{\HH_S}\otimes W^*)(\pi(x)\otimes 1_{\KK})(1_{\HH_S}\otimes W)\nonumber\\
        =&\yen(x),
    \end{align}
    with $\KK$, $W$ and $V$ as in Cor.\ \ref{cor:inv_unitary_param}. We can hence conclude that, in the case of a unitary implementable group action on $\MM_S$, the map $\zen$ indeed gives a (not necessarily unique) extension of the relativisation map $\yen$, hence motivating the term \textit{extended relativisation map}.

    For a (normal) state which does not correlate system and reference frame, expectation values of relativised observables can be given in the form of an integral expression:
    \begin{proposition}
        \label{prop:rel_expect}
        For $\MM_S,G,H,\alpha_S,\mathcal{R}$ as in Prop.\ \ref{prop:yen_properties}, let $\omega_S$ and $\omega_R$ be normal states on $\MM_S$ and $B(\HH_R)$ respectively. Then for any $x\in\MM_S^{\alpha_S\restriction_H}$, one finds
        \begin{equation}
            (\omega_S\otimes\omega_R)(\yen(x))=\int_{G/H} \omega_S(\alpha_S(g,x))\,\diff(\omega_{R}\circ E)(gH).
        \end{equation}
    \end{proposition}
    We refer again to Appendix \ref{appx:relativisation_proofs} for the proof. This in particular shows that for $x\in \MM_S^{\alpha_S\restriction_H}$ the formal expression
    \begin{equation} \int_{G/H}\alpha_S(g,x)\otimes \diff E(gH),
    \end{equation}
    is made precise by the operator $\yen(x)$.
    
    We now apply the relativisation map defined above to the setting of relativistic quantum measurement. Recall the discussion at the start of Sec.\ \ref{sec:qrf and measurement}. Let $M$ be a spacetime with a locally compact, second countable Hausdorff group $G$ of (time)-orientation perserving isometries, $\mathfrak{M}$ a set of measurement schemes for a QFT $\Af(M)$ closed under the natural group action of $G$ on measurement schemes, and $\mathcal{R}=(U_R,E,\HH_R)$ a compactly stabilised quantum reference frame for $G$ with value space $G/H$. We assume that $G/H$ resolves the measurement schemes in $\mathfrak{M}$ that are in the same $G$-orbit (see Def. \ref{def:suf_res}, note that compact stabilisation of $\mathcal{R}$ is automatic under the assumptions of \ref{prop:compact_stab}). One can then choose a $\mathfrak{f}:G/H\times \mathfrak{M}/G\to \mathfrak{M}$ satisfying the conditions of Eq.\ \eqref{eq:qrf_measurement_assignment}.
    
    We saw that any such function $\mathfrak{f}$ yields a map $(gH,[m])\mapsto A_{\mathfrak{f}(gH,[m])}\in\Af(M)$ for $gH\in G/H$, $[m]\in\mathfrak{M}/G$, where $A_m$ is the induced system observable for a measurement scheme $m\in\mathfrak{M}$. We furthermore assumed that  all operators of the form $A_{\mathfrak{f}(gH,[m])}$ belong to a von Neumann algebra $\MM_S$ such that the group action of $G$ on $\Af(M)$ extends to a continuous action of unital *-automorphisms $\alpha_S:G\times \MM_S\to\MM_S$.
    One then easily sees that 
    \begin{equation}
        A_{\mathfrak{f}(gH,[m])}=\alpha_S(g,A_{\mathfrak{f}(H,[m])})\,.
    \end{equation}
    By covariance of the map $\mathfrak{f}$, the measurement scheme $\mathfrak{f}(H,[m])$ must be $H$ invariant, so we have $A_{\mathfrak{f}(H,[m])}\in \MM_S^{\alpha_S\restriction_H}$. As a result, one can define
    \begin{equation}\label{eq:yen_measured}
        A_{[m]}:=\yen(A_{\mathfrak{f}(H,[m])})\in (\MM_S\otimes B(\HH_R))^{\alpha_S\otimes \Ad U_R},
    \end{equation}
    and verify that this operator solves Eq.\ \eqref{eq:relative_measured_exp}. This follows directly from applying Prop.~\ref{prop:rel_expect} to Eq.~\eqref{eq:yen_measured}. Summarising, we have proved the following.
    \begin{corollary}
        \label{cor:yen_measured}Under the assumptions just recalled, we have, for $\omega_S$ a normal state on $\MM_S$ and $\omega_R$ on $B(\HH_R)$, that
        \begin{equation}
            (\omega_S\otimes\omega_R)(\yen(A_{\mathfrak{f}(H,[m])}))=\int_{G/H}\omega_S(A_{\mathfrak{f}(gH,[m])})\otimes \diff(\omega_R\circ E)(gH).
        \end{equation}
    \end{corollary}

    As the operators $A_{[m]}$ are invariant under the action of $G$, we can regard them as physical observables from the perspective of our operational framework. We have the following inclusions following from \ref{prop:yen_properties}.
    \begin{equation}
         A_{[m]}\in \yen(\MM_S^{\alpha_S\restriction_H})''\subset(\MM_S\otimes E(\Bor(G/H))'')^{\alpha_S\otimes\Ad { U_R}}\subset (\MM_S\otimes B(\HH_R))^{\alpha_S\otimes\Ad { U_R}},
     \end{equation}
     and that $\yen(\MM_S^{\alpha_S\restriction_H})=\yen(\MM_S^{\alpha_S\restriction_H})''=(\MM_S\otimes E(\Bor(G/H))'')^{\alpha_S\otimes\Ad { U_R}}$ if the quantum reference frame $\mathcal{R}$ is sharp, i.e.\ if its covariant POVM is projection valued, i.e.~if $E(X\cap Y)=E(X)E(Y)$ for each $X,Y\in\Bor(G/H)$.

    In this section we have so far established a relation between invariant operators of a system algebra combined with a compactly stabilised quantum reference frame on the one hand and crossed product von Neumann algebras on the other. 
    We have seen in particular that these crossed product algebras can be used to parameterise the algebra of invariant operators through a compression. Furthermore, we can identify a particular subalgebra of these crossed products that parameterises the space of relativised observables, which is a subspace of the full algebra of invariant operators that appears in the literature on quantum reference frames and has a straightforward physical interpretation as discussed in Sec.\ \ref{subsec:relativisation1}.

    The condition of compact stabilisation arose from the requirement that the QRF should resolve equivalence classes of measurement schemes with compactly localised couplings, see Sec.\ \ref{sec:qrf_measurement_general}. Hence we have a concrete physical motivation for the use of this class of reference frames. This is a significant advance on the results of \cite{chandrasekaran2023algebra}, which only studied one example of an `observer' (i.e., for us, QRF), and which therefore did not provide insight into which properties of this example are important in order to obtain the description as a crossed product; by contrast, we have obtained a much more general class of QRFs, motivated by the need for relative measurements, for which the crossed product description arises naturally.
    
    We can now ask to what extent the results of \cite{chandrasekaran2023algebra} on the type structure of the algebra of invariants also extend to our generalised setting.  To derive a concrete result of this kind, we shall however need to make some further assumptions on the group $G$ acting on the system algebra.
    
\subsection{Invariant operators for $G=\RR\times H$ with compact $H$}
\label{sec:RHcase}
In what follows, we will apply the characterisation of $G$-invariants in a joint system/reference frame algebra in the specific case where $G = \RR \times H$, where $H$ is assumed to be a compact second countable Hausdorff group. Note that this group $H$ does not necessarily correspond to a stabiliser subgroup of the value space for the quantum reference frame.  As in the preceding, our system observables are given by a von Neumann algebra $\mathcal{M}_S$ equipped with a continuous action $\alpha_S$ of $G$ by unital *-algebra automorphisms (in both the ($\sigma$-)weak and strong operator topologies). The quantum reference frames that we consider for $G$ are compactly stabilised.

\subsubsection{Reference frames}
\label{sec:RHQRF}
Under the same assumptions as before, a compactly stabilised quantum reference frame $( U_{R} , E ,\HH_{R} )$ for the group $G=\RR\times H$ is defined on a left $G$-space $\Sigma\cong \RR\times H/H'$, where $H'\subset H$ is a compact subgroup. Note that since compact groups are unimodular, the value space $\Sigma$ will always admit a nonzero left $G$-invariant Radon measure (see Appendix \ref{appx:G_integrals}). Because these reference frames are compactly stabilised, by Thm.\ \ref{thm:qref_embed} they can always be described via an isometric embedding of the unitary representation $( \mathcal{H}_{R} , U_{R} )$ of $G=\RR\times H$ into the unitary representation $(L^2(G)\otimes \KK,\lambda\otimes 1_\KK)$ (This is a special case of the general embedding of quantum reference frames into an induced representation, see Prop.\ \ref{prop:cat_impr_thm}). Decomposing $L^2(G)\cong L^2(\RR)\otimes L^2(H)$, the quantum reference frame $( U_{R} , E ,\HH_{R} )$ can be described in terms of an isometric embedding 
\begin{equation}
    W:\HH_R\to L^2(\RR)\otimes L^2(H)\otimes \KK.
\end{equation}
This map $W$ must be such that it intertwines the group action, i.e.\ for each $(t,h)\in \RR\times H$ 
\begin{equation}
    W U_R(t,h)=(\lambda_\RR(t)\times \lambda_{H}(h)\otimes 1_\KK)W,
\end{equation}
with $\lambda_\RR\times\lambda_{H}$ the diagonal regular left translation action of $\RR\times H$ on $L^2(\RR)\otimes L^2(H)$. Furthermore, for $X\in \Bor(\RR\times H/H')$ we must have
\begin{equation}
    E(X)=W^*P(X)W,
\end{equation}
with $P:\Bor(\RR\times H/H')\to L^\infty(\RR)\otimes L^\infty (H)\otimes 1_\KK$ the PVM defined so that for each $U\subset \RR$ and $V\subset H/H'$ one has
\begin{equation}
    P(U\times V)=T_{\chi_U}\otimes T_{\chi_{q^{-1}V}}\otimes 1_\KK,
\end{equation}
with $q:H\to H/H'$ the canonical quotient map. 
Concretely, this means that 
\begin{equation}
    \HH_R\cong p (L^2(\RR)\otimes L^2(H)\otimes \KK),
\end{equation}
with $p=W^*W\in \lambda_\RR(\RR)'\otimes \lambda_{H}(H)'\otimes B(\KK)$, such that under this isomorphism $U_R\cong \lambda_\RR\times \lambda_H\otimes 1_\KK$. By acting on the $L^2(\RR)$ entry with a Fourier transform, we show that this Hilbert space decomposes into a direct integral (see \cite[Ch.\ IV.8]{Takesaki2001} for the relevant definitions) of the following form.
\begin{proposition}
\label{prop:spect_decomp}
    Let $(U_R,E,\HH_R)$ be a compactly stabilised quantum reference frame for the group $\RR\times H$ with $H$ a compact second countable Hausdorff group. Then 
    \begin{equation}
        \HH_R\cong\int_\RR^\oplus \KK(\xi)\diff\xi,
    \end{equation}
    where $(U_{H}(\xi),\KK(\xi))$ is a measurable field of separable unitary representations of $H$ and under this isomorphism one has for any $(t,g)\in \RR\times H$ and $\psi\in \int_\RR^\oplus \KK(\xi)\diff\xi$ that
    \begin{equation}
        (U(t,g)\psi)(\xi)=\exp(it\xi)U_{H}(\xi)(g)\psi(\xi).
    \end{equation}
    This decomposition is unique up to almost everywhere unitary equivalence of the field of $\KK(\xi)$.  
\end{proposition}
A proof is given in Appendix.\ \ref{appx:spect_decomp_proof}. This decomposition can now be used to define a spectral multiplicity function.
\begin{definition}
\label{def:spec_mult}
    Let $(U_R,E,\HH_R)$ be a compactly stabilised quantum reference frame for the group $\RR\times H$ with $H$ a compact second countable Hausdorff group, any direct integral representation of $\HH_R$ as in Prop.\ \ref{prop:spect_decomp} 
    \begin{equation}
        \HH_R\cong \int_\RR^\oplus \KK(\xi)\diff \xi,
    \end{equation}
    will be called a \textup{spectral decomposition}, with corresponding \textup{spectral multiplicity function} $m_{U_R}:\RR\to \mathbb{N}_0\cup\{\infty\}$ given by
    \begin{equation}
        m_{U_R}(\xi)=\dim(\KK(\xi)).
    \end{equation}
\end{definition}
Note that by \cite[Lem.\ IV.8.12]{Takesaki2001} the spectral multiplicity function is always (Borel) measurable and by the essential uniqueness of the spectral decomposition it is essentially uniquely defined by the unitary representation $(U_R(.,1_H),\HH_R)$ of $\RR$, i.e. up to a zero-measure set with respect to the Lebesgue measure on $\RR$. The spectral multiplicity as defined here resembles the multiplicity function for general normal bounded operators on separable Hilbert spaces as given in \cite[Ch.\ IX \S 10]{Conway}. We shall see in Sec.\ \ref{sec:type_change_thm} that this spectral multiplicity function has implications for the structure of the algebra of
invariant joint observables of the system and reference frame.

\subsubsection{Invariant observables}
\label{sec:RHinv_obs}
Let $( U_{R} , E ,\mathcal{H}_{R} )$ be a compactly stabilised QRF for the group $G=\RR\times H$ acting on $\Sigma=\RR\times H/H'$ defined by an isometric embedding $W:\HH_R\to L^2(\RR)\otimes L^2(H)\otimes \KK$ as above. Via Thm.~\ref{thm:invariants_crossed_prod}, we shall parameterise the von Neumann algebra of \emph{invariant operators} 
\begin{equation}
    ( \MM_S \otimes B(\mathcal{H}_{R}) )^{\alpha_S\otimes \Ad U_R}
\end{equation} by the algebra 
\begin{equation}
    (\MM_S\otimes B(L^2(\RR)\otimes L^2(H)))^{\alpha_S\otimes \Ad \lambda_\Sigma\times \lambda_H}\otimes B(\KK).
\end{equation} Since $G=\RR\times H$, operators are invariant under the action of $G$ if and only if they are invariant under the action of both $\RR$ and $H$, or
\begin{multline}
\label{eq:RHfixedpt}
    (\MM_S\otimes B(L^2(\RR)\otimes L^2(H)))^{\alpha_S\otimes \Ad \lambda_\RR\times \lambda_H}=\\\left((\MM_S\otimes B(L^2(\RR)\otimes L^2(H)))^{\alpha_S\restriction_{\RR}\otimes \Ad (\lambda_{\RR}\otimes 1_{L^2(H)})}\right)^{\alpha_S\restriction_H\otimes \Ad (1_{L^2(\RR)}\otimes \lambda_H)},
\end{multline}
where we have used $\alpha_S\restriction_\RR$ and $\alpha_S\restriction_H$ as a shorthand for the group actions of $\RR$ and $H$ given by the maps $t\mapsto \alpha_S((t,1_H),.)$ and $h\mapsto \alpha_S((1_\RR,h),.)$ respectively. We can hence analyse this fixed point algebra in steps, where we first consider the fixed points under the action of $\RR$ and then restrict further to the fixed points of $H$. The latter fixed points can be described either via a crossed product or via an averaging prescription.

\paragraph{A step wise analysis of the fixed point algebra Eq.~\eqref{eq:RHfixedpt}} We first consider the fixed points with respect to the action of $\RR$, we see that
\begin{align}
    (\MM_S\otimes B(L^2(\RR)\otimes L^2(H))&)^{\alpha_S\restriction_{\RR}\otimes \Ad(\lambda_{\RR}\otimes 1_{L^2(H)})}\nonumber\\&\cong (\MM_S\otimes B(L^2(\RR))^{\alpha_S\restriction_{\RR}\otimes \Ad \lambda_\RR}\otimes B( L^2(H))\nonumber\\
    &\cong (\MM_S\rtimes_{\alpha_S\restriction_{\RR}} \RR)\otimes B( L^2(H)),
\end{align}
where we have used natural isomorphisms and Prop.\ \ref{prop:comm_thm}. Applying this to Eq.\ \eqref{eq:RHfixedpt}, this gives
\begin{align}
\label{eq:RHRcrossed}
    (\MM_S\otimes B(L^2(\RR)\otimes L^2(H))&)^{\alpha_S\otimes \Ad \lambda_\RR\times \lambda_H}\nonumber\\
    &\cong \left((\MM_S\rtimes_{\alpha_S\restriction_{\RR}} \RR)\otimes B( L^2(H))\right)^{\alpha_S\restriction_H\otimes \Ad(1_{L^2(\RR)}\otimes \lambda_H)}\,,
\end{align}
where we note that $\alpha_S(h)\restriction_H\otimes \Ad 1_{L^2(\RR)}$ acts as a unital *-automorphism on $(\MM_S\rtimes_{\alpha_S\restriction_{\RR}} \RR)$ via the identification $L^2(\RR,\HH_S)\cong \HH_S\otimes L^2(\RR)$, such that $(\MM_S\rtimes_{\alpha_S\restriction_{\RR}} \RR)\subset \MM_S\otimes B(L^2(\RR))$. The fixed point algebra of Eq.\ \eqref{eq:RHRcrossed} can be expressed in various ways, either via a crossed product over $H$
\begin{multline}
    \left((\MM_S\rtimes_{\alpha_S\restriction_{\RR}} \RR)\otimes B( L^2(H))\right)^{\alpha_S\restriction_H\otimes \Ad (1_{L^2(\RR)}\otimes \lambda_H)}\\
    \cong \left((\MM_S\rtimes_{\alpha_S\restriction_{\RR}} \RR)\rtimes_{\alpha_S(h)\restriction_H\otimes \Ad 1_{L^2(\RR)}}H\right),
\end{multline}
or, as $H$ is compact, via an averaging prescription
\begin{multline}
    \left((\MM_S\rtimes_{\alpha_S\restriction_{\RR}} \RR)\otimes B( L^2(H))\right)^{\alpha_S\restriction_H\otimes \Ad (1_{L^2(\RR)}\otimes \lambda_H)}\\=\left\{\int_H (\alpha_S\restriction_H\otimes\Ad(1_{L^2(\RR)}\otimes \lambda_H))(h,x)\,\diff\mu_H(h):x\in (M_S\rtimes_{\alpha_S\restriction_\RR}\RR)\otimes B(L^2(H))\right\}.
\end{multline}

\paragraph{Parameterising the invariant observables} For a QRF $(U_R,E,\HH_R)$ as given above, which can be described using the isometry $W:\HH_R\to L^2(\RR)\otimes L^2(H)\otimes \KK$, this entails that
\begin{multline}
    \left(\MM_S\otimes B(\HH_R)\right)^{\alpha_S\otimes \Ad U_R}\\\cong(1_{\HH_S}\otimes W)^*\left(\left((\MM_S\rtimes_{\alpha_S\restriction_{\RR}} \RR)\otimes B( L^2(H))\right)^{\alpha_S\restriction_H\otimes \Ad (1_{L^2(\RR)}\otimes \lambda_H)}\otimes B(\KK)\right)(1_{\HH_S}\otimes W)\\=\left((1_{\HH_S}\otimes W)^*\left((\MM_S\rtimes_{\alpha_S\restriction_{\RR}} \RR)\otimes B( L^2(H)\otimes \KK)\right)(1_{\HH_S}\otimes W)\right)^{\alpha_S\restriction_H\otimes\Ad U_R\restriction_H}.
\end{multline}
The fixed points under the $H$ action can now be directly obtained by an averaging prescription. In analysing the structure of $\left(\MM_S\otimes B(\HH_R)\right)^{\alpha_S\otimes \Ad U_R}$, 
it may sometimes prove convenient to first analyse the structure of the algebra of observables invariant under the action of $\RR$ 
\begin{multline}
    \left(\MM_S\otimes B(\HH_R)\right)^{\alpha_S\restriction_\RR\otimes \Ad U_R\restriction_\RR}\\\cong (1_{\HH_S}\otimes W)^*\left((\MM_S\rtimes_{\alpha_S\restriction_{\RR}} \RR)\otimes B( L^2(H)\otimes \KK)\right)(1_{\HH_S}\otimes W),
\end{multline}
and then note that 
\begin{equation}
    \left(\MM_S\otimes B(\HH_R)\right)^{\alpha_S\otimes \Ad U_R}\subset \left(\MM_S\otimes B(\HH_R)\right)^{\alpha_S\restriction_\RR\otimes \Ad U_R\restriction_\RR}.
\end{equation}
We shall exploit this feature in the following section, where we prove that the algebra $\left(\MM_S\otimes B(\HH_R)\right)^{\alpha_S\restriction_\RR\otimes \Ad U_R\restriction_\RR}$, and hence also $\left(\MM_S\otimes B(\HH_R)\right)^{\alpha_S\otimes \Ad U_R}$, is a finite von Neumann algebra under certain assumptions on $\alpha_S\restriction_\RR$ and the spectral multiplicity function $m_{U_R}$.

\section{(Semi-)finiteness of the invariant observable algebra}
\label{sec:type_change} 

Consider a spacetime $M$, with group of smooth time-orientation and orientation preserving isometries $G$.
In Sec.~\ref{sec: measurement in QFT}, we saw how for a $G$-covariant QFT on $M$, defined by a net $N\mapsto\Af(M;N)$ of $C^*$-algebras, any faithful state $\omega$ determines a faithful GNS representation which in turn defines a net of von Neumann algebras, to be denoted $\mathfrak{A}(M;N)$, and which are typically expected to be of type $\textnormal{III}$ for $N$ sufficiently regular. By assumption, the group $G$ has an action $\alpha:G\times \Af(M)\to \Af(M)$ on $\Af(M)$, where $\alpha(g)$ are unital *-automorphisms. If furthermore a state $\omega$ is invariant under $\alpha$, i.e. for each $a\in \Af(M)$ one has $\omega(\alpha(g,a))=\omega(a)$, then for the associated GNS representation $(\HH,\pi,\Omega)$ one has a unitary representation $U:G\to \mathbf{U}(\HH)$ such that $\pi(\alpha(g,a))=U(g)\pi(a)U(g)^*$ and $U(g)\Omega=\Omega$, see e.g.\ \cite[Prop.\ 5.1.17]{AdvAQFT}. In the particular case where $\RR\subset G$ forms a one-paramenter subgroup, it can be shown that the one-paramenter unitary group $U\restriction_\RR:\RR\to \mathbf{U}(\HH)$ is strongly continuous (and hence defines a generally unbounded self adjoint generator) if and only if for each $a\in \Af(M)$ 
\begin{equation}
\label{eq:strong_cont_action}
    \lim_{t\to 0}\omega(a^*\alpha(t,a))=\omega(a^*a).
\end{equation}
Note that, if a region $N\in \Reg(M)$ is closed under the action $\alpha$, it follows that $\mathfrak{A}(M;N)$ is closed under the action $\Ad U$. By a slight abouse of notation, we shall also denote this action by $\alpha:G\times \mathfrak{A}(M;N)\to \mathfrak{A}(M;N)$.

Next we discuss an important class of states on $\mathcal{M}\equiv \mathfrak{A}(M;N)$, relevant for the analysis of type reduction, namely the \emph{KMS states}. These states are a natural generalisation of Gibbs states defined in the context of statistical mechanics, see e.g.\ \cite{cmp/1103840050} or \cite[Ch.\ V.1]{Haag:book}. Here is the formal definition, following \cite[Def.~5.3.1 and Prop.~5.3.7]{Bratteli2}.

\begin{definition}
    For a von Neumann algebra $\mathcal{M}$ with a continuous $\RR$-action of  unital *-automorphisms  $\alpha:\RR\times\mathcal{M}\to\MM$ (equivalently in the ($\sigma$-)weak and strong operator topologies), a normal state $\omega$  (i.e. $\sigma$-weakly continuous, or equivalently given in terms of a positive trace-class operator with unit trace, as in Sec.~\ref{Sec:QRFs}) is a \emph{KMS state} with respect to $\alpha$ of inverse temperature $\beta>0$ if for each $a,b\in\MM$, there is a bounded continuous function $F_{x,y}:S_\beta\to\CC$ defined on the strip 
    \begin{equation}
        S_\beta:=\{z\in \CC: 0\leq\Im(z)\leq \beta\},
    \end{equation} such that $F_{x,y}$ is holomorphic on the interior of $S_\beta$ and for $t\in\RR$
\begin{equation}
    F_{x,y}(t)=\omega(y\alpha(t,x)),\qquad F_{x,y}(t+i\beta)=\omega(\alpha(t,x)y).
\end{equation}
\end{definition}

The KMS condition also finds applications in the general theory of von Neumann algebras, via a set of results which have become known as modular or Tomita-Takasaki theory, see e.g.\ \cite{takesaki1970tomita,vanDaele:1978,takesaki2002}. A key result in this theory is the following
\begin{theorem}
    \label{thm:modular_action}
    For $\MM$ a ($\sigma$-finite) von Neumann algebra and $\phi:\MM\to\CC$ a faithful positive normal linear functional, there exists a unique strongly continuous group action $\sigma:\mathbb{R}\times\MM\to \MM$ such that $\sigma(t,.)$ is a *-automorphism for which $\phi$ satisfies a KMS condition as follows: For any $x,y\in \MM$ there exists a bounded and continuous function 
    \begin{equation}
        F_{x,y}:\{z\in \CC:0\leq \Im(z)\leq 1\}\to \CC
    \end{equation}
    that is holomorphic on the interior of its domain, such that for $t\in \RR$
    \begin{equation}
        F_{x,y}(t)=\phi(\sigma(t,x)y),\qquad F_{x,y}(t+i)=\phi(y\sigma(t,x)).
    \end{equation}
\end{theorem}
Here $\sigma$ is known as the \textit{modular action} and $\sigma(.,x)$ as the modular flow of $x$ associated with $\phi$. Here the requirement that $\MM$ is $\sigma$-finite entails that the set of faithful positive normal linear functionals on $\MM$ is non-empty, see \cite[Prop.\ II.3.19]{Takesaki2001}. Although this theorem can be extended beyond the case of $\sigma$-finite von Neumann algebras using the theory of weights and left Hilbert algebras, in particular assigning modular actions to weights (see def.\ \ref{def:weight_trace}), for our purposes the result above will be sufficient and hence we shall adhere to this simplified setting following \cite{vanDaele:1978}. The construction of the modular action can be sketched as follows: Assume that $\MM$ acts on a Hilbert space $\HH$ such that there is a $\Omega\in\HH$ cyclic and separating for $\MM$ such that $\phi(x)=\langle \Omega,x\Omega\rangle $ (i.e. we identify $\MM$ with its GNS representation).
Define $\mathcal{D}=\MM\Omega$ dense in $\HH$ and define $S:\mathcal{D}\to\mathcal{D}$ as the (in general unbounded) linear operator such that for each $x\in \MM$
\begin{equation}
    S x\Omega=x^*\Omega.
\end{equation}
It can be shown that (the closure of) $S$ has a unique polar decomposition $S=J\Delta^{\frac{1}{2}}$, where $\Delta$ is a positive (in general unbounded) operator and $J$ is an anti-unitary operator on $\HH$. This allows for the modular action to be defined as 
\begin{equation}
    \sigma(t,x)=\Delta^{it}x\Delta^{-it}.
\end{equation} That the von Neumann algebra $\mathcal{M}$ is indeed closed under this action, is a result known as \textit{Tomita's theorem}, see e.g.\ \cite{takesaki1970tomita}. It can be verified that this group action satisfies the properties listed in Thm.\ \ref{thm:modular_action}. We refer to \cite[Thm.\ VI.1.19 and VIII.1.2]{takesaki2002} for a full proof of existence and uniqueness of the modular action.

Even though any faithful normal state $\omega$ on $\mathcal{M}$ defines a modular action, in general this action is not related to any geometric action of the group $\RR$ on $M$. However, if $\omega$ is a KMS state with respect to some action of *-automorphisms $\alpha:\RR\times \MM\to\MM$ at (positive) inverse temperature $\beta$, where $\alpha$ was induced by some spacetime symmetry, it follows from Thm.\ \ref{thm:modular_action} that $\sigma(t,.)=\alpha(-\beta t,.)$ and we say that the modular flow is geometric. Note here the modular flow goes in the opposite direction to the flow defined by $\alpha$, as the KMS condition in Thm.~\ref{thm:modular_action} corresponds to an inverse temperature of $-1$. This is a somewhat unfortunate convention in the definition of the modular action which we shall adhere to in view of consistency with the mathematical literature.

As alluded to above, the modular action finds applications in the general theory of von Neumann algebras. Particularly in combination with the crossed product, the definition of which we recalled in Sec.\ \ref{sec:inv_op_crossed}, the modular action becomes a powerful tool to analyse the structure of type III von Neumann algebras. One result in this context is particularly relevant to our analysis of algebras of invariant operators, see e.g.~\cite[Ch.~II.3]{vanDaele:1978}.
\begin{theorem}
\label{thm:cross_prod_semifinite}
    Let $\MM$ be a $\sigma$-finite von Neumann algebra and $\sigma$ be the modular action associated with a faithful normal positive linear functional $\phi:\MM\to \CC$. Then the von Neumann algebra $\MM\rtimes_\sigma \RR$ admits a faithful normal semifinite trace, and is hence a semifinite von Neumann algebra.
\end{theorem}
We recall the definitions of a semifinite von Neumann algebra and trace, as well as their relation to the type classification of factors in Appendix \ref{appx:factor_types}. We shall refer to $\MM\rtimes_\sigma \RR$ as the \textit{modular crossed product} of $\MM$. While not crucial to our analysis, it should be pointed out that, up to spatial isomorphism (i.e.\ the adjoint action of a unitary in $\mathbf{U}(L^2(\RR,\HH))$), the modular crossed product algebra does not depend on the choice faithful normal positive linear functional on $\MM$, see e.g.\ \cite[Thm.\ II.2.3]{vanDaele:1978}. 

A constructive proof of Thm.\ \ref{thm:cross_prod_semifinite} is given in \cite{vanDaele:1978}, and though we shall not reproduce the full proof of this statement here, we shall recall how the trace on $\MM\rtimes_\sigma\RR$ is constructed. Without loss of generality we once again identify $\MM$ with its GNS representation  defined by $\phi$, i.e.\ we embed $\MM \subset B(\HH)$ for some Hilbert space $\HH$, such that there is a cyclic and separating vector $\Omega$ for $\MM$ satisfying $\phi(\,\cdot\,)=\langle\Omega,\,\cdot\,\Omega\rangle$. Note that $\MM\rtimes_\sigma \RR$ acts on the Hilbert space $L^2(\RR,\HH)\cong\HH\otimes L^2(\RR)$. In the construction of the semifinite trace we require the Fourier transform $\mathcal{F}\in\mathbf{U}(L^2(\RR))$, which we here set to be the continuous extension of the map $f\mapsto \mathcal{F}f$ for $f\in C_c(\RR)\subset L^2(\RR)$ with
    \begin{equation}
        \mathcal{F}f(\xi)=\frac{1}{\sqrt{2\pi}}\int_\RR \exp(i\xi t)f(t)\diff t.
    \end{equation}
Now for each $K\subset \RR$ compact one can define a vector $\Psi_K\in L^2(\RR,\HH)$ given by
\begin{equation}
    \label{eq:psyK}
    \Psi_K(t)=(\mathcal{F}^*f_K)(t)\Omega,
\end{equation}
with $f_K\in L^2(\RR)$ given by
\begin{equation}
    f_K(\xi)=\chi_K(\xi)\exp(\xi/2),
\end{equation}
where $\chi_K$ is the characteristic function with support $K$. This allows one to define the following normal positive linear functional $\tau_K:\MM\rtimes_\sigma \RR\to\CC$ via
\begin{equation}
    \tau_K(x)=\langle\Psi_K,x\Psi_K\rangle.
\end{equation}
One now defines a trace $\tau:(\MM\rtimes_{\sigma}\RR)^+\to\RR^+\cup\{\infty\}$ on the modular crossed product via
    \begin{equation}
    \label{eq:modular_trace}
        \tau(x)=\sup_{K\subset\RR\text{ compact}}\tau_K(x).
    \end{equation}
In \cite{vanDaele:1978}, the following properties are established.  
\begin{proposition}
    \label{prop:modular_trace}
    The map $\tau:(\MM\rtimes_{\sigma}\RR)^+\to\RR^+\cup\{\infty\}$ defined above is a semifinite faithful normal trace. For each $K\subset \RR$ compact, defining a projection $p_K=1_{\HH}\otimes \mathcal{F}^*T_{\chi_K}\mathcal{F}\in (\MM\rtimes_{\sigma}\RR)^+$, seen as an operator on $L^2(\RR,\HH)\cong \HH\otimes L^2(\RR)$ under the usual isomorphism, one has
    \begin{equation}
        \label{eq:finite_trace}\tau(p_K)=\tau_K(p_K)=\int_K\exp(\xi)\diff \xi<\infty.
    \end{equation}
    The trace $\tau\restriction_{p_K(\MM\rtimes_{\sigma}\RR)^+p_K}=\tau_K\restriction_{p_K(\MM\rtimes_{\sigma}\RR)^+p_K}$ extends uniquely to a finite faithful normal trace on $p_K(\MM\rtimes_{\sigma}\RR)p_K$.
\end{proposition}
It should be noted that in fact for any measurable $U\subset\RR$ such that $\xi\mapsto \chi_U(\xi)\exp(\xi)$ is in $L^1(\RR)$, one can define a projection $p_U=1_{\HH}\otimes \mathcal{F}^*T_{\chi_U}\mathcal{F}\in(\MM\rtimes_{\sigma}\RR)^+$ with
\begin{equation}
    \tau(p_U)=\langle \Omega,\Omega\rangle \sup_K\int_K\chi_U(\xi)\exp(\xi)\diff\xi=\int_U\exp(\xi)\diff\xi<\infty.
\end{equation} This implies via \cite[Eq.~V.2.2]{Takesaki2001} that $\tau$ restricts to a finite trace on $p_U(\MM\rtimes_{\sigma}\RR)^+p_U$ with for $x\in (\MM\rtimes_{\sigma}\RR)^+$
\begin{equation}
    \tau(p_Uxp_U)\leq\Vert x\Vert\int_U\exp(\xi)\diff \xi.
\end{equation}

We shall apply this result to the setting where one has geometric modular flow on a von Neumann algebra $\MM=\mathfrak{A}(M;N)$ associated to quantum field observables localisable on a spacetime region $N$, i.e.\ where $\MM$ admits a normal faithful KMS state with respect to some geometrically induced action of unital *-automorphisms $\alpha:\RR\times\MM\to\MM$. In this case, one has that $\MM\rtimes_{\alpha}\RR\cong \MM\rtimes_{\sigma}\RR$ is semifinite. In Sec.\ \ref{sec:inv_op_crossed} we saw that crossed products naturally arise when one considers invariant operator algebra for some (sufficiently well behaved) group acting on a system von Neumann algebra joint with some compactly stabilised quantum reference frame. In this setting, one can therefore use the semifinite faithful normal trace on the modular crossed product described above to construct a trace on the algebra of invariant operators.

\subsection{Type reduction for invariant operators}
\label{sec:type_change_thm}
Recall our setting discussed at the start of Sec.\ \ref{sec:type_change}.
We consider a $\sigma$-finite von Neumann algebra $\MM_S$ (the system algebra) that admits a group action of unital *-automorphisms for $G=\RR\times H$, $\alpha_S:G\times\MM_S\to\MM_S$, where $H$ is a compact second countable Hausdorff group, $\omega$ is a normal faithful KMS state on $\MM_S$ of inverse temperature $\beta\in(0,\infty)$ with respect to the action $\alpha_S\restriction_\RR:\RR\times\MM_S\to\MM_S$, where we are viewing $\RR$ naturally as a subgroup of $G$. We furthermore assume that $\omega$ is invariant under the full action of $G$ and (without loss of generality) that $\alpha_S$ is unitarily implemented, meaning that $\MM_S$ acts on a Hilbert space $\HH_S$ such that $U_S:G\to\mathbf{U}(\HH_S)$ is a unitary representation of $G$ with $\alpha_S=\Ad U_S$. Furthermore, it follows from the KMS condition that $U_S$ can be chosen such that $U_S\restriction_\RR$ is a strongly continuous unitary representation. As a mild additional condition, which is satisfied in the specific cases that we consider, we assume that the full map $U_S:G\to \mathbf{U}(\HH_S)$ is continuous in the strong operator topology. 

In addition to the system algebra, we have some compactly stabilised quantum reference frame $(U_R,E,\HH_R)$ for the group $G$. We are interested in the algebra of invariant operators $(\MM_S\otimes B(\HH_R))^{\Ad U_S\otimes U_R}$. As we saw in Sec.\ \ref{sec:RHinv_obs}, these invariant observables can be parameterised by a von Neumann algebra involving the crossed product $(\MM_S\rtimes_{\Ad U_S\restriction_\RR}\RR)$, which by Thm.\ \ref{thm:cross_prod_semifinite} and the fact that $\MM_S$ admits a faithful normal KMS state for $\Ad U_S\restriction_\RR$, is a semifinite von Neumann algebra. We shall show that under the assumptions above, semifiniteness of this crossed product implies that the algebra $(\MM_S\otimes B(\HH_R))^{\Ad U_S\otimes U_R}$ is semifinite, or, following an argument similar to that appearing in \cite{chandrasekaran2023algebra}, even a finite von Neumann algebra when a further (sufficient) condition is met. We refer to this fact as `type reduction', because $\MM_S\otimes B(\HH_R)$ is typically a type $\mathrm{III}_1$ von Neumann algebra, while the invariant operator algebra $(\MM_S\otimes B(\HH_R))^{\Ad U_S\otimes U_R}$ is of a `lower' type. The sufficient condition leading to the algebra of invariants being a finite von Neumann algebra, is formulated as a condition on the spectral multiplicity function $m_{U_R}$ of the quantum reference frame as defined in Def.\ \ref{def:spec_mult}. This is all made precise in the following theorem.
\begin{theorem}
\label{thm:type_change}
    Let $\MM_S$ be a $\sigma$-finite von Neumann algebra on a separable Hilbert space $\HH_S$ and $G=\RR\times H$ with $H$ a compact Hausdorff group that acts spatially on $\MM_S$ via a strongly continuous unitary representation $U_S:G\to\mathbf{U}(\HH_S)$, 
    such that there exists a $G$-invariant faithful normal KMS state $\omega:\MM_S\to \CC$ of inverse temperature $\beta\in (0,\infty)$ with respect to the action of $\RR$. Furthermore, let $(U_R,E,\HH_R)$ be a compactly stabilised quantum reference frame for $G$. Then the algebra of invariant operators
    \begin{equation}
        \left(\MM_S\otimes B(\HH_R)\right)^{\Ad(U_S\otimes U_R)}
    \end{equation} is semifinite. If furthermore
    \begin{equation}
        \label{eq:type_change_condition}
        \int_\RR\exp(-\beta\xi)m_{U_R}(\xi)\diff\xi<\infty,
    \end{equation}
    where $m_{U_R}$ is the spectral multiplicity function as given by Def.\ \ref{def:spec_mult}, then the algebra of invariant operators is finite.
\end{theorem}
\begin{proof}
    We shall construct a faithful normal trace on $\left(\MM_S\otimes B(\HH_R)\right)^{\Ad(U_S\otimes U_R)}$ and show that it is (semi)finite under the given assumptions. This construction builds on the semifinite trace on $\MM_S\rtimes_\sigma\RR$ for $\sigma$ a modular action on $\MM_S$, as given by Eq.\ \ref{eq:modular_trace}. Recall that up to unitary equivalence, the unitary representation $(U_R,\HH_R)$ associated with a compactly stabilised quantum reference frame for the group $G=\RR\times H$ is a subrepresentation of $(\lambda_\RR\times\lambda_H\otimes 1_\KK,L^2(\RR)\otimes L^2(H)\otimes \KK)$, for some separable Hilbert space $\KK$. This is a direct consequence of Thm.~\ref{thm:qref_embed} as applied to the group $G=\RR\times H$ as discussed in Sec.~\ref{sec:RHQRF}. In particular, there exists a projection $p\in \lambda_\RR(\RR)'\otimes\lambda_H(H)'\otimes B(\KK)$ such that $\HH_R\cong p(L^2(\RR)\otimes L^2(H)\otimes \KK)$. As in the proof of Prop.\ \ref{prop:spect_decomp} (see Appx.~\ref{appx:spect_decomp_proof}), this implies 
    \begin{equation}
        \hat{p}:=(\mathcal{F}\otimes 1_{L^2(H)\otimes \KK})p(\mathcal{F}^*\otimes 1_{L^2(H)\otimes \KK})\in L^\infty(\RR)\otimes\lambda_H(H)'\otimes B(\KK),
    \end{equation} and there exists a projection valued measurable function  $\xi\mapsto \hat{p}(\xi)\in \lambda_H(H)'\otimes B(\KK)$ such that for $\Psi,\Psi'\in L^2(\RR)$, $\psi,\psi'\in L^2(H)\otimes \KK$, one has
    \begin{equation}
        \langle \Psi\otimes \psi,\hat{p}(\Psi'\otimes\psi')\rangle=\int_\RR\overline{\Psi(\xi)}\Psi'(\xi)\langle \psi,\hat{p}(\xi)\psi'\rangle\diff\xi,
    \end{equation} yielding a spectral decomposition (see Def.\ \ref{def:spec_mult})
    \begin{equation}
        \HH_R\cong \int_\RR^\oplus \hat{p}(\xi)(L^2(H)\otimes \KK)\diff\xi.
    \end{equation}

    Consider the algebra
    \begin{equation}
        (\MM_S\otimes B(\HH_R))^{\Ad U_S\restriction_\RR\otimes U_R\restriction_\RR}.
    \end{equation}
    By Thm.\ \ref{thm:invariants_crossed_prod} (see also the discussion of Sec.\ \ref{sec:RHinv_obs}) we have that 
    \begin{equation}
    \label{eq:fixed_point_iso}
        (\MM_S\otimes B(\HH_R))^{\Ad U_S\restriction_\RR\otimes U_R\restriction_\RR}\cong (1_{\HH_S}\otimes p)((\MM_S\rtimes_{\Ad U_S\restriction_\RR}\RR)\otimes B(L^2(H)\otimes\KK))(1_{\HH_S}\otimes p).
    \end{equation}
    By the fact that $\MM_S$ admits a faithful normal KMS state $\omega$ of inverse temperature $\beta>0$, it follows that $\Ad U_S\restriction_\RR(t,.)=\sigma(-\frac{t}{\beta},.)$ for $\sigma$ the modular action associated with $\omega$. Identifying $\MM_S$ once again with its GNS representation defined by $\omega$, we see that $\MM_S\rtimes_{\Ad U_S\restriction_\RR}\RR= \tilde{U}(\MM_S\rtimes_\sigma\RR)\tilde{U}^*$ for the unitary transformation $\tilde{U}\in\mathbf{U}(L^2(\RR,\HH_S))$ implementing the scale transformation
    \begin{equation}
        (\tilde{U}\psi)(t)=\frac{1}{\sqrt{\beta}}\psi\left(-\frac{t}{\beta}\right).
    \end{equation}
    $(\MM_S\rtimes_\sigma\RR)$ admits a semifinite trace $\tau_1$  by Prop.\ \ref{prop:modular_trace}, hence $\MM_S\rtimes_{\Ad U_S\restriction_\RR}\RR$ also admits a semifinite trace $\tilde{\tau}_1$ given by
    \begin{equation}
        \tilde{\tau}_1(x)=\tau_1(\tilde{U}^*x\tilde{U})=\sup_{K\subset\RR\text{ compact}}\langle\tilde{U}\Psi_K,x\tilde{U}\Psi_K\rangle,
    \end{equation}
    for $x\in (\MM_S\rtimes_{\Ad U_S\restriction_\RR}\RR)^+$ and $\Psi_K\in L^2(\RR,\HH_S)$ as given in Eq.\ \eqref{eq:psyK}.
    Note that $B(L^2(H)\otimes\KK)$ is a type I factor (see Appendix \ref{appx:factor_types}), and therefore admits a unique (up to normalisation) semifinite trace $\tau_2:B(L^2(H)\otimes\KK)^+\to\RR^+\cup\{\infty\}$ defined by
    \begin{equation}\tau_2(x)=\sup_N\sum_{n=0}^N\langle \psi_n,x\psi_n\rangle,\end{equation}
    for $\{\psi_n\}_{n\in \mathbb{N}}$ a countable orthonormal basis of $L^2(H)\otimes\KK$. One can now define a semifinite trace $\tau:=\tilde{\tau}_1\otimes\tau_2$ on $(\MM_S\rtimes_{\Ad U_S\restriction_\RR}\RR)\otimes B(L^2(H)\otimes\KK)$ given by
    \begin{equation}
        \tau(x)=\sup_{\substack{K\subset\RR\text{ compact, }\\N\in\mathbb{N}}}\sum_{n=0}^N\langle \tilde{U}\Psi_K\otimes\psi_n,x(\tilde{U}\Psi_K\otimes\psi_n)\rangle,
    \end{equation}
    for $x\in ((\MM_S\rtimes_{\Ad U_S\restriction_\RR}\RR)\otimes B(L^2(H)\otimes\KK))^+$. That this is a faithful normal semifinite trace, follows from Prop.\ \ref{prop:trace_product}. As a result $(\MM_S\rtimes_{\Ad U_S\restriction_\RR}\RR)\otimes B(L^2(H)\otimes\KK))$ is a semifinite von Neumann algebra. For the projection $p\in \lambda_\RR(\RR)'\otimes\lambda_H(H)'\otimes B(\KK)$ above we have 
    \begin{equation}
        (1_{\HH_S}\otimes p)\in ((\MM_S\rtimes_{\Ad U_S\restriction_\RR}\RR)\otimes B(L^2(H)\otimes\KK))^+,
    \end{equation}
    and hence
    \begin{multline}
    \label{eq:compression_positive}
        \left((1_{\HH_S}\otimes p)((\MM_S\rtimes_{\Ad U_S\restriction_\RR}\RR)\otimes B(L^2(H)\otimes\KK))(1_{\HH_S}\otimes p)\right)^+\\
        \subset ((\MM_S\rtimes_{\Ad U_S\restriction_\RR}\RR)\otimes B(L^2(H)\otimes\KK))^+,
    \end{multline}
    By semifiniteness of $\tau$ on $((\MM_S\rtimes_{\Ad U_S\restriction_\RR}\RR)\otimes B(L^2(H)\otimes\KK))$, we have that for any
    \begin{equation}
        x\in \left((1_{\HH_S}\otimes p)((\MM_S\rtimes_{\Ad U_S\restriction_\RR}\RR)\otimes B(L^2(H)\otimes\KK))(1_{\HH_S}\otimes p)\right)^+,
    \end{equation}
    there exists a $y\in (\MM_S\rtimes_{\Ad U_S\restriction_\RR}\RR)\otimes B(L^2(H)\otimes\KK))$ with $0<y\leq x$ with $0<\tau(y)<\infty$. Denote $\tilde{\HH}=\HH_S\otimes L^2(\RR)\otimes L^2(H)\otimes\KK$ and $\tilde{\HH}_p=(1_{\HH_S}\otimes p)\tilde{\HH}\subset \tilde{\HH}$. Then for any $\psi\in\tilde{\HH}_p^\perp$ and $\varphi\in\tilde{\HH}$ we have that
    \begin{equation}
        \vert\langle \phi,y\psi\rangle\vert\leq \Vert \sqrt{y}\varphi\Vert\Vert \sqrt{y}\psi\Vert,
    \end{equation}
    and 
    \begin{equation}
        \Vert \sqrt{y}\psi\Vert^2=\langle\psi,y\psi\rangle\leq \langle \psi,x\psi\rangle=0.
    \end{equation}
    It follows that $y=y(1_{\HH_S}\otimes p)$, and hence by self-adjointness $y=(1_{\HH_S}\otimes p)y(1_{\HH_S}\otimes p)$. We can thus conclude that $\tau$ is semifinite on 
    \begin{equation}
        (1_{\HH_S}\otimes p)((\MM_S\rtimes_{\Ad U_S\restriction_\RR}\RR)\otimes B(L^2(H)\otimes\KK))(1_{\HH_S}\otimes p),
    \end{equation}
    and by the isomorphism \eqref{eq:fixed_point_iso}, we see that $(\MM_S\otimes B(\HH_R))^{\Ad U_S\restriction_\RR\otimes U_R\restriction_\RR}$ is semifinite.

    For the algebra of $G$-invariants $(\MM_S\otimes B(\HH_R))^{\Ad U_S\otimes U_R}\subset (\MM_S\otimes B(\HH_R))^{\Ad U_S\restriction_\RR\otimes U_R\restriction_\RR}$, we can similarly see that this algebra is semifinite if the dilated algebra $((\MM_S\rtimes_{\Ad U_S\restriction_\RR}\RR)\otimes B(L^2(H)\otimes\KK))^{\Ad U_R\restriction_H\otimes 1_{L^2(\RR)}\otimes \lambda_H\otimes 1_\KK}$ is semifinite. To show this, we use that $\omega$ is invariant under the action of $H$, hence $\tilde{\tau}_1$ is invariant under $\Ad U_R\restriction_H\otimes 1_{L^2(\RR)}$. The trace $\tau_2$ is invariant under $\Ad \lambda_H\otimes 1_\KK$ due to cyclicity of the trace, therefore $\tau=\tilde{\tau}_1\otimes\tau_2$ is invariant under the action of $H$. Since $H$ is compact and its unitary representation $U_R\restriction_H$ is strongly continuous, we can define the surjective map 
    \begin{multline}
        \operatorname{Avg}_H:(\MM_S\rtimes_{\Ad U_S\restriction_\RR}\RR)\otimes B(L^2(H)\otimes\KK)\\\to((\MM_S\rtimes_{\Ad U_S\restriction_\RR}\RR)\otimes B(L^2(H)\otimes\KK))^{\Ad U_R\restriction_H\otimes 1_{L^2(\RR)}\otimes \lambda_H\otimes 1_\KK},
    \end{multline}
    given by
    \begin{equation}
        \operatorname{Avg}_H(x)=\int_H \Ad (U_R\restriction_H\otimes 1_{L^2(\RR)}\otimes \lambda_H\otimes 1_\KK)(h,x)\,\diff\mu_H(h).
    \end{equation}
    Note that this is a linear positive map, so in particular for \begin{equation}
        x,y\in ((\MM_S\rtimes_{\Ad U_S\restriction_\RR}\RR)\otimes B(L^2(H)\otimes\KK))^+,
    \end{equation}  with $x\geq y\geq 0$, one has
    \begin{equation}
        \operatorname{Avg}_H(x)\geq\operatorname{Avg}_H(y)\geq 0,
    \end{equation}
    so that 
    \begin{multline}
        \operatorname{Avg}_H(((\MM_S\rtimes_{\Ad U_S\restriction_\RR}\RR)\otimes B(L^2(H)\otimes\KK))^+)\\=\left(((\MM_S\rtimes_{\Ad U_S\restriction_\RR}\RR)\otimes B(L^2(H)\otimes\KK))^{\Ad U_R\restriction_H\otimes 1_{L^2(\RR)}\otimes \lambda_H\otimes 1_\KK}\right)^+. 
    \end{multline}
    By dominated convergence, we can see that for $x\in ((\MM_S\rtimes_{\Ad U_S\restriction_\RR}\RR)\otimes B(L^2(H)\otimes\KK))^+$
    \begin{equation}
        \tau(\operatorname{Avg}_H(x))=\int_H \tau\left(\Ad (U_R\restriction_H\otimes 1_{L^2(\RR)}\otimes \lambda_H\otimes 1_\KK)(h,x)\right)\diff\mu_H(x)=\tau(x).
    \end{equation}
    Since the map $\operatorname{Avg}_H(x)$ is positive linear, we conclude that $\tau$ restricts to a semifinite trace on  $((\MM_S\rtimes_{\Ad U_S\restriction_\RR}\RR)\otimes B(L^2(H)\otimes\KK))^{\Ad U_R\restriction_H\otimes 1_{L^2(\RR)}\otimes \lambda_H\otimes 1_\KK}$ and hence this algebra is semifinite.
    
    We now verify finiteness of $(\MM_S\otimes B(\HH_R))^{\Ad U_S\restriction_\RR\otimes U_R\restriction_\RR}$ under the assumption given by Eq.\ \eqref{eq:type_change_condition}. By Eq.~\eqref{eq:compression_positive}, $\tau$ restricts to a finite trace on $(1_{\HH_S}\otimes p)((\MM_S\rtimes_{\Ad U_S\restriction_\RR}\RR)\otimes B(L^2(H)\otimes\KK))(1_{\HH_S}\otimes p)$ if and only if 
    \begin{equation}
        \tau(1_{\HH_S}\otimes p)<\infty.
    \end{equation}
    To evaluate this trace, we use 
    \begin{equation}
        (\tilde{U}\Psi_K)(t)=\frac{1}{\sqrt{\beta}}(\mathcal{F}^*f_K)\left(-\frac{t}{\beta}\right)\Omega=(\mathcal{F}^*\tilde{f}_K)(t)\Omega,
    \end{equation}
    where $\tilde{f}_K(\xi)=\sqrt{\beta}f_K(-\beta\xi),$ with $f_K(x)=\exp(x/2)\chi_K(x)$. We can now calculate that
    \begin{align}
        \tau(1_{\HH_S}\otimes p)=&\sup_{\substack{K\subset\RR\text{ compact, }\\N\in\mathbb{N}}}\sum_{n=0}^N\langle \mathcal{F}^*\tilde{f}_K\otimes\psi_n,p(\mathcal{F}^*\tilde{f}_K\otimes\psi_n)\rangle\nonumber\\
        =&\sup_{\substack{K\subset\RR\text{ compact, }\\N\in\mathbb{N}}}\sum_{n=0}^N\langle \tilde{f}_K\otimes\psi_n,\hat{p}(\tilde{f}_K\otimes\psi_n)\rangle\nonumber\\
        =&\sup_{\substack{K\subset\RR\text{ compact, }\\N\in\mathbb{N}}}\beta\int_\RR\chi_K(-\beta\xi)\exp(-\beta\xi)\sum_{n=0}^N \langle\psi_n,\hat{p}(\xi)\psi_n\rangle\diff\xi\nonumber\\
        =&\sup_{N\in\mathbb{N}}\beta\int_\RR\exp(-\beta\xi)\sum_{n=0}^N \langle\psi_n,\hat{p}(\xi)\psi_n\rangle\diff\xi.
    \end{align}
    Note that 
    \begin{equation}\label{eq:multiplicity function in a basis}
        \sum_{n\in \mathbb{N}}\langle\psi_n,\hat{p}(\xi)\psi_n\rangle=\tau_2(p(\xi))=\dim(\hat{p}(\xi)(L_2(H)\otimes\KK))=m_{U_R}(\xi),
    \end{equation} 
    with $m_{U_R}$ the spectral multiplicity function of the QRF.
    By monotone convergence, it follows that $\tau(1_{\HH_S}\otimes p)<\infty$ if and only if 
    \begin{equation}
        \beta\int_\RR\exp(-\beta\xi)m_{U_R}(\xi)\diff\xi<\infty.
    \end{equation}

    This shows that if Eq.\ \eqref{eq:type_change_condition} holds, then $(\MM_S\otimes B(\HH_R))^{\Ad U_S\restriction_\RR\otimes U_R\restriction_\RR}$ is finite. Since any subalgebra of a finite algebra is finite, this also implies that $(\MM_S\otimes B(\HH_R))^{\Ad U_S\otimes U_R}$ is finite.
\end{proof}
\begin{remark}
    To prove semifiniteness of the algebra of invariant operators, it was required that the unitary representation $U_S$ of $G$ is strongly continuous, as this is for instance needed to ensure that the map $\operatorname{Avg}_H$ used in the proof above is well defined. For the proof that $(\MM_S\otimes B(\HH_R))^{\Ad U_S\otimes U_R}$ is finite under the conditions described above, we can actually drop the requirement that $U_S\restriction_H$ is strongly continuous. The only continuity property on $U_S$ that is relevant to this case is that $U_S\restriction_\RR$ is strongly continuous, which can always be realised due to the assumption that $\MM_S$ admits a faithful normal KMS state with respect to $\Ad U_S\restriction_\RR$. In fact, one can similarly relax the compactness assumption on $H$ in the proof of finiteness of the algebra of invariants, provided that Eq.~\eqref{eq:type_change_condition} still holds.
\end{remark}

It should be noted that we only prove that Eq.\ \eqref{eq:type_change_condition} is a sufficient condition for finiteness of the algebra of invariants, but in general not a necessary one. This is due to the fact that in general one does not expect $(\MM_S\otimes B(\HH_R))^{\Ad(U_S\otimes U_R)}$ to be a factor. Hence, even though by the theorem above this algebra is semifinite and the trace constructed in the proof is a semifinite trace, it need not be a unique trace (up to rescaling) on this algebra. Therefore, even when this trace is not in fact finite, that does not mean that no trace on this algebra can be a finite trace. It is therefore conceivable that, even when Eq.\ \eqref{eq:type_change_condition} is not satisfied, the algebra of invariant operators may still be a finite von Neumann algebra. While we do not rule out the possibility of formulating a sufficient and necessary condition of a similar nature to  Eq.\ \eqref{eq:type_change_condition}, we shall leave this open for now.

We remark that under certain simplifying assumptions, Eq.\ \eqref{eq:type_change_condition} does actually become a necessary condition for finiteness of the algebra of invariants. First of all, by \cite{Takesaki1973duality}, if $\sigma$ is a modular action on $\MM_S$ one has that the crossed product $\MM_S\rtimes_{\sigma}\RR$ is a type $\mathrm{II}_\infty$ \textit{factor} if and only if $\MM_S$ is a type $\mathrm{III}_1$ factor (see Prop.\ \ref{prop:type3to2fact}). Second, the tensor product of a type $\mathrm{II}_\infty$ factor with a type $\mathrm{I}$ factor is again a type $\mathrm{II}_\infty$ factor, see \cite[Thm.~11.2.27, Prop.~11.2.21-22]{kadison1997fundamentals}. Third, for $\MM$ a factor and $p\in\MM$ a projection, the algebra $p\MM p$ is also a factor, see \cite[Ch. 2.1.]{dixmier1981VonNeumann}. Up to rescaling, factors admit a unique faithful normal trace. Hence any such trace on a finite factor will be a finite trace. We therefore can conclude the following.
\begin{corollary}
    \label{cor:III1_type_change}
    Let $\MM_S, \HH_S, G, U_S, \omega, U_R, E,\HH_R$ be as in Thm.\ \ref{thm:type_change}. Assume that $\MM_S$ is a type $\mathrm{III}_1$ factor and that the group $H$ is trivial. Then the algebra $(\MM_S\otimes B(\HH_R))^{\Ad (U_S\otimes U_R)}$ is a type $\mathrm{II}_1$ factor if and only if Eq.\ \eqref{eq:type_change_condition} is satisfied.
\end{corollary}
Here we imposed the requirement that $H$ is trivial, as for a factor $\MM$ with continuous action of $H$, the subalgebra $\operatorname{Avg}_H(\MM)$ will generically not be a factor. 

\subsection{A thermal interpretation of the type reduction condition}
\label{sec:thermal_int}
The condition Eq.~\eqref{eq:type_change_condition} can be given an interpretation in terms of thermal properties of the quantum reference frame. Suppose that, for a compactly stabilised QRF $(U_R,E,\HH_R)$ for a group $G=\RR\times H$ as above, the energy spectrum is essentially bounded from below, i.e. there is some $\xi_0\in\RR$ such that $m_{U_R}(\xi)=0$ for $\xi<\xi_0$ almost everywhere. In this case, defining $N_{U_R}:\RR\to[0,\infty]$ by
\begin{equation}
    N_{U_R}(\xi)=\int_{\xi_0}^{\xi} m_{U_R}(\eta)\diff\eta,
\end{equation}
one readily sees on integrating by parts, that 
\begin{equation}
    \int_\RR\exp(-\beta\xi)m_{U_R}(\xi)\diff\xi=\lim_{\xi_1\to\infty}\left(\exp(-\beta\xi_1)N_{U_R}(\xi_1)+\beta\int_{\xi_0}^{\xi_1}\exp(-\beta\xi)N_{U_R}(\xi)\diff\xi\right).
\end{equation}
Therefore, under the assumptions above, Eq.\ \eqref{eq:type_change_condition} in particular implies
\begin{equation}
    \label{eq:spectral_growth_cond}
    \int_{\xi_0}^{\infty}\exp(-\beta\xi)N_{U_R}(\xi)\diff\xi<\infty.
\end{equation}
Loosely speaking, $N_{U_R}(\xi)$ quantifies the degrees of freedom of the quantum reference frame that have an energy below $\xi$. Since this function is increasing in $\xi$, one sees that as a condition, Eq.\ \eqref{eq:spectral_growth_cond} restricts the asymptotic growth of the function $N_{U_R}(\xi)$, i.e. this function should not grow too quickly. In that sense, one can view this condition as similar in spirit to the nuclearity condition appearing in the context of quantum field theory on Minkowski spacetime, see e.g.\ \cite{Buchholz:1986dy}. Roughly speaking, the QFT nuclearity conditions restrict how the number of states that are, in some sense, localised in a given spacetime region and have an energy bounded above by $E$, grow as a function of $E$. One particular consequence of such a nuclearity condition is the existence of positive temperature KMS states on that quantum field theory, see \cite{Buchholz:1988fk}. Another is the split property \cite{Buchholz1987Universal}.

For quantum reference frames as above, one can straightforwardly formulate a KMS condition with respect to the evolution given by $U_R\restriction_\RR$. Note however that due to the existence of a covariant POVM $E$ on $\HH_R$ and that our QRFs are compactly stabilised, the existence of time translation invariant (normal) states is ruled out. This can be seen as follows. Let $\Sigma=\RR\times H/H'$ be the value space of the POVM $E$, consider $X_n=[n,n+1)\times H/H'$ and let $\omega$ a time translation invariant state on the $C^*$ algebra generated by $E(\Bor(\Sigma))$. Then note in particular that
\begin{equation}
    \omega(E(X_{n+1}))=\omega(E(X_n)).
\end{equation}
However, since for each $N\in\mathbb{N}$ one has
\begin{equation}
    \sum_{n=-N}^N E(X_n)\leq 1,
\end{equation}
it follows that $\omega(E(X_n))\leq \frac{1}{2N+1}$. Hence $\omega(E(X_n))=0$. Clearly, such a state cannot extend to a normal state on $B(\HH_R)$ or even on $E(\Bor(\Sigma))''$. Nevertheless, one can use condition \eqref{eq:type_change_condition} to construct a normal \textit{weight} on $B(\HH_R)$ that satisfies the KMS condition on a particular subalgebra. This construction is given below, deferring some of the technical details to Appendix \ref{appx:QRF_thermal_weight}.

Note that any compactly stabilised QRF $(U_R,E,\HH_R)$ for $G=\RR\times H$ with value space $\Sigma=\RR\times H/H'$ defines a principal QRF for $\RR$ given by $(U_R\restriction_\RR,\tilde{E},\HH_R)$, with $\tilde{E}:\Bor(\RR)\to B(\HH_R)$ given by $\tilde{E}(X)=E(X\times H/H')$. Hence without loss of generality we can restrict our discussion to principal QRFs for the group $G=\RR$. For such a reference frame, we first consider a subset of $B(\HH_R)$ consisting of operators that satisfy an integrability condition.
\begin{definition}
\label{def:pos_atl}
    For $(U_R,E,\HH_R)$ a principal QRF for the additive group $\RR$, consider for each $a\in B(\HH_R)$ and $\psi\in\HH_R$ the integral $I_{a,\psi}\in[0,\infty]$ given by
    \begin{equation}
        I_{a,\psi}=\int_\RR \Vert aU_R(t)^*\psi\Vert^2\diff t,
    \end{equation}
    and for fixed $a\in B(\HH_R)$, let $C_a\in [0,\infty]$ be
    \begin{equation}\label{eq:Ca}        C_a=\sqrt{\sup_{\psi\in\HH_R,\,\Vert\psi\Vert=1}I_{a,\psi}},
    \end{equation}
    with the convention $\sqrt{\infty}=\infty$. We now define the set $\mathcal{P}\subset B(\HH_R)$ of \textup{positive $t$-integrable operators} by
    \begin{equation}
        \mathcal{P}=\{a^*a:a\in B(\HH_R),\, C_a<\infty\}.
    \end{equation}
    We refer to the linear span of $\mathcal{P}$ as the \textup{*-algebra of $t$-integrable operators},
    \begin{equation}
        \mathcal{D}:=\left\{\sum_{n=1}^Nc_n x_n:N\in\mathbb{N},\, c_n\in\CC,\,x_n\in\mathcal{P}\right\}.
    \end{equation}
\end{definition}
The operators $x\in\mathcal{P}$ have the property that for $\psi\in\HH_R$, the following integral is well defined, positive and admits the upper bound
\begin{equation}
    \int_\RR \langle\psi ,U_R(t)xU_R(t)^*\psi\rangle \,\diff t\leq C_{\sqrt{x}}^2\Vert\psi\Vert^2. 
\end{equation}
This means that as a weak integral
\begin{equation}
    I(x):=\int_\RR  U_R(t)xU_R(t)^*\,\diff t\in (B(\HH_R)^+)^{\Ad U_R},
\end{equation}
hence motivating the term $t$-integrable. In particular, the operators 
\begin{equation}
    E(f)=\int_\RR f(t)\,\diff E(t),
\end{equation}
are positive $t$-integrable for $f\in L^1(\RR)\cap L^\infty(\RR)^+$, i.e.\ a positive bounded $L^1$ function, see Lem.\ \ref{lem:pos_atl}. By linearity this implies that $E(f)\in\mathcal{D}$ for $f\in L^1(\RR)\cap L^\infty(\RR)$.

That $\mathcal{D}$ is indeed a *-algebra, is also shown in Lem.\ \ref{lem:pos_atl}. This is a consequence of the fact that $\mathcal{P}$ is a hereditary convex cone. This means that for each $x\in\mathcal{P}$ and $y\in B(\HH_R)^+$ with $x\geq y$ one has $y\in \mathcal{P}$ and that $\mathcal{P}$ is closed under convex sums and rescaling, see \cite[Lem.~VII.1.2]{takesaki2002}. In particular, this allows one to linearly extend the $t$-integral 
\begin{equation}
    I:\mathcal{D}\to B(\HH_R)^{\Ad U_R}.
\end{equation}

We now construct a weight on $B(\HH_R)$ (see Def.\ \ref{def:weight_trace}). 
\begin{definition}
    \label{def:beta_weight}
    Given a principal QRF $(U_R,E,\HH_R)$ for the group $\RR$ with $W:\HH_R\to L^2(\RR)\otimes\KK$ as in Thm.\ \ref{thm:qref_embed}, let $\{\psi_n\}_{n=1}^\infty$ be an orthonormal basis of $\KK$ (with $\psi_n=0$ for $n>\dim(\KK)$ if $\KK$ finite dimensional). Furthermore, for $K\subset\RR$ compact and $z\in\CC$, let $f_K^{(z)}\subset L^2(\RR)$ be given by
    \begin{equation}
        f_K^{(z)}(\xi)=\chi_K(\xi)\exp(-z\xi/2).
    \end{equation}
    For $\beta>0$, we define the weight $\varphi_\beta:B(\HH_R)^+\to [0,\infty]$ by
    \begin{equation}
        \varphi_\beta(x)=\sup_{K,N}\sum_{n=1}^N\int_\RR \langle W^*(\mathcal{F}^*f_K^{(\beta)}\otimes \psi_n),U_\RR(t)xU_\RR(t)^*W^*(\mathcal{F}^*f_K^{(\beta)}\otimes \psi_n)\rangle\diff t,
    \end{equation}
    where the supremum is taken over all $N\in\mathbb{N}$ and $K\subset \RR$.
\end{definition}
Using the fact that $U_\RR(t)^*W^*(\mathcal{F}^*\otimes 1_{\KK})=W^*(\mathcal{F}^*T_{\exp(-it.)}\otimes 1_{\KK})$, and defining $\Psi_{K,n}^{(z)}=(\mathcal{F}^*f_K^{(z)}\otimes \psi_n)$, we can alternatively write
\begin{equation}
        \varphi_\beta(a)=\sup_{K,N}\sum_{n=1}^N\int_\RR \langle W^*\Psi_{K,n}^{(\beta+2it)},aW^*\Psi_{K,n}^{(\beta+2it)}\rangle\diff t.
    \end{equation}
It is not immediately obvious that $\varphi_\beta$ is a normal weight (see Def.\ \ref{def:weight_trace}), as one typically may not interchange integration and suprema over non-countable sets. Using however that $\langle \psi, U_R(t)xU_R(t)^*\psi\rangle$ is a continuous function of $t$, normality is proven in Lem.\ \ref{lem:weight_normal}.

For $x\in \mathcal{P}$, the time integrated $I(x)$ is invariant under $\Ad U_R$. Using this, we see that on $\mathcal{P}$ the weight $\varphi_\beta$ relates to the trace constructed in the proof of Thm.\ \ref{thm:type_change} on the algebra $(\MM_S\otimes B(\HH_R))^{\sigma_\beta\otimes \Ad U_R}$ for $\sigma_\beta(t,.)=\sigma(-t/\beta,.)$ with $\sigma$ given by the modular action on a von Neumann algebra $\MM_S$. Namely, $\varphi_\beta$ is given by
\begin{equation}
    \varphi_\beta(x)=\tau(1_{\MM_S}\otimes I(x)).
\end{equation}
It is thus clear that $\varphi_\beta$ takes finite values on $\mathcal{P}$ provided that Eq.\ \eqref{eq:type_change_condition} holds, as then $\tau$ is a finite trace. That this implication in fact holds both ways, is proven in Lem.\ \ref{lem:weight_type_change}. Hence, if Eq.\ \eqref{eq:type_change_condition} holds, one can linearly extend $\varphi_\beta$ to $\mathcal{D}$. We shall denote this linear map by the same symbol as the weight, i.e.\  $\varphi_\beta:\mathcal{D}\to\CC$.

We now investigate if $\varphi_\beta$ indeed satisfies a KMS property. Note that the set $\mathcal{D}$ admits a natural norm $\Vert.\Vert_{\mathcal{D}}:\mathcal{D}\to\RR^+$ given by
\begin{equation}
    \Vert x\Vert_\mathcal{D}=\sup_{\psi\in\HH_R,\,\Vert\psi\Vert=1}\int_\RR\vert\langle \psi,U_R(t)xU_R(t)^*\psi\rangle\vert\diff t,
\end{equation}
where for $a^*a\in \mathcal{P}$ we have $\Vert a^*a\Vert_\mathcal{D}=\Vert I(a^*a)\Vert=C_{a}^2$. Subadditivity and absolute homogeneity of $\Vert.\Vert_{\mathcal{D}}$ are easily verified. To show positive definiteness of this map, we use that $U_R$ is a strongly continuous map, therefore $t\mapsto \vert\langle \psi,U_R(t)xU_R(t)^*\psi\rangle\vert$ is a continuous function. Hence $\Vert x\Vert_\mathcal{D}=0$ if and only if $\langle \psi,x\psi\rangle=0$ for all $\psi\in\HH_R$ and hence $x=0$. Note furthermore that $\Vert x\Vert_\mathcal{D}=\Vert x^*\Vert_\mathcal{D}$.
One now sees by a similar calculation as in Lem.\ \ref{lem:weight_type_change} that for $x\in\mathcal{D}$, that 
\begin{equation}
    \vert \varphi_\beta(x)\vert\leq \Vert x\Vert_{\mathcal{D}}\int_\RR \exp(-\beta\xi)m_{U_R}(\xi)\diff\xi.
\end{equation}
Given $x,y\in\mathcal{D}$, we have for $t\in \RR$ that $yU_R(t)xU_R(t)^*\in \mathcal{D}$. In fact we find
\begin{align}
    \Vert yU_R(t)xU_R(t)^*\Vert_\mathcal{D}=&\sup_{\Vert\psi\Vert=1}\int_\RR\vert \langle\psi,U(s)yU_R(t)xU_R(t+s)^*,\psi\rangle\vert \diff s\nonumber\\
    \leq& \sup_{\Vert\psi\Vert=1}\int_\RR\Vert y^*U_R(s)^*\psi\Vert\Vert xU_R(t+s)^*\psi\Vert \diff s\nonumber\\
    \leq& \sup_{\Vert\psi\Vert=1}\sqrt{\int_\RR\Vert y^*U_R(s)^*\psi\Vert^2\diff s}\sqrt{\int_\RR\Vert xU_R(t+s)^*\psi\Vert^2 \diff s}\nonumber\\
    \leq& C_{y^*}C_{x},
\end{align}
It follows that $t\mapsto \varphi_\beta(yU_R(t)xU_R(t)^*)$ is a bounded function. Of course a similar result holds for $t\mapsto \varphi_\beta(U_R(t)xU_R(t)^*y)$. Ideally, we would find a function on the strip $S_\beta$ that takes these functions as boundary values, but we will find it necessary to restrict to a subalgebra to achieve this. We write
\begin{align}
    \varphi_\beta(yU_R(t)xU_R(t)^*)=&\lim_{K\uparrow\RR}\sum_{n=1}^\infty \int_\RR \langle W^*\Psi_{K,n}^{(\beta+2is)},yU_R(t)xU_R(t)^* W^*\Psi_{K,n}^{(\beta+2is)}\rangle\diff s\nonumber\\
    =&\lim_{K\uparrow\RR}\sum_{n=1}^\infty \int_\RR \langle W^*\Psi_{K,n}^{(\beta+2is)},yU_R\left(t\right)xW^*\Psi_{K,n}^{(\beta+2is+2it)}\rangle\diff s\nonumber\\
    =&\lim_{K\uparrow\RR}\sum_{n=1}^\infty \int_\RR \langle W^*\Psi_{K,n}^{(\beta-it+2is)},yU_R\left(\frac{t}{2}\right)U_R\left(\frac{t}{2}\right)xW^*\Psi_{K,n}^{(\beta+it+2is)}\rangle\diff s,
\end{align}
where we have changed variables from $s$ to $s+t/2$ in the last line. To the expression above, we can apply Eq.\ \eqref{eq:time_trace_completeness} of Lem.\ \ref{lem:time_trace_completeness} and find that
\begin{align}
    \varphi_\beta(yU_R(t)xU_R(t)^*)=\frac{1}{2\pi}\lim_{K\uparrow\RR}\sum_{n=1}^\infty \lim_{K'\uparrow\RR}\sum_{n'=1}^\infty\int_{\RR^2}& \langle W^*\Psi_{K,n}^{(\beta-it+2is)},yW^*\Psi_{K',n'}^{(2is'-it)}\rangle\nonumber\\
    &\times\langle W^*\Psi_{K',n'}^{(2is'+it)},xW^*\Psi_{K,n}^{(\beta+it+2is)}\rangle\diff s\diff s'.
\end{align}
At this point one is tempted to define a candidate holomorphic interpolating function $F_{x,y}$ on the usual complex strip $S_\beta$ by
\begin{align}
    F_{x,y}(z)=\frac{1}{2\pi}\lim_{K\uparrow\RR}\sum_{n=1}^\infty \lim_{K'\uparrow\RR}\sum_{n'=1}^\infty\int_{\RR^2}& \langle W^*\Psi_{K,n}^{(\beta-i\overline{z}+2is)},yW^*\Psi_{K',n'}^{(2is'-iz)}\rangle\nonumber\\
    &\times\langle W^*\Psi_{K',n'}^{(2is'+i\overline{z})},xW^*\Psi_{K,n}^{(\beta+iz+2is)}\rangle\diff s\diff s',
\end{align}
in which  the complex conjugate of $z$ appears at entries where the inner product is antilinear. Certainly for $t\in \RR$ $F_{x,y}(t)=\varphi_\beta(yU_R(t)xU_R(t)^*)$ is well defined. If we consider $z=t+i\beta$, we find
\begin{align}
    F_{x,y}(t+i\beta)=\frac{1}{2\pi}\lim_{K\uparrow\RR}\sum_{n=1}^\infty \lim_{K'\uparrow\RR}\sum_{n'=1}^\infty\int_{\RR^2}& \langle W^*\Psi_{K,n}^{(2is-it)},yW^*\Psi_{K',n'}^{(\beta+2is'-it)}\rangle\nonumber\\
    &\times\langle W^*\Psi_{K',n'}^{(\beta+2is'+it)},xW^*\Psi_{K,n}^{(it+2is)}\rangle\diff s\diff s'.
\end{align}
That the limiting operations $\lim_{K\uparrow\RR}\sum_{n=1}^\infty$ and $\lim_{K'\uparrow\RR}\sum_{n'=1}^\infty$ can be interchanged, is shown in Lem.\ \ref{lem:F_int_lim}. It now follows that $F_{x,y}(t+i\beta)=\varphi_\beta(x U_R(t)^*yU_R(t))=\varphi_\beta(U_R(t)x U_R(t)^*y)$.

From the arguments above, it seems that the function $F_{x,y}$ is a promising candidate to interpolate between $\varphi_\beta(yU_R(t)xU_R(t)^*)$ and $\varphi_\beta(U_R(t)xU_R(t)^*y)$. The problematic aspect however, is that for $z$ in the interior of $S_\beta$, the expression for $F_{x,y}(z)$ in general does not converge, let alone define an analytic function. To overcome this, we follow a similar strategy as detailed in \cite{Buchholz:1988fk} and introduce an algebra of analytic operators on which we will show that the weight $\varphi_\beta$ satisfies the KMS condition.
\begin{definition}
    \label{def:ana_pos_atl}
    For $(U_R,E,\HH_R)$ a QRF with $t$-integrable operators $\mathcal{D}\subset B(\HH_R)$ as in Def.\ \ref{def:pos_atl}, let $\mathcal{D}_1\subset B(\HH_R)$ be the *-algebra generated by operators of the form
    \begin{equation}
        \label{eq:op_smear}
        x(f)=\int_{\RR}f(t)U_R(t)xU_R(t)^*\diff t,
    \end{equation}
    where $f\in L^1(\RR)$ is a Gaussian, i.e. $f(t)=\gamma\exp(-\alpha (t-z)^2)$ for some $\alpha>0$ and $\gamma,z\in\CC$, and $x\in \mathcal{D}$.
\end{definition}
Some convenient properties of this algebra are listed in Lem.\ \ref{lem:ana_pos_atl}, some of which are analogous to those found in \cite[Lem.\ 3.2]{Buchholz:1988fk}. In particular, $\mathcal{D}_1$ is weakly and strongly dense in $B(\HH_R)$ and contains $E(f)$ for Gaussian $f\in L^\infty(\RR)$. The latter guarantees that $\varphi_\beta$ is nonzero on $\mathcal{D}_1$, as for $f\in L^1(\RR)\cap L^\infty(\RR)^+$ it is shown in the proof of Lem.\ \ref{lem:weight_type_change} that
\begin{equation}
    \varphi_\beta(E(f))=\Vert f\Vert_{L^1(\RR)}\int_\RR\exp(-\beta\xi)m_{U_R}(\xi).
\end{equation}
Moreover, it is shown in Lem.\ \ref{lem:ana_pos_atl} that for $x\in\mathcal{D}_1$, the action $t\mapsto U_R(t)x U_R(t)^*$ extends to a holomorphic function $z\mapsto \alpha(z,x)$ for $z\in\CC$.

We now consider $x,y\in\mathcal{D}_1$. We define for $z\in S_\beta$
\begin{equation}
    F_{x,y}(z)=\varphi_\beta(y\alpha(z,x)).
\end{equation}
It is shown in Lem.\ \ref{lem:KMS_F_int} that this function is continuous and bounded, and analytic on the interior on $S_\beta$. Furthermore, it satisfies the boundary conditions
\begin{equation}
    F_{x,y}(t)=\varphi_\beta(y\alpha(t,x)),\, F_{x,y}(t+i\beta)=\varphi_\beta(\alpha(t,x)y).
\end{equation}
This leads us to the following conclusion.

\begin{theorem}\label{thm:KMS_weight}
    Let $(U_R,E,\HH_R)$ be a quantum reference frame as in Def.\ \ref{def:pos_atl}, $\varphi_\beta$ the normal weight as in Def.\ \ref{def:beta_weight} and $\mathcal{D}_1\subset B(\HH_R)$ as in Def.\ \ref{def:ana_pos_atl} a weakly/strongly dense *-subalgebra, then $\varphi_\beta$ takes finite values on $\mathcal{D}_1$ if and only if Eq.\ \eqref{eq:type_change_condition} holds. It then further satisfies the KMS condition on $\mathcal{D}_1$ at inverse temperature $\beta$ with respect to $\Ad U_R$.
\end{theorem}

\section{Application to the de Sitter static patch}
\label{sec: deSitter}

We will now show how our general results apply in a physical example of recent interest: quantum fields on de Sitter spacetime. A particular scenario of local algebras' type reduction in the presence of an observer moving along a geodesic in the static patch of de Sitter has been analysed in \cite{chandrasekaran2023algebra}. We will demonstrate how this scenario arises as a special case in our general framework.

Recall that $n$-dimensional de Sitter spacetime is defined as the  hyperboloid 
\begin{equation}
    \{(x^0,...,x^n)\in \RR^{n+1}:(x^0)^2-\sum_{j=1}^n (x^j)^2 =-\rho^2\}\,,
\end{equation}
for some $\rho>0$, embedded in $(n+1)$-dimensional Minkowski spacetime, equipped with the induced metric. 

Up to isometries, the observer's geodesic can be parametrised in proper time $\tau$ by 
\begin{equation}
x^0=\rho \sinh (\tau/\rho)\,,\quad x^1=\rho \cosh (\tau/\rho)\,,\quad x^2=\dots=x^n=0\,.
\end{equation}
The associated static patch is defined as the causal hull of this geodesic and it is the intersection of the hyperboloid with a wedge of Minkowski space-time, where $\vert x^0\vert <x^1$. We regard the static patch as spacetime in its own right, and denote it as $M$. The static patch contains all the points that the observer could send and receive messages to and from. Hence, the static patch constitutes the maximal experimental domain for the observer, in which they can both coordinate an experiment, as well as receive the results. Crucially, these experiments may be conducted far away from the geodesic, possibly by other agents. We describe such an experiment by the relativistic quantum measurement framework given in Sec.\ \ref{sec: measurement in QFT}. Here we recall that $M$ has a group of (time-orientation and orientation preserving) isometries $G=\RR\times \mathrm{SO}(n-1)$. Its action, which preserves the observer's trajectory, is given by $\mathrm{SO}(n-1)$-rotations in $(x^2,\dots,x^n)$ coordinates and boosts in $(x^0,x^1)$ coordinates paramatrised by $t\in\RR$, i.e.:
\begin{equation}
    (x^0,x^1)\mapsto (x^0 \cosh(t/\rho)+x^1 \sinh(t/\rho),x^0 \sinh(t/\rho)+x^1 \cosh(t/\rho))\,.
\end{equation}
$M$ is static with respect to this $\RR$-action, which defines a time-like flow of spacetime points that we shall refer to as the \textit{static flow}. In order to distinguish between experiments related by a $G$-symmetry, the observer will use a compactly stabilised quantum reference frame as defined in Sec.\ \ref{Sec:QRFs}, see Sec.\ \ref{sec:qrf_measurement_general} for a general discussion. A full description of compactly stabilised QRFs applicable for the group $G=\RR\times \mathrm{SO}(n-1)$ is given in Sec.\ \ref{sec:RHcase}. In particular in Def.\ \ref{def:spec_mult} we defined for such a QRF $(U_R,E,\HH_R)$ the spectral multiplicity function $m_{U_R}:\RR\to [0,\infty]$.

 We describe a quantum field theory on $M$ in terms of a net of local algebras $\mathscr{A}(M;N)$ on $M$ with corresponding global algebra $\mathscr{A}(M)$, for simplicity we assume that these are $C^*$-algebras, see Sec.\ \ref{sec:AQFTbackground}. Possibly, this theory arises as a subnet of a theory on the full de Sitter spacetime. We assume
  that there exists a faithful state $\omega$   on $\Af(M)$, which induces a thermal state on
  $\MM_S\equiv \pi(\Af(M))''$ of inverse temperature $\beta\in (0,\infty)$ (i.e.\ KMS with respect to the static flow) which is invariant under the full action of $G$. 
Performing the usual completion in the associated GNS representation $\pi$, we obtain the von Neumann algebra
\begin{equation}
    \MM_S\equiv \pi(\Af(M))''\,.
\end{equation}
It is assumed that the action of $G$ on this algebra can be described in terms of a strongly continuous unitary representation $U_S$ on the GNS Hilbert space. 

  In Sec.\ \ref{sec:qrf_measurement_general} we gave a description of operators encoding outcomes of experiments probing $\Af(M)$ made relative to a QRF $(U_R,E,\HH_R)$. Provided that the quantum field is in a state in the folium of $\omega$, these operators are elements of
\begin{equation}
    (\MM_S\otimes B(\HH_R))^{\Ad U_S\otimes U_R}.
\end{equation}
By Thm.\ \ref{thm:type_change}, this algebra is semifinite, meaning that it admits a semifinite faithful normal trace. Moreover, if the spectral multiplicity of the frame satisfies
\begin{equation}
    \int_\RR \exp(-\beta\xi)m_{U_R}(\xi)\diff \xi<\infty,
\end{equation}
then this algebra is finite, meaning that there is a faithful normal trace on this algebra taking finite values on all elements.

Below we shall illustrate this general result by way of various examples of compactly stabilised quantum reference frames for which the resulting algebra of invariants is either semifinite or finite. We first consider the example most closely related to the discussion in \cite{chandrasekaran2023algebra}.

\paragraph{Example 1 (trivial $\textnormal{SO}(n-1)$ action)} A simple (yet relevant) class of compactly stabilised quantum reference frames for $G=\RR\times \mathrm{SO}(n-1)$ is given as follows. In each case, the value space is $\Sigma=\RR$, while the Hilbert space is $\HH_R=p_V L^2(\RR)$ for a projection
\begin{equation}
    p_V=\mathcal{F}^*T_{\chi_V}\mathcal{F},
\end{equation}
where $V\in \Bor(\RR)$ is a measurable set of nonzero Lebesgue measure. The group $G$ acts on $\HH_R$ via $U_R(t,h)=\lambda_\RR(t)$, the left translation action for $\RR$, so $\mathrm{SO}(n-1)$ acts trivially. Meanwhile, the POVM $E:\Bor(\RR)\to B(\HH_R)$ is defined by
\begin{equation}
    E(X)=p_V T_{\chi_X}.
\end{equation}
Note that indeed $U_R(t,h)E(X)U_R(t,h)^*=E(X+t)$. The spectral multiplicity function of this quantum reference frame is simply given by
\begin{equation}
    m_{U_R}(\xi)=\chi_V(\xi).
\end{equation} 
Since the action of $\mathrm{SO}(n-1)$ on $B(\HH_R)$ is trivial, one can now easily describe the algebra of invariant observables in terms of a crossed product algebra, see e.g.\ Sec.\ \ref{sec:RHinv_obs}.
\begin{equation}
    (\MM_S\otimes B(\HH_R))^{\Ad U_S\otimes U_R}=(1\otimes p_V)\left(\MM_S^{\Ad U_S\restriction_{SO(n-1)}}\rtimes_{\Ad U_S\restriction_\RR}\RR\right)(1\otimes p_V).
\end{equation}
Since $\MM_S$ admits a faithful normal KMS state w.r.t.\ $\Ad U_R\restriction_\RR$, this action matches (up to reparameterisation) with a modular action on $\MM_S$, i.e.\ \begin{equation}\label{eq:accountancy}
    \Ad U_R\restriction_\RR(t)=\sigma\left(-\frac{t}{\beta}\right).
\end{equation} As $\MM_S^{\Ad U_S\restriction_{\mathrm{SO}(n-1)}}\subset \MM_S$ and the unitary actions $U_S\restriction_{\RR}$ and $U_S\restriction_{\mathrm{SO}(n-1)}
$ commute, $\sigma$ also restricts to a modular action on $\MM_S^{\Ad U_S\restriction_{\mathrm{SO}(n-1)}}$. We apply Prop.\ \ref{prop:modular_trace}, to find that \begin{equation}
    \MM_S^{\Ad U_S\restriction_{\mathrm{SO}(n-1)}}\rtimes_{\Ad U_S\restriction_\RR}\RR
\end{equation} is a semifinite von Neumann algebra (see Appendix \ref{appx:factor_types} for a definition). Hence this also holds for its compression $(\MM_S\otimes B(\HH_R))^{\Ad U_S\otimes U_R}$. Moreover, by Eq.\ \eqref{eq:finite_trace}, this algebra is finite if
\begin{equation}
    \label{eq:simple_type_change}
    \int_V\exp(-\xi\beta)\diff \xi<\infty,
\end{equation}
Where we have accounted for the reparameterisation relating $\Ad U_S\restriction_\RR$ to a modular action via Eq.~\eqref{eq:accountancy}.
Eq.~\eqref{eq:simple_type_change} holds in particular if $V$ is semibounded from below, $\inf V>-\infty$. This constitutes a special case of our general result given in Thm.\ \ref{thm:type_change}. 

The example above relates quite explicitly to the discussion in \cite{chandrasekaran2023algebra}, see also Example 3 in Sec.~\ref{sec:qrf examples}. Even though the connection with quantum reference frames is not made explicit in~\cite{chandrasekaran2023algebra}, the quantum systems that are introduced to engineer an algebra of invariant observables are in fact QRFs of the form given above, with $V=[0,\infty)$. One noteable difference here is that~\cite{chandrasekaran2023algebra} only considers the algebra
\begin{equation}
    (\MM_S\otimes B(\HH_R))^{\Ad U_S\restriction_{\RR}\otimes U_R\restriction_{\RR}}=(1\otimes p_V)\left(\MM_S\rtimes_{\Ad U_S\restriction_\RR}\RR\right)(1\otimes p_V),
\end{equation}
i.e.\ only the time-translation invariance on the de Sitter static patch is taken into account. In this simplified setting we can apply Cor.\ \ref{cor:III1_type_change} if we furthermore assume that $\MM_S$ is a type III${}_1$ factor. This is known to be the case for various quantum field theories provided that the net $\mathcal{A}(M;N)$ is a subnet of a theory on full de Sitter and that the KMS state on $\mathcal{A}(M)$ extends to a `physically reasonable' state on this larger spacetime, see e.g.\ \cite[Thm.\ 3.6.g]{Verch:1996wv}. Under this assumption, $(\MM_S\otimes B(\HH_R))^{\Ad U_S\restriction_{\RR}\otimes U_R\restriction_{\RR}}$ is finite (and indeed a type $\textnormal{II}_1$ factor) if and only if Eq.\ \eqref{eq:simple_type_change} holds. As Eq.\ \eqref{eq:simple_type_change} certainly holds for the $V=[0,\infty)$ made in \cite{chandrasekaran2023algebra}, our results for this case are in clear agreement.

\paragraph{Example 2 (nontrivial $\textnormal{SO}(3)$ action in four dimensions)} For our second class of examples, we set for concreteness $n=4$. We consider a class of compactly stabilised quantum reference frames for $G=\RR\times \mathrm{SO}(3)$, where in this case we let the value space be $\Sigma=G$. We set $\HH_R=p(L^2(\RR)\otimes L^2(\mathrm{SO}(3)))$ for some projection $p\in \lambda_\RR(\RR)'\otimes\lambda_{\mathrm{SO}(3)}(\mathrm{SO}(3))'$. Via Cor.\ \ref{cor:QRF_compact_subrep} this projection defines a natural (principal) quantum reference frame with the unitary representation of $G$ given by 
\begin{equation}
    U_R(t,h)=\left(\lambda_\RR(t)\otimes\lambda_{\mathrm{SO}(3)}(h)\right)\restriction_{\HH_R}.
\end{equation} One could similarly as in Example 1 make the choice $p=p_V\otimes 1_{L^2(\mathrm{SO}(3))}$, for $V\in\Bor(\RR)$, which yields for the spectral multiplicity 
\begin{equation}
    m_{U_R}(\xi)=\chi_V(\xi)\cdot\infty,
\end{equation}
due to the infinite dimension of $L^2(\textnormal{SO}(3))$.
It is clear that for $V$ of nonzero measure Eq.\ \eqref{eq:type_change_condition} is not satisfied. Hence for this choice of $p$, the algebra of invariant operators
\begin{multline}
    (\MM_S\otimes B(\HH_R))^{\Ad U_S\otimes U_R}=\\\left((1_{\MM_S}\otimes p_V)(\MM_S\rtimes_{\Ad U_S\restriction_\RR}\RR)(1_{\MM_S}\otimes p_V)\otimes B(L^2(SO(3)))\right)^{\Ad U_S\restriction_{\mathrm{SO}(3)}\otimes 1_{L^2(\RR)}\otimes \lambda_{\mathrm{SO}(3)}}
\end{multline}
is generally not expected to be finite. Note however that in this example we have only shown that Eq.\ \eqref{eq:type_change_condition} is a sufficient condition for finiteness. Depending on the details of the action $\Ad U_S\restriction_{\mathrm{SO}(3)}$, the algebra $(\MM_S\otimes B(\HH_R))^{\Ad U_S\otimes U_R}$ may still be finite. 

Arguably the choice $p=p_V\otimes 1_{L^2(\mathrm{SO}(3))}$ is not very physical, which can be seen as follows. Note that the left regular action on $L^2(\mathrm{SO}(3))$ decomposes into irreducible representations due to the Peter-Weyl theorem \cite[Thm.\ 1.12]{Knapp_rep_semisimple}
\begin{equation}
    L^2(\mathrm{SO}(3))\cong \overline{\bigoplus_{l=0}^{\infty}\HH_{l}^{\oplus(2l+1)}},
\end{equation}
where $\HH_{l}$ is the Hilbert space associated with the spin-$l$ irreducible representation of $\mathrm{SO}(3)$. The spaces $\HH_{l}^{\oplus(2l+1)}\cong \HH_l\otimes\HH_l$, seen as closed subspaces of $L^2(\mathrm{SO}(3))$ are in particular $(2l+1)^2$-dimensional eigenspaces for the Casimir element $\mathbf{L}^2$ of the Lie algebra $\mathfrak{so}(3)$. Here $\mathbf{L}^2$ acts as an (unbounded) linear operator on $L^2(\mathrm{SO}(3))$ with eigenvalues \begin{equation}
    \mathbf{L}^2\psi_l=l(l+1)\psi_{l},
\end{equation} 
for $\psi_l\in \HH_{l}^{\oplus(2l+1)}$. In typical models from quantum mechanics that admit an $\mathrm{SO}(3)$ symmetry, the Casimir element $\mathbf{L}^2$ is proportional to the squared angular momentum and typically contributes to the Hamiltonian of that system. We can model this in our quantum reference frame by setting $\HH_R=p (L^2(\RR)\otimes L^2(\mathrm{SO}(3)))$ with
\begin{equation}
    p=\sum_{l=0}^\infty p_{V_l}\otimes \tilde{p}_l,
\end{equation}
for $V_l\in\Bor(\RR)$ and where $\tilde{p}_l\in \lambda_{SO(3)}(SO(3))'$ projects $L^2(\mathrm{SO}(3))$ onto its closed subspace $\HH_{l}^{\oplus(2l+1)}$. This projection defines a principal quantum reference frame for $\RR\times \mathrm{SO}(3)$ with spectral multiplicity
\begin{equation}
    m_{U_R}(\xi)=\sum_{l=0}^\infty (2l+1)^2\chi_{V_l}(\xi).
\end{equation}
One can for instance set 
$V_l=[E_l,\infty)$ for some set of numbers $E_l\in \RR$. In this case, Eq.\ \eqref{eq:type_change_condition} holds if and only if 
\begin{equation}
    \sum_{l=0}^\infty (2l+1)^2\exp(-\beta E_l)<\infty.
\end{equation}
Note that $\HH_R\cong p_{[0,\infty)}L^2(\RR)\otimes L^2(SO(3))$ for an isomorphism under which \begin{equation}
    U_R\restriction_{\RR}(t)\cong\lambda_\RR(t)\otimes \sum_{l=0}^\infty \exp(i E_l t)\tilde{p}_l.
\end{equation}
Due to this identification, one can view $\hat{H}_{\mathrm{SO}(3)}:=\sum_{l=0}^\infty E_l\tilde{p}_l$ as the Hamiltonian of the subsystem associated with the Hilbert space $L^2(\mathrm{SO}(3))$. One now sees that Eq.\ \eqref{eq:type_change_condition} holds precisely when the operator $\exp(-\beta \hat{H}_{\mathrm{SO}(3)})\in B(L^2(\mathrm{SO}(3)))^+$ is trace-class, i.e.\ when using the standard trace $\tau$ on type $\mathrm{I}_\infty$ factors, we have
\begin{equation}
    \tau(\exp(-\beta \hat{H}_{\mathrm{SO}(3)}))=\sum_{l=0}^\infty (2l+1)^2\exp(-\beta E_l)<\infty.
\end{equation}
In other words, for this QRF Eq.\ \eqref{eq:type_change_condition} is equivalent to the existence of a Gibbs state on $B(L^2(\mathrm{SO}(3)))$ with respect to the time-translation generated by the Hamiltonian $\hat{H}_{\mathrm{SO}(3)}$, defined by the density matrix
\begin{equation}
    \rho_\beta=\frac{\exp(-\beta \hat{H}_{\mathrm{SO}(3)})}{\tau(\exp(-\beta \hat{H}_{\mathrm{SO}(3)}))}\in B(L^2(\mathrm{SO}(3))).
\end{equation} 
Given that $\rho_\beta$ is trace class, one can use the resulting Gibbs state on $B(L^2(\mathrm{SO}(3)))$ to define a KMS weight $\varphi_\beta$ on the full QRF as in Sec.\ \ref{sec:thermal_int}. Namely, under the isomorphism $\HH_R\cong p_{[0,\infty)}L^2(\RR)\otimes L^2(SO(3))$ as above we set
\begin{equation}
    \varphi_\beta\cong \varphi_\beta^{(1)}\otimes \tau(\rho_\beta\,\cdot\,),
\end{equation}
where $\varphi_\beta^{(1)}$ is the weight on $B(p_{[0,\infty)}L^2(\RR))$ such that for $a\in B(p_{[0,\infty)}L^2(\RR))^+$
\begin{equation}
    \varphi_\beta^{(1)}(a)=\sup_{K\subset\RR\text{ compact}}\int_\RR\langle\mathcal{F}^*f_K,\lambda_{\RR}(t)a\lambda_{\RR}(t)^* \mathcal{F}^*f_K\rangle\diff t,
\end{equation}
with $f_K(\xi)=\chi_K(\xi)\exp(-\frac{\beta}{2}\xi)$. Here we note that by the analysis described in Sec.\ \ref{sec:thermal_int}, $\varphi_\beta^{(1)}$ defines a normal KMS weight on $B(p_{[0,\infty)}L^2(\RR))$ with respect to the action $\Ad\lambda_{\RR}$ (in the sense that the KMS condition is fulfilled on a dense subalgebra of $B(p_{[0,\infty)}L^2(\RR))$). This gives a concrete example of the relation between Eq.\ \eqref{eq:type_change_condition} and thermal properties of the quantum reference frame that were studied in more generality in Sec.\ \ref{sec:thermal_int}.  

\section{Conclusion}

In this paper we have developed an operational framework that combines the theory of quantum reference frames with measurements in quantum field theory. We have demonstrated how measurement schemes related by a spacetime symmetry transformation can be distinguished by using a quantum system -- a QRF -- as a reference. Through the analysis of a suitable class of QRFs, we have proved a general connection between physical observables and crossed product algebras. Applying these results to scenarios of physical interest, we have shown how, through the theory of modular crossed products, this connection has direct implications for the type structure of algebras of physical observables.

Taking a broader perspective, our results emphasise that physical observables in QFT should be understood operationally; this perspective has been systematically developed in~\cite{FewVer_QFLM:2018},
which focussed attention on measurement schemes in which observables are measured through physical interactions with probes. In the presence of symmetries, 
local measurements are only meaningful relative to a physical reference frame.  As illustrated by our analysis, the structure of the algebra of physical observables can differ significantly from local algebras that one typically considers in algebraic quantum field theory. We contrast our operational viewpoint with~\cite{chandrasekaran2023algebra}, in which a type reduction phenomenon was ascribed to the need to `gravitationally dress' the observables of QFT; in our approach type reduction arises from the inclusion of a reference frame and is not directly gravitational in origin.

One question that could be raised is whether or not the results of this paper and~\cite{chandrasekaran2023algebra} indicate that the paradigm of local type III algebras should be abandoned entirely in QFT. We do not believe that the results here provide sufficient evidence for such a view, because they depend on a number of special features: namely, the geometric symmetries and the coincidence of the modular flow with time translations, as well as the properties of the QRF that is employed. General spacetime regions do not have nontrivial isometries, and geometric modular action is rare and fragile. Nevertheless, we emphasise that our reservations concern generic quantum field theory in curved spacetimes, and it could well be that type reduction occurs more generally in specific models, especially those with extended symmetry groups beyond background isometries. In particular, our results cannot cast direct light on the local algebras to be expected in a quantum theory of gravity.

Nonetheless, it is worth commenting that recent results hinting at semifinite physical subalgebras in quantum gravity, such as presented in \cite{witten2022gravity,jensenGeneralizedEntropyGeneral2023,kudler-flamGeneralizedBlackHole2024}, rely similarly on modular crossed products. In these works crossed product algebras are used in the analysis of diffeomorphism invariant observables in (effective) quantum gravity on some given background. The background isometries considered in this paper can then be seen as a subset of these diffeomorphisms. It is assumed in \cite{jensenGeneralizedEntropyGeneral2023} that their larger group of symmetries on general backgrounds includes a subgroup implemented by modular action. Although we view this as an interesting possibility, we are not aware of a concrete argument supporting this assumption. We hope that further investigations into the relation between quantum reference frames, physical subsystems and diffeomorphism invariance may provide insights into this issue, as well as the potential ubiquity of type reduction in quantum gravity.

\backmatter

%
%
%

\bmhead{Acknowledgements}

It is a pleasure to thank Eli Hawkins, Atsushi Higuchi, Roberto Longo, and Rainer Verch for useful comments at various stages of this work. \\

We are aware of a concurrent independent development by J. De Vuyst, S. Eccles, P. H\"{o}hn and J. Kirklin of a treatment of the ideas in \cite{chandrasekaran2023algebra} from a QRF perspective. A preliminary version of their work \cite{devuystGravitationalEntropyObserverdependent2024} has subsequently appeared, and we hope to compare these proposals with ours elsewhere. We also note the subsequent work by other authors~\cite{aliahmadSemifiniteNeumannAlgebras2024}, analysing type reduction for crossed product algebras involving more general groups of automorphisms than those of form $\mathbb{R}\times H$ analysed in Sec.~\ref{sec:type_change}.

\section*{Statements and declarations}

\subsection*{Funding}
 This research has been supported by the Engineering \& Physical Sciences Research Council grant \verb!EP/Y000099/1! to the University of York.
For the purpose of open access a Creative Commons Attribution (CC BY) licence is applied to any Author Accepted Manuscript version arising from this submission.

\subsection*{Competing interests}
The authors have no competing interests to declare that are relevant to the content of this article.
\subsection*{Data availability} No datasets were generated or analysed in this work.
\appendix
\section{The measurement framework in functorial terms}\label{appx:functorial} 

In this appendix we describe the measurement framework in functorial language to put it in a broader perspective that provides motivation for the conditions (a) and (b) assumed in Section~\ref{sec: measurement with symm}. As in the main text, $M$ is a globally hyperbolic spacetime and $\Reg(M)$ is the set of its open regions. If $N_1,N_2\in\Reg(M)$ and $N_1\subset N_2$, let $\iota_{N_2;N_1}:N_1\to N_2$ denote the inclusion map of $N_1$ in $N_2$. Then $\Reg(M)$ becomes a category with inclusions as morphisms. Another category of interest is the category $\Alg$ of unital $*$-algebras with unit-preserving injective $*$-homomorphisms.

Any AQFT $\Af$ on $M$ determines a functor $\Af:\Reg(M)\to\Alg$ so that
$\Af(N)=\Af(M;N)$ and $\Af(\iota_{N_2;N_1}):\Af(N_1)\to\Af(N_2)$ is the inclusion map of $\Af(M;N_1)$ in $\Af(M;N_2)$. Two theories are equivalent if there is a natural isomorphism
$\zeta:\Af\to\Bf$, which means that there should be isomorphisms $\zeta_N:\Af(N)\to\Bf(N)$ so that $\Bf(\iota_{N_2;N_1})\circ \zeta_{N_1}= \zeta_{N_2}\circ \Af(\iota_{N_2;N_1})$ for all regions $N_1,N_2\in\Reg(M)$ with $N_1\subset N_2$. Because each $\Af(N)$ is realised as a subalgebra of $\Af(M)$, every component $\zeta_N$ is a restriction of the isomorphism $\zeta_M:\Af(M)\to\Bf(M)$; conversely, any isomorphism of $\zeta_M:\Af(M)\to\Bf(M)$ such that
$\zeta_M(\Af(M;N))= \Bf(M;N)$ for all $N\in\Reg(M)$ determines an equivalence of $\Af$ and $\Bf$.

In particular, a global gauge transformation of $\Af$ is an equivalence of $\Af$ with itself, i.e., an element of $\Aut(\Af)$ -- see~\cite{Fewster:gauge} for a discussion in a broader context -- and $\Aut(\Af)$ may be identified with a subgroup of $\Aut(\Af(M))$.
Similarly, if $\Af$ and $\Bf$ are QFTs describing the system and probe respectively, the isomorphisms $\chi_N$ relating the uncoupled combination $\Uf=\Af\otimes\Bf$ to a coupled combination $\Cf$ form a partial natural isomorphism $\chi:\Uf\rightharpoonup\Cf$ defined on regions $N$ outside the causal hull of the coupling region. Thus, they obey $\Cf(\iota_{N_2;N_1})\circ \chi_{N_1}= \chi_{N_2}\circ\Uf(\iota_{N_2;N_1})$ for all such regions $N_1$ and $N_2$ with $N_1\subset N_2$. This is a more formal expression of the idea that $\chi_{N_1}$ and $\chi_{N_2}$ agree on $\Uf(M;N_1)$.

Next, suppose that the group $G$ acts on $M$ by time-orientation-preserving isometries. 
Because this action respects inclusions, every $g\in G$ determines a functor $\Tf(g):\Reg(M)\to\Reg(M)$ so that $\Tf(g)(N)=g.N$ and $\Tf(g)(\iota_{N_2;N_1})= \iota_{g.N_2;g.N_1}$, and indeed the assignment $g\mapsto \Tf(g)$ is a group homomorphism $\Tf:G\to\Aut(\Reg(M))$, the group of invertible endofunctors on $\Reg(M)$. 

 Given a QFT $\Af:\Reg(M)\to\Alg$, each $g\in G$ determines a modified theory $\act{g}{\Af}=\Af\circ \Tf(g^{-1})$, as described in Section~\ref{sec: measurement with symm}. Note that $\act{g}{\Af}(\iota_{N_2;N_1})= \Af(\iota_{g^{-1}.N_2;g^{-1}.N_1})$. 
 The notion of $G$-covariance set out in Section~\ref{sec: measurement with symm} can be expressed equivalently as the condition that all theories $\act{g}{\Af}$ are equivalent, whereupon any family $(\alpha(g))_{g\in G}$
 of natural isomorphisms $\alpha(g):\act{g}{\Af}\to\Af$ with $\alpha(1_G)=\id_{\Af}$ is said to implement the $G$-covariance. This notion was explored in some detail and generality in~\cite{Fewster:2018}, where it was used in a proof of the Coleman--Mandula theorem for QFTs in curved spacetimes.\footnote{We adopt different conventions to~\cite{Fewster:2018}, which used $\act{g}{\Af}$ to denote the 
 \emph{right} action $\act{g}{\Af}=\Af\circ\Tf(g)$. Correspondingly, $\alpha(g)$ here corresponds to
 $\eta(g^{-1})^{-1}$ in the notation of~\cite{Fewster:2018}.} It turns out that every implementation $g\mapsto\alpha(g)$ of the $G$-covariance determines a normalised cocycle in $Z^2(G,\Aut(\Af))$ (Theorem 5 in~\cite{Fewster:2018}) and a group extension $E$ of the spacetime symmetry group $G$ by the global gauge group $\Aut(\Af)$, so that the 
 theory $\Af$ is also $E$-covariant (see Theorem~7 in~\cite{Fewster:2018}). In the present setting, 
 the naturality of $\alpha(g)$ implies that each component $\alpha(g)_N:\act{g}{\Af}(N)\to \Af(N)$ can be obtained as the restriction of $\alpha(g)_M\in\Aut(\Af(M))$ to $\act{g}{\Af}(N)=\Af(M;g^{-1}.N)$, and indeed the discussion in Section~\ref{sec: measurement with symm} 
 identified $\alpha(g)$ with the component $\alpha(g)_M$. The condition (a) assumed there, namely that $G\owns g\mapsto\alpha(g)_M\in\Aut(\Af(M))$ is a homomorphism, implies
 \begin{equation}
     \alpha(gg')_N=\alpha(g)_N \alpha(g')_{g^{-1}.N},
 \end{equation}
 for all $g,g'\in G$, $N\in\Reg(M)$,
 while condition (b) asserted that $\zeta\circ\alpha(g)=\alpha(g)\circ\zeta$ for all $g\in G$, $\zeta\in\Aut(\Af)$. By Theorem 5 in~\cite{Fewster:2018}, conditions (a) and (b) imply that the cocycle determined by $\alpha$ is trivial, which means that the group extension is a direct product, $E=\Aut(\Af)\times G$, in line with the expectations of the Coleman--Mandula theorem. 

 We can provide further motivation for conditions (a) and (b) from another direction. The AQFT axioms used here and in~\cite{FewVer_QFLM:2018} represent a cut-down version of the framework of locally covariant QFT~\cite{BrFrVe03,FewVerch_aqftincst:2015}, which sets out to describe a QFT on all globally hyperbolic spacetimes. There, one introduces a category $\Loc$ of oriented and time-oriented globally hyperbolic spacetimes, with morphisms that are (time)-orientation-preserving isometries with causally convex image. A QFT is then a functor $\Af:\Loc\to\Alg$ obeying further axioms generalising those we have stated in Section~\ref{sec: measurement in QFT}. The connection with the previous presentation is that each region of a fixed spacetime $M$ may be regarded as a globally hyperbolic spacetime in its own right, with the induced metric and (time)-orientation, whereupon the inclusion map becomes a $\Loc$ morphism. In this way, $\Reg(M)$ becomes a subcategory of $\Loc$ and by restricting a locally covariant AQFT $\Af:\Loc\to\Alg$ to $\Reg(M)$, we obtain an AQFT on $M$ in the sense of Section~\ref{sec: measurement in QFT}. We denote this by $\Af|_M$ in what follows. 
 
 If $M\in\Loc$ is a spacetime, the group $G=\Aut(M)$ of $\Loc$-isomorphisms from $M$ to itself consists of isometric isomorphisms of $M$ preserving orientation and time-orientation. If $g\in \Aut(M)$ then its restriction to $N\in\Reg(M)$ determines a $\Loc$-isomorphism, $g|_N:N\to g.N$, so that 
 \begin{equation}\label{eq:prenat}
     g|_{N_2}\circ\iota_{N_2;N_1} = \iota_{g.N_2;g.N_1}\circ g|_{N_1},
 \end{equation}
 whenever $N_1\subset N_2$ are nested regions. 
 For any $N\in\Reg(M)$, now set 
 \begin{equation}
 \alpha(g)_N=\Af|_M(g|_{g^{-1}.N}):\Af(g^{-1}.N)\to \Af(N).
 \end{equation}
 Thus $\alpha(g)_N:\act{g}{\Af}|_M(N)\to\Af|_M(N)$ and~\eqref{eq:prenat} implies that 
 \begin{equation}
 \Af|_M(\iota_{N_2;N_1})\circ \alpha(g)_{N_1}= \alpha(g)_{N_2}\circ\act{g}{\Af}|_M(\iota_{N_2;N_1})
 \end{equation}
 for all nested regions $N_1\subset N_2$. Thus
 $\alpha(g):\act{g}{\Af}|_M\to \Af|_M$ is a natural isomorphism for each $g\in G$ and
 $\alpha(1_G)=\id_{\Af|_M}$, so $g\mapsto\alpha(g)$ implements a $G$-covariance of $\Af|_M$. Furthermore, $g\mapsto \alpha(g)_M=\Af|_M(g)=\Af(g)$ is a homomorphism of
 $G$ to $\Aut(\Af(M))$ by functoriality of $\Af$, so condition (a) is satisfied. 
 Defining the global gauge group of $\Af:\Loc\to \Alg$ as the group $\Aut(\Af)$, it also follows that global gauge transformations commute with the spacetime symmetries, as shown in Section~2.2 of~\cite{Fewster:gauge}. This is our condition (b), modulo the possible subtlety that the global gauge group of $\Af|_M$ might admit transformations that do not extend to global gauge transformations in all spacetimes.

\section{Mackey's imprimitivity theorem and compactly stabilised QRFs}
\label{appx:imprimitvity_thm}
\label{appx:G_integrals}
In this appendix we give a short review of Mackey's imprimitivity theorem and its application to characterise systems of covariance as in \cite{Cattaneo:1979}. Following this, we apply this to the specific case of compactly stabilised QRFs and conclude with a proof of Thm.~\ref{thm:qref_embed}. We shall assume in this appendix that $G$ is a locally compact second countable Hausdorff group and $H\subset G$ a closed subgroup. $G$ defines the (separable complex) Hilbert space $L^2(G)$ of square integrable functions w.r.t. a left Haar measure on $G$ (unique up to a scale), $\mu_G:\Bor(G)\to[0,\infty]$, satisfying
\begin{equation}
    \mu_G(gX)=\mu_G(X),
\end{equation}
and
\begin{equation}
\label{eq:modular_func}
    \mu_G(Xg)=\Delta_G(g)\mu_G(X),
\end{equation}
for $\Delta_G:G\to \RR_{>0}$ the modular function (see \cite{deitmar2008principles} for details). Similarly, the group $H\subset G$ admits a left Haar measure $\mu_H$ and modular function $\Delta_H$. For convenience, we shall assume in this appendix that $\Delta_H=\Delta_{G}\restriction_H$. This is equivalent to the topological space $G/H$, given by
\begin{equation}
    G/H:=\{gH:g\in G\},
\end{equation}
with the usual quotient topology, admitting a non-zero left $G$-invariant Radon measure $\mu_{G/H}$, defining the Hilbert space $L^2(G/H)$, see \cite[Thm.~1.5.2]{deitmar2008principles}. In particular, this measure can be chosen such that for any compactly supported $f\in C_c(G)$, we have
\begin{equation}
\label{eq:quotient_measure}
    \int_G f\diff\mu_G=\int_{G/H}f^H\diff\mu_{G/H},
\end{equation}
where $f^H\in C_c(G/H)$ given by
\begin{equation}
    f^H(gH)=\int_H f(gh)\diff\mu_H(h).
\end{equation}
In the case that $H$ is compact, $\mu_{G/H}$ is given by the formula
\begin{equation}\mu_{G/H}(X)=\frac{\mu_G(q^{-1}(X))}{\mu_H(H)},\end{equation}
where $q:G\to G/H$ is the canonical quotient map given by $q(g)=gH$. In what follows we shall always assume that whenever $H$ is compact, the measure $\mu_H$ is normalised, i.e.~$\mu_H(H)=1$. 

The Hilbert space $L^2(G)$ (and similarly $L^2(H)$) admit natural strongly continuous unitary representations of $G$, namely the \emph{left} and \emph{right translation actions}
$\lambda_G,\rho_G:G\to \mathbf{U}(L^2(G))$, where for $g,g'\in G$, $\psi\in L^2(G)$ the left translation action is given by
\begin{equation}
\label{eq:left_action}
(\lambda_G(g)\psi)(g')=\psi(g^{-1}g'),\end{equation}
and similarly the right translation action is
\begin{equation}
\label{eq:right_action}
(\rho_G(g)\psi)(g')=\Delta_G(g)\psi(g'g).\end{equation}
As per usual for $L^2$ spaces, these Hilbert spaces also admit multiplication operators. For each essentially bounded $\mu_G$ $f\in L^\infty(G)$, the operator $T_f\in B(L^2(G))$ is such that for each $\psi\in L^2(G)$, $g\in G$
\begin{equation}
    (T_f\psi)(g)=f(g)\psi(g).
\end{equation}

The function spaces given above can be used to define the class of induced representations for the group $G$. These representations play a key role in Mackey's imprimitivity theorem.

\subsection{Induced representations and the imprimitivity theorem}\label{appsec:catt}
There are several equivalent ways to describe an induced representation. Here we give a definition as presented in \cite{kaniuth2012induced}, specialised to our assumptions above.
\begin{definition}
    \label{def:induced_rep}
    Let $G$ be a locally compact second countable Hausdorff group and $H\subset G$ be a closed subgroup such that $\Delta_H=\Delta_G\restriction_H$, whereupon $G/H$ admits a $G$-invariant Radon measure $\mu_{G/H}$ fixed by Eq. \eqref{eq:quotient_measure}. Furthermore, let $(U,\HH)$ be a strongly continuous unitary representation of $H$ on some separable Hilbert space $\HH$. We define the \textup{induced representation} $(\Ind_H^G U,L^2_U(G,\HH))$ as follows. First, $L^2_U(G,\HH)$ is the Hilbert space completion of equivalence classes of functions $\varphi:G\to \HH$ with
    \begin{enumerate}
        \item $g\mapsto\langle\xi,\varphi(g)\rangle_\HH$ is Borel measurable for each $\xi\in \HH$,
        \item for $g\in G$ and $h\in H$, $\varphi(gh)=U(h)^*\varphi(g)$,
        \item the function $gH\mapsto \Vert\varphi(g)\Vert_\HH$ is in $L^2(G/H)$,
    \end{enumerate}
    with pre-inner product given by 
    \begin{equation}
    \label{eq:induced_prod}
        \langle \varphi,\varphi'\rangle_{L^2_U(G,\HH)}=\int_{G/H}\langle \varphi(g),\varphi'(g)\rangle_\HH\,\diff\mu_{G/H}(gH),
    \end{equation}
    and the equivalence relation given as usual by $\varphi\sim\varphi'$ if $\Vert\varphi-\varphi'\Vert_{L^2_U(G,\HH)}=0$. 
    
    Second, we define the unitary representation $\Ind_H^G$ of $G$ on $L^2_U(G,\HH)$ via the left translation action
    \begin{equation}
        (\Ind_H^G U(g)\varphi)(g')=\varphi(g^{-1}g').
    \end{equation}
    \end{definition}
    Note that the function $gH\mapsto \Vert\varphi(g)\Vert_\HH$ is well defined, as $\Vert\varphi(gh)\Vert_\HH=\Vert U(h)^*\varphi(g)\Vert_\HH=\Vert \varphi(g)\Vert_\HH$. 
    
    Any induced representation admits a natural system of imprimitivity.
    \begin{definition}
        \label{def:induced_si}
        For $G,H,U,\HH$ as in definition \ref{def:induced_rep}, one defines the \textup{induced system of imprimitivity} as $(\Ind_H^G U,P_U,L^2_U(G,\HH))$ with $P_U:\Bor(G/H)\to B(L^2_U(G,\HH))$ a projection valued measure where for each $\varphi\in L^2_U(G,\HH)$, $g\in G$ and $X\in\Bor(G/H)$
        \begin{equation}
            (P_U(X)\varphi)(g)=\chi_X(q(g))\varphi(g).
        \end{equation}
        Here $\chi_X$ denotes the characteristic function of $X$ on $G/H$ and $q:G\to G/H$ the canonical quotient map.
    \end{definition}
    Indeed one can easily check that $P_U$ is a well defined projection valued measure on $G/H$ where for each $X\in \Bor(G/H)$, $g\in G$ one has $\Ind_H^G U(g)P_U(X)\Ind_H^G U(g)^*=P_U(gX)$.
    
    Such systems of imprimitivity play an important role in the characterisation of systems of covariance, as appearing in the context of representation theory for locally compact groups, but which can also be understood as quantum reference frames as defined in Def.\ \ref{def:qrf1}. We recall a result due to \cite{Cattaneo:1979}, generalizing Mackey's imprimitivity theorem, see e.g.\ \cite{mackey1949imprimitivity,raczka1986theory,kaniuth2012induced}, which tells us that any quantum reference frame $(U_R,E,\HH_R)$ (or system of covariance) for $G$ with value space $G/H$ is equivalent to a compression of an induced system of imprimitivity, as given in Def.\ \ref{def:induced_si}.
\begin{proposition}[Cattaneo]
        \label{prop:cat_impr_thm}
        Let $G$ be a second countable locally compact Hausdorff group and let $H$ be a (not necessarily compact) closed subgroup such that $G/H$ admits a nonzero G-invariant Radon measure. Let $\mathcal{R}=(U_R,E,\HH_R)$ be a quantum reference frame for $G$ with value space $G/H$ as in Def.~\ref{def:qrf1}. 
        Then there exists a separable Hilbert space $\KK$, a strongly continuous unitary representation $(V,\KK)$ of $H$ and an isometry $W_\mathcal{R}:\HH_R\to L^2_V(G,\KK)$ such that for the induced system of imprimitivity $(\Ind_H^G V,P_V,L^2_V(G,\KK))$, one has for all $g\in G$, $X\in\Bor(G/H)$
        \begin{equation}
            W_\mathcal{R}U_R(g)=\left(\Ind_H^G V(g)\right)W_\mathcal{R},\qquad E(X)=W_\mathcal{R}^*P_V(X)W_\mathcal{R}.
        \end{equation}
        Furthermore $(V,\KK)$ can always be chosen such that $W_\mathcal{R}$ is a minimal covariant Naimark dilation, i.e. the set
        \begin{equation}
            P_V(\Bor(G/H))W_\mathcal{R} \HH_R=\{P_V(X)W_\mathcal{R}\psi:X\in \Bor(G/H),\, \psi\in \HH_R\}
        \end{equation}
         spans a dense subset of $L_V^2(G,\KK)$. The representation $(V,\KK)$ yielding this minimality property is unique up to unitary equivalence and under this minimal choice $W_\mathcal{R}$ is an isomorphism if and only if $E$ is projection valued.
\end{proposition}
We refer to \cite{Cattaneo:1979} and references therein for the full proof of this proposition. Here we will just remark that the proof combines the following two ingredients. Firstly, given a quantum reference frame $(U_R,E,\HH_R)$ as in Prop.\ \ref{prop:cat_impr_thm}, one can apply a version of \textit{Naimark's dilation theorem}, see e.g.\ \cite[Thm.\ 4.6.]{Paulsen_2003} or Thm.\ \ref{thm:naimark} , to construct a system of imprimitivity $(\tilde{U}_R,P,\tilde{\HH}_R)$ for which there is an embedding $\tilde{W}:\HH_R\to \tilde{\HH}_R$ with for $g\in G$, $X\in \Bor(G/H)$
\begin{equation}
    \tilde{W} U_R(g)=\tilde{U}_R(g)\tilde{W},\qquad E(X)=\tilde{W}^*P(X)\tilde{W},
\end{equation}
such that $P(\Bor(G/H))\tilde{W}\HH_R$ spans a dense subset of $\tilde{\HH}_R$, i.e.\ such that $\tilde{W}$ is a minimal covariant Naimark dilation. Here the covariance aspect is reflected by the fact that $\tilde{W}$ intertwines the unitary representations of $G$, $U_R$ and $\tilde{U}_R$. As shown in \cite[Prop.\ 1.]{Cattaneo:1979}, this system of imprimitivity is unique up to unitary equivalence. Here this equivalence is meant in the following sense.
\begin{definition}
    \label{def:qrf_equiv}
    Let $\mathcal{R}=(U,E,\HH)$ and $\mathcal{R}' =(U',E',\HH')$ be QRFs for a group $G$ with value space $\Sigma$ as in Def.~\ref{def:qrf1}. We say $\mathcal{R}$ and $\mathcal{R}'$ are \textup{unitarily equivalent} if there exists an isometric isomorphism $V:\HH\to \HH'$ such that for each $g\in G,\, X\in\Bor(\Sigma)$, 
    \begin{equation}
        VU(g)=U'(g)V,\qquad VE(X)=E'(X)V.
    \end{equation}
    In other words, if $V$ intertwines the unitary representations $U$ and $U'$ and the POVMs $E$ and $E'$.
\end{definition}

Secondly, one can now apply Mackey's imprimititvity theorem, such as given in \cite{kaniuth2012induced}, to the dilated system of imprimitivity $(\tilde{U}_R,P,\tilde{\HH}_R)$. By this theorem, this system of imprimitivity is unitarily equivalent to an induced system of imprimitivity as given by Def.\ \ref{def:induced_si}. One can hence conclude that one can always dilate the quantum reference frame $(U_R,E,\HH_R)$ to an induced system of imprimitivity.

Proposition \ref{prop:cat_impr_thm} gives a characterisation of a very general class of quantum reference frames. In this paper, we are mostly interested in quantum reference frames that are compactly stabilised. Below we describe how the imprimitivity theorem simplifies for this particular class of reference frames.

\subsection{Characterising compactly stabilised quantum reference frames}
\label{appx:compact_stab_qrf_charact}
For $H\subset G$ be compact, then induced representations have an alternative characterization to the one given in Def.\ \ref{def:induced_rep}. Let $\mu_H$ be normalised such that $\mu_H(H)=1$. We can apply the relation \eqref{eq:quotient_measure} to \eqref{eq:induced_prod} to see that for $\varphi,\varphi'\in L^2_U(G,\HH)$
    \begin{equation}
        \langle\varphi,\varphi'\rangle_{L^2_U(G,\HH)}=\int_G\langle\varphi(g),\varphi'(g)\rangle_\HH\,\diff\mu_G(g),
    \end{equation}
    hence $L^2_U(G,\HH)\subset L^2(G,\HH)$. This Hilbert space admits a strongly continuous unitary representation $(\tilde{U},L^2(G,\HH))$ of $H$ with 
    \begin{equation}
        (\tilde{U}(h)\varphi)(g)=U(h)\varphi(gh),
    \end{equation}
    or, under the natural isomorphism $L^2(G,\HH)\cong \HH\otimes L^2(G)$, $\tilde{U}$ corresponds to $U\otimes\rho_H$, with $\rho_H=\rho_G\restriction_H$ the right translation action of $H$ on $L^2(G)$ (see Sec.\ \ref{appx:G_integrals}). Note that $\varphi\in L^2(G,\HH)$ is a fixed point for this group action if and only if $\varphi(gh)=U(h)^*\varphi(g)$ for any $h\in H$ and almost every $g\in G$. Hence we see by (2) of Def.~\ref{def:induced_rep} that, under the natural isomorphism $L^2(G,\HH)\cong \HH\otimes L^2(G)$ 
    \begin{equation}
        L^2_U(G,\HH)\cong (\HH\otimes L^2(G))^{U\otimes \rho_H}.
    \end{equation}
    As $H$ is compact, the fixed points under this unitary action are exactly given by 
    \begin{equation}
        (\HH\otimes L^2(G))^{U\otimes \rho_H}=\Pi_{U}(\HH\otimes L^2(G)),
    \end{equation} with $\Pi_{U}\in B(\HH\otimes L^2(G))$ the projection
    \begin{equation}
    \label{eq:average_proj}
        \Pi_{U}=\int_HU(h)\otimes\rho_H(h)\,\diff\mu_H(h).
    \end{equation}
    If we consider the left translation action $\lambda_G$ of $G$ on $L^2(G)$, we see that $U\otimes\rho_H$ and $1_\HH\otimes\lambda_G$ are commuting group actions of $H$ and $G$ respectively on $L^2(G,\HH)$, this means in particular that
    \begin{equation}
        \Pi_U\in B(\HH)\otimes\lambda_G(G)'.
    \end{equation}
    It follows that $\Pi_{U}(\HH\otimes L^2(G))$ is closed under the action of $(1_\HH\otimes \lambda)$, and in fact one sees by Def.\ \ref{def:induced_rep} that, under the natural isomorphism, $\Ind_H^G U(g)$ corresponds to $(1_\HH\otimes\lambda(g))\restriction_{\Pi_{U}(\HH\otimes L^2(G))}$. This allows us to conclude the following.
    \begin{lemma}
        \label{lem:compact_subrep}
        For $G$ locally compact, second countable and Hausdorff, $H\subset G$ compact such that $G/H$ admits a nonzero left $G$-invariant Radon measure, the induced representation of $(U,\HH)$ is unitarily equivalent to a subrepresentation of $(1_\HH\otimes\lambda_G,\HH\otimes L^2(G))$ given by the invariant subspace $\Pi_U(\HH\otimes L^2(G))$, with $\Pi_U$ given by Eq.\ \eqref{eq:average_proj}.
    \end{lemma}

    We now describe the induced system of imprimitivity using this alternative characterization of the induced reprepresentation. Consider the system of imprimitivity
    \begin{equation}
        \label{eq:compact_stab_qrf_dilated}
        \mathcal{R}=(1_\HH\otimes \lambda, 1_\HH\otimes P,\HH\otimes L^2(G))
    \end{equation} for $G$ with value space $G/H$, where $P:\Bor(G/H)\to B(L^2(G))$ is the projection valued measure such that for each $X\in\Bor(G/H)$ one has \begin{equation}
        \label{eq:compact_stab_dilated_proj}
        P(X)=T_{\chi_{q^{-1}(X)}}.
    \end{equation} Here as in Sec.\ \ref{appx:G_integrals},  $T_f\in B(L^2(G))$ is the multiplication operator for $f\in L^\infty(G)$ and $q:G\to G/H$ the canonical quotient map. Note in particular that for any $X\in \Bor(G/H)$, $P(X)$ commutes with $\rho_H$, hence for $\Pi_U$ as defined by Eq.\ \eqref{eq:average_proj}, we have $\Pi_U\in B(\HH)\otimes P(\Bor(G/H))'$. Under the isomorphism $L^2_U(G,\HH)\cong \Pi_U(\HH\otimes L^2(G))$, one can now easily verify that $(1_\HH\otimes P(X))\restriction_{\Pi_U(\HH\otimes L^2(G))}$ corresponds to $P_U$ defined as in Def.\ \ref{def:induced_si}. Hence we can extend Lem.\ \ref{lem:compact_subrep} as follows.
    \begin{proposition}
        \label{prop:compact_ind_impr}
        For $H\subset G$ as in Lem.\ \ref{lem:compact_subrep}, 
        $(U,\HH)$ a strongly continuous representation of $H$ on a separable Hilbert space $\HH$ 
        and $(\Ind_H^GU,P_U,L^2_U(G,\HH))$ its induced system of imprimitivity, there exists an isometry $W_U:L^2_U(G,\HH)\to\HH\otimes L^2(G)$ with $W_UW_U^*=\Pi_U$ such that
        \begin{equation}
            W_U(\Ind_H^G U)=(1_\HH\otimes\lambda)W_U,\; W_UP_U=(1_\HH\otimes P)W_U,
        \end{equation}
        i.e. $W_U$ intertwines the systems of imprimitivity $(\Ind_H^GU,P_U,L^2_U(G,\HH))$ and $(1_\HH\otimes \lambda, 1_\HH\otimes P,\HH\otimes L^2(G))$.
    \end{proposition}
    
    We can now prove Thm.\ \ref{thm:qref_embed}.
    \begin{proof}[Proof of Thm.\ \ref{thm:qref_embed}]
    \begin{enumerate}
        \item If one combines Prop.\ \ref{prop:compact_ind_impr} with Prop.\ \ref{prop:cat_impr_thm}, it follows that any compactly stabilised quantum reference frame $(U_R,E,\HH_R)$ can be dilated to a QRF as given by Eq.\ \eqref{eq:compact_stab_qrf_dilated}, although not necessarily in a minimal sense.
        \item  Conversely, for any subrepresentation of $(1_\HH\otimes\lambda_G,\HH\otimes L^2(G))$ with $\HH$ separable, i.e. for every projection $p\in B(\HH)\otimes \lambda_G(G)'$, one can verify that the compression $(p(1_\HH\otimes \lambda), p(1_\HH\otimes P)p, p(\HH\otimes L^2(G)))$ for $P:\Bor(G/H)\to L^\infty(G)$ as given by Eq.\ \eqref{eq:compact_stab_dilated_proj} defines a compactly stabilised quantum reference frame. Up to unitary equivalence, these are precisely the quantum reference frames given in Thm.\ \ref{thm:qref_embed}.
    \end{enumerate}
    \end{proof}
    
\section{Further proofs for statements in Sec.\ \ref{sec:qrf and measurement}}
\label{appx:proofs}
In this appendix we give the proofs of various results stated in sections \ref{subsec:relativisation} and \ref{sec:RHcase}.
\subsection{Proofs for statements is Sec.\ \ref{subsec:relativisation}}
\label{appx:relativisation_proofs}
\begin{proof}[Proof of Lem.\ \ref{lem:pi_range}]
        Identify $L^\infty(G)$ with the von Neumann algebra \begin{equation}
            \{T_f\in B(L^2(G)):f\in L^\infty(G)\}''.
        \end{equation} This algebra satisfies $L^\infty(G)'=L^\infty(G)$ (see \cite[Thm.~I.7.3.2]{dixmier1981VonNeumann}). Now note that
        \begin{equation}
            P(\Bor(G/H))''\subset (L^\infty (G)\cup \rho(H))'= L_H^\infty(G)\subset P(\Bor(G/H))'',
        \end{equation} where we recall that for $X\in\Bor(G/H)$
        \begin{equation}
            P(X)=T_{\chi_{q^{-1}X}}\in L^\infty (G)\cup \rho(H)'
        \end{equation} for $q:G\to G/H$ the canonical quotient map. Here we have used that by $\eqref{eq:quotient_measure}$, $L_H^\infty(G):=L^\infty(G)^{\Ad \rho\restriction_{H}}$ is isomorphic to $L^\infty(G/H)$.
        Hence $P(\Bor(G/H))''=L_H^\infty(G)$. By \cite[Part I, Prop.\ 2.5]{vanDaele:1978}, the map $\pi:\MM_S\to B(L^2(G,\HH_S))$ is an injective normal *-homomorphism, with $\pi(\MM_S)\in \MM_S\otimes  L^\infty(G)$. In fact, following the proof of \cite[Part I, Thm.\ 3.11]{vanDaele:1978} one sees that 
        \begin{equation}
            \pi(\MM_S)= (\MM_S\otimes L^\infty(G))^{\alpha_S\otimes\Ad \lambda}.
        \end{equation}
        Since $\rho(G)\subset \lambda(G)'$, it follows that
        \begin{align}
            \pi(\MM_S)^{\Ad \tilde\rho\restriction_H} =& (\MM_S\otimes L^\infty(G)^{\Ad\rho\restriction_H})^{\alpha_S\otimes\Ad\lambda}\nonumber\\
            =&(\MM_S\otimes P(\Bor(G/H))'')^{\alpha_S\otimes\Ad\lambda}.
        \end{align}
        By \cite[Lem.~I.2.9]{vanDaele:1978}, where we recall that we have here exchanged the role of the right and left translation action in comparison to that reference, we have for $x\in \MM_S$, $h\in H$ that
        \begin{equation}
            \tilde{\rho}(h)\pi(x)\tilde{\rho}(h)^*=\pi(\alpha_S(h,x))
        \end{equation}
        Hence $\pi(x)\in\pi(\MM_S)^{\Ad \tilde\rho\restriction_H}$ means that $\pi(x)=\pi(\alpha_S(h,x))$. By injectivity of $\pi$, we have $x\in \MM_S^{\alpha_S\restriction_H}$. Since $\pi$ is a normal *-homomorphism, we see that
        \begin{align}
            \pi(\MM_S^{\alpha_S\restriction_H})=& \pi(\MM_S)^{\Ad \tilde\rho\restriction_H}\nonumber\\
            =&(\MM_S\otimes P(\Bor(G/H))'')^{\alpha_S\otimes\Ad \lambda}.
        \end{align}
    \end{proof}

    \begin{proof}[Proof of Prop.\ \ref{prop:yen_properties}]
        We first show independence of the choice of the map $W:\HH_R\to L^2(G)\otimes \KK$ satisfying the properties in Def.\ \ref{def:relativisation_compact_stab} for $\yen$. Choose $\varphi_i\in \HH_S$ and $\psi_i\in \HH_R$ for $i=1,2$. We calculate for $x\in \MM_S^{\alpha_S\restriction_H}$
        \begin{align}
            \langle \varphi_1\otimes \psi_1,\yen(x)(\varphi_2\otimes \psi_2)\rangle=&\int_G \langle\varphi_1,\alpha(g,x)\varphi_2\rangle\langle (W\psi_1)(g),(W\psi_2)(g)\rangle\,\diff\mu_G(g)\nonumber\\
            =&\int_{G/H} \langle\varphi_1,\alpha(g,x)\varphi_2\rangle\,\diff\nu_{W\psi_1,W\psi_2}(gH),
        \end{align}
        where similarly as in Eq.\ \eqref{eq:proj_measure} for $X\in \Bor(G/H)$
        \begin{equation}
            \nu_{W\psi_1,W\psi_2}(X)=\langle W\psi_1,(T_{\chi_{q^{-1}X}}\otimes 1_\KK)W\psi_2\rangle=\langle \psi_1,E(X)\psi_2\rangle.
        \end{equation}
        Note thus that the measure $\nu_{W\psi_1,W\psi_2}$ is independent on $W$ and only depends on $\psi_1,\psi_2$. Therefore $\langle \varphi_1\otimes \psi_1,\yen(x)(\varphi_2\otimes \psi_2)\rangle$ is independent on $W$ and we see that the definition of $\yen$ is unique.

        Since $\pi$ is a normal *-homomorphism \cite[Prop.\ 2.5]{vanDaele:1978}, complete positivity, unitality, normality and linearity of $\yen$ naturally follow, see e.g.\ \cite[Thm.\ IV.3.6]{Takesaki2001}.
        
         To prove that $\yen$ is a *-isomorphism if $E$ is projection valued, we first show that one can take $W$ to intertwine the PVM $E$ and $P$. Let $E:\Bor(G/H)\to B(\HH_R)$ be a covariant PVM, then by Prop.\ \ref{prop:cat_impr_thm} there exists a unitary isomorphism $W_{\mathcal{R}}:\HH_R\to L_V^2(G,\KK)$ for some unitary representation $(V,\KK)$ of $H$ with $W_{\mathcal{R}}E=P_VW_{\mathcal{R}}$. By prop.\ \ref{prop:compact_ind_impr}, there exists an isometric embedding $W_V:L_V^2(G,\KK)\to L^2(G)\otimes\KK$ such that $PW_{V}=W_{V}P_V$ with $P(X)=T_{\chi_{q^{-1}(X)}}\otimes 1_{\KK}$. For the projection $\Pi_V=W_VW_V^*$,
        with
        \begin{equation}
            \Pi_V=\int_H\rho(h)\otimes V(h)\diff\mu_H(h),
        \end{equation}
        it follows that $\Pi_VP=P\Pi_V$. Hence there exists an isometric embedding $W:\HH_R\to L^2(G)\otimes\KK$, where $W=W_VW_\mathcal{R}$, such that $WE=PW$. As $WW^*=\Pi_V$ we have $WW^*P=PWW^*$.
        
        By definition of the relativisation map, we have for $x\in \MM_S^{\alpha_S\restriction_H}$
        \begin{equation}
            \yen(x)=(1_{\HH_S}\otimes W)^*(\pi(x)\otimes 1_{\KK})(1_{\HH_S}\otimes W).
        \end{equation}
        Note furthermore that due to invariance of $x$ under $\alpha_S\restriction_H$, one finds \begin{equation}
            (\pi(x)\otimes 1_\KK)(1_{\HH_S}\otimes W)(1_{\HH_S}\otimes W)^*=(\pi(x)\otimes 1_\KK)(1_{\HH_S}\otimes\Pi_V)=(1_{\HH_S}\otimes\Pi_V)(\pi(x)\otimes 1_\KK).
        \end{equation}
        We can now see that for $x,y\in \MM_S^{\alpha_S\restriction_H}$, we have that
        \begin{align}
            \yen(x)\yen(y)=&(1_{\HH_S}\otimes W)^*(\pi(x)\otimes 1_\KK)(1_{\HH_S}\otimes\Pi_V)(\pi(y)\otimes 1_\KK)(1_{\HH_S}\otimes W)\nonumber\\
            =&(1_{\HH_S}\otimes W)^*(1_{\HH_S}\otimes\Pi_V)(\pi(x)\otimes 1_\KK)(\pi(x)\otimes 1_\KK)(1_{\HH_S}\otimes W)\nonumber\\
            =&(1_{\HH_S}\otimes W)^*(\pi(xy)\otimes 1_\KK)(1_{\HH_S}\otimes W)=\yen(xy),
        \end{align}
        where we have used that $\pi$ is a *-homomorphism. It hence follows that $\yen$ is a *-homomorphism.
        
        To show that under the assumptions described above, \begin{equation}
            \yen:\MM_S^{\alpha_S\restriction_H}\to (\MM_S\otimes E(\Bor(G/H))'')^{\alpha_S\otimes \Ad U_R}
        \end{equation} is an isomorphism, we use the following argument. Note that for $x\in \MM_S \otimes L_H^\infty(G)\otimes \CC 1_\KK$, $x(1_{\HH_S}\otimes\Pi_V)=0$ if and only if $x=0$. If for $y\in \MM_S^{\alpha_S\restriction_H}$ one has $\yen(y)=0$, this implies $(\pi(y)\otimes 1_\KK)(1_{\HH_S}\otimes\Pi_V)=0$. It thus follows that $\pi(y)=0$ and since $\pi$ is injective it follows that $y=0$, i.e.~$\yen$ is injective.
        
        Since $\Pi_V\in (L^\infty_H(G)\otimes 1_\KK)'$, with $L^\infty_H(G)=L^\infty(G)^{\Ad\rho\restriction_H}\subset B(L^2(G))$, we have by \cite[Prop.\ I.2.1]{dixmier1981VonNeumann} that $( L^\infty_H(G)\otimes 1_\KK)\restriction_{\Pi_V(
        L^2(G)\otimes\KK)}$ is a von Neumann algebra. Since $W:\HH_R\to \Pi_V(L^2(G)\otimes\KK)$ is an isomorphism, one has
        \begin{align}
            W^*(L^\infty_H(G)\otimes 1_\KK)W=&(W^*(L^\infty_H(G)\otimes 1_\KK)W)''\nonumber\\
            =&E(L^\infty(G/H))''\nonumber\\
            =&E(\Bor(G/H))'',
        \end{align}
        where by definition for $f\in L^\infty(G/H)\cong L^\infty_H(G)$
        \begin{equation}
            E(f)=W^*P(f)W=\int_{G/H} f(gH)\,\diff E(gH).
        \end{equation}
        Since $W$ intertwines $U_R$ and $\lambda\otimes 1_\KK$, we can now conclude that
        \begin{align}
            \yen(\MM_S^{\alpha_S\restriction_H})=&(1_{\HH_R}\otimes W)^*\left((\MM_S\otimes L^\infty_H(G))^{\alpha_S\otimes\Ad\lambda}\otimes 1_\KK\right)(1_{\HH_R}\otimes W)\nonumber\\=&
            \left((1_{\HH_R}\otimes W)^*(\MM_S\otimes L^\infty_H(G)\otimes 1_\KK)(1_{\HH_R}\otimes W)\right)^{\alpha_S\otimes\Ad U_R}\nonumber\\
            =&(\MM_S\otimes E(\Bor(G/H))'')^{\alpha_S\otimes \Ad U_R},
        \end{align}
        where we have used lemma \ref{lem:pi_range}. This proves surjectivity of $\yen$ and hence $\yen$ is an isomorphism.
    \end{proof}
    \begin{proof}[Proof of Prop.~\ref{prop:rel_expect}]
    For vectors $\psi_S,\psi_S'\in\HH_S$ and $\psi_R,\psi_R'\in\HH_R$ respectively, one sees by Eq.\ \eqref{eq:yenproj} and \eqref{eq:pi_matrixelt} that
    \begin{equation}
            \langle \psi_S\otimes\psi_R,\yen(x)(\psi_S'\otimes\psi_R')\rangle=\int_{G/H} \langle\psi_S,\alpha_s(g,x)\psi_S'\rangle\, \diff\nu_{\psi_R,\psi_R'}(gH),
        \end{equation}
        with for $X\in \Bor(G/H)$
        \begin{equation}
            \nu_{\psi_R,\psi_R'}(X)=\langle \psi_R,E(X)\psi_R'\rangle.
        \end{equation}

    Now consider normal states $\omega_S$ and $\omega_R$ on $\MM_S$ and $B(\HH_R)$ respectively. By \cite[Thm.~7.1.12]{kadison1997fundamentals} there exists sequences of vectors $\{\psi_{S/R,n}\in \HH_{S/R}\}_{n=1}^\infty$ such that
    \begin{equation}
        \omega_{S/R}(\,\cdot\,)=\sum_{n=1}^\infty \langle \psi_{S/R,n},(\,\cdot\,)\psi_{S/R,n}\rangle.
    \end{equation}
    For $x\in \left(\MM_S^{\alpha_S\restriction_H}\right)^+$, we have by dominated convergence that
    \begin{align}
        (\omega_S\otimes\omega_R)(\yen(x))=&\sup_{N,N'}\sum_{n'=1}^{N'}\int_{G/H} \sum_{n=1}^N\langle\psi_{S,n},\alpha_s(g,x)\psi_{S,n}\rangle \,\diff\nu_{\psi_{R,n'},\psi_{R,n'}}(gH)\nonumber\\
        =&\sup_{N'}\sum_{n'=1}^{N'}\int_{G/H} \omega_{S}(\alpha_s(g,x)) \,\diff\nu_{\psi_{R,n'},\psi_{R,n'}}(gH).
    \end{align}
    Since furthermore for $X\in \Bor(G/H)$
    \begin{equation}
        \sum_{n=1}^{N}\nu_{\psi_{R,n},\psi_{R,n}}(X)\leq \sum_{n=1}^{\infty}\nu_{\psi_{R,n},\psi_{R,n}}(X)=(\omega_R\circ E)(X),
    \end{equation}
    Denote $\mu_R:=\omega_R\circ E$. By the Radon-Nikodym theorem, there exists a $\mu_R$-almost everywhere increasing sequence of $\mu_R$-measurable functions $f_{N}:G/H\to [0,1]$ such that
    \begin{equation}
        \sum_{n=1}^{N}\nu_{\psi_{R,n},\psi_{R,n}}(X)=\int_X f_N(gH)\,\diff\mu_R(gH).
    \end{equation}
    Note that
    \begin{equation}
        \mu_R(X)=\sup_{N}\int_X f_N(gH)\,\diff\mu_R(gH)=\int_X \sup_{N}f_N(gH)\,\diff\mu_R(gH).
    \end{equation}
    It follows that $\sup_{N}f_N(gH)=1$ $\mu_R$-almost everywhere. We use this to show
    \begin{align}
        (\omega_S\otimes\omega_R)(\yen(x))=&\int_{G/H} \sup_{N}\omega_{S}(\alpha_s(g,x)) f_N(gH)\,\diff\mu_R(gH)\nonumber\\
        =&\int_{G/H} \omega_{S}(\alpha_s(g,x)) \,\diff\mu_R(gH).
    \end{align}
    \end{proof}

    \subsection{Proofs for statements in Sec.\ \ref{sec:RHcase}}
    \label{appx:spect_decomp_proof}
    \begin{proof}[Proof of Prop.\ \ref{prop:spect_decomp}]
    By Thm.\ \ref{thm:qref_embed}, there exists an isometric embedding $W:\HH_R\to L^2(\RR)\otimes L^2(H)\otimes \KK'$ for $\KK'$ separable such that for each $(t,h)\in \RR\times H$, one has $WU_R(t,h)=(\lambda_\RR(t)\otimes \lambda_H(h)\otimes 1_{\KK'})W$, hence also $WW^*(\lambda_\RR(t)\otimes \lambda_H(h)\otimes 1_{\KK'})=(\lambda_\RR(t)\otimes \lambda_H(h)\otimes 1_{\KK'})WW^*$, or $p:=WW^*\in \lambda_\RR(\RR)'\otimes \lambda_H(H)'\otimes B(\KK')$. We now define 
    \begin{equation}
        V=(\mathcal{F}\otimes 1_{L^2(H)}\otimes 1_{\KK'})W,
    \end{equation}
    with the Fourier transform $\mathcal{F}\in\mathbf{U}(L^2(\RR))$ being the continuous extension of 
    \begin{equation}
        (\mathcal{F}f)(\xi)=\frac{1}{\sqrt{2\pi}}\int_\RR \exp(i\xi t)f(t)\diff t,
    \end{equation}
    for $f\in C_c(\RR)$. Naturally $V:\HH\to L^2(\RR)\otimes L^2(H)\otimes\KK'$ is also an isometric embedding. Furthermore
    \begin{align}
        VU_R(t,h)=&(\mathcal{F}\lambda_\RR(t)\mathcal{F}^*\otimes \lambda_H(h)\otimes 1_{\KK'})V\nonumber\\
        =&(T_{\exp(it.)}\otimes \lambda_H(h)\otimes 1_{\KK'})V
    \end{align}
    where $T_{\exp(it.)}$ denotes the multiplication operator defined by the function $\RR\ni\xi \mapsto \exp(it\xi)$.
    Defining
    \begin{equation}
        \hat{p}:=VV^*,  
    \end{equation}
    and noting that, by \cite[Lem.\ B2]{vanDaele:1978}, $\{T_{\exp(it.)}:t\in\RR\}'=L^\infty(\RR)'$, it follows that 
    \begin{equation}
        \hat{p}\in L^\infty(\RR)'\otimes\lambda_H(H)'\otimes B(\KK'),
    \end{equation}
    where we embed $L^\infty(\RR)\subset B(L^2(\RR))$ as multiplication operators. Note that $L^\infty(\RR)'=L^\infty(\RR)$, see e.g.\ \cite[Thm.\ I.7.3.2]{dixmier1981VonNeumann}, thus
    \begin{equation}
        \hat{p}\in L^\infty(\RR)\otimes\lambda_H(H)'\otimes B(\KK').
    \end{equation}
    One can trivially write $L^2(\RR)\otimes L^2(H)\otimes \KK'=\int_\RR^\oplus L^2(H)\otimes \KK'\,\diff \xi$, see \cite[Ch.\ IV.8]{Takesaki2001}. By the above observation and by \cite[Cor.\ IV.8.16]{Takesaki2001}, $\hat{p}$ is a decomposable operator on this Hilbert space, i.e.\ there exists a measurable map
    \begin{equation}
    \xi\mapsto \hat{p}(\xi)\in \lambda_H(H)' \otimes B(\KK') \subset B(L^2(H)\otimes \KK') ,
    \end{equation}
    such that for $\psi =\int_\RR^\oplus \psi(\xi)\diff \xi \in \int_\RR^\oplus L^2(H)\otimes \KK'\diff \xi$ with $\xi\mapsto\psi(\xi)$ square integrable, one has 
    \begin{equation}\hat{p}\psi=\int_\RR^\oplus \hat{p}(\xi)\psi(\xi)\diff \xi.
    \end{equation}
    As $\hat{p}$ is a projection, it follows that $\hat{p}(\xi)\psi(\xi)=\hat{p}(\xi)^2\psi(\xi)=\hat{p}(\xi)^*\psi(\xi)$ almost everywhere, so that without loss of generality one can choose $\hat{p}(\xi)$ to be a projection valued.
    Note that, as $\hat{p}(\xi)$ is measurable, so is the family of Hilbert spaces $\KK(\xi):=\hat{p}(\xi)(L^2(H)\otimes \KK')$ as in \cite[Def.\ IV.8.9]{Takesaki2001}, where the measurable vector fields are given by $\xi\mapsto \hat{p}(\xi)\psi(\xi)$ for $\psi(\xi)\in L^2(H)\otimes \KK'$ measurable and given a countable basis $\{\psi^{(n)}\}_{n\in \mathbb{N}}$ of $L^2(H)\otimes \KK'$, the measurable vector fields $\{\xi\mapsto \hat{p}(\xi)\psi^{(n)}\}_{n\in \mathbb{N}}$ form a fundamental sequence of $\xi\mapsto \KK(\xi)$. One can now see that
    \begin{equation}
         \HH\cong \hat{p}\int_{\RR}^{\oplus}L^2(H)\otimes \KK'\diff\xi\cong \int_{\RR}^{\oplus}\hat{p}(\xi)(L^2(H)\otimes \KK')\diff\xi= \int_{\RR}^{\oplus}\KK(\xi)\diff\xi.
    \end{equation}
    Note that under the isomorphism above, one has for a square integrable $\psi(\xi)\in \KK(\xi)$ that \begin{equation}
        (U(t,h)\psi)(\xi)=\exp(it\xi)(\lambda_{H}(h)\otimes 1_{\KK'})\psi(\xi).
    \end{equation}
    Note that this in particular implies that given the spectral measure $P:\Bor(\RR)\to B(\HH)$ associated with the action of $\RR$ via the SNAG theorem
    \begin{equation}
        U(t,1_H)=\int_\RR \exp(i\xi t)\diff P(\xi),
    \end{equation}
    one finds under the isomorphism described above that for any $X\in \Bor(\RR)$
    \begin{equation}
        P(X)\HH\cong \int_X^\oplus \KK(\xi)\diff \xi,
    \end{equation}

    Now consider $\xi\mapsto\tilde{\KK}(\xi)$ some second measurable family of separable Hilbert spaces such that
    \begin{equation}
        \HH\cong \int_\RR^\oplus \tilde{\KK}(\xi)\diff\xi,
    \end{equation}
    where for any square integrable $\tilde{\psi}(\xi)\in \tilde{\KK}(\xi)$ one has 
    \begin{equation}
        (U(t,1_\HH)\tilde{\psi})(\xi)=\exp(it\xi)\tilde{\psi}(\xi).
    \end{equation}
    Hence there exists a unitary $\tilde{U}:\int_\RR^\oplus \tilde{\KK}_\xi\diff\xi\to \int_\RR^\oplus \KK_\xi\diff\xi$ such that for any $t\in \RR$
    \begin{equation}
        (\tilde{U}\exp(it.)\tilde{\psi})(\xi)=\exp(it\xi)(\tilde{U}\tilde{\psi})(\xi).
    \end{equation}
    By linearity and continuity of $\tilde{U}$, this can be extended to any $f\in L^\infty(\RR)$, i.e.
    \begin{equation}
        (\tilde{U}(f\tilde{\psi}))(\xi)=f(\xi)(\tilde{U}\tilde{\psi})(\xi).
    \end{equation}
    This means in particular that for any square integrable vectorfields $\tilde{\psi}(\xi),\tilde{\psi}'(\xi)$ and $f\in C_c(\RR)$, one has 
    \begin{equation}
        \int_\RR f(\xi)\langle \tilde{U}\tilde{\psi}'(\xi),\tilde{U}\tilde{\psi}(\xi)\rangle\diff\xi=\int_\RR f(\xi)\langle \tilde{\psi}'(\xi),\tilde{\psi}(\xi)\rangle\diff\xi,
    \end{equation}
    or $\langle \tilde{U}\tilde{\psi}'(\xi),\tilde{U}\tilde{\psi}(\xi)\rangle=\langle \tilde{\psi}'(\xi),\tilde{\psi}(\xi)\rangle$ almost everywhere.
    Now let $\{\tilde{\psi}^{(n)}(\xi)\in \tilde{\KK}(\xi)\}_{n=1}^\infty$ a sequence of measurable vectorfields such that, setting $N(\xi)=\dim(\tilde{\KK}(\xi))$, one has for each $\xi\in \RR$ that $\{\tilde{\psi}^{(n)}(\xi)\}_{n=1}^{N(\xi)}$ is an orthonormal basis for $\tilde{\KK}(\xi)$ and $\tilde{\psi}^{(n)}(\xi)=0$ for $n> N(\xi)$ (the latter only occurs for $N(\xi)<\infty$). Such a sequence always exists, see \cite[Lem.\ IV.8.12]{Takesaki2001}. For any $X\subset \RR$ having finite measure, the vectorfields $\tilde{\psi}_X^{(n)}(\xi):=\chi_X(\xi)\tilde{\psi}^{(n)}(\xi)$ are square integrable. One can now see that $\{(\tilde{U}\tilde{\psi}_{X}^{(n)})(\xi)\}_{n=1}^{N(\xi)}$ form an orthonormal set for almost all $\xi\in X$. One can conclude that $\dim(\KK_\xi)\geq N(\xi)=\dim(\tilde{\KK}_\xi)$ almost everywhere. By repeating the argument with $\tilde{\KK}(\xi)$ and $\KK(\xi)$ reversed, one concludes that $\dim(\KK(\xi))=\dim(\tilde{\KK}(\xi))$ almost everywhere and as these are separable Hilbert spaces, one concludes that $\KK(\xi)\cong \tilde{\KK}(\xi)$ almost everywhere.
    \end{proof}

    \section{Von Neumann algebra factors, types and traces}
    \label{appx:factor_types}
    In this section we recall some of the definitions relevant to the type classification of von Neumann algebras and factors due to \cite{Murray_1936}, which in this paper mainly play a role in Sec.\ \ref{sec:type_change}. A detailed treatment of this topic can be found for instance in \cite{kadison1997fundamentals} or \cite{Takesaki2001}. Recall that a von Neumann algebra is defined as follows
    \begin{definition}
    \label{def:vNeumannAlg}
        Let $\MM$ be a unital *-subalgebra of $B(\HH)$ for some Hilbert space $\HH$, we say $\MM$ is a \textup{von Neumann algebra} if
        \begin{equation}
            \MM''=\MM,
        \end{equation}
        where for any $S\subset B(\HH)$, the commutant $S'$ is defined by
        \begin{equation}
            S'=\{a\in B(\HH):\text{ for all }s\in S,\, as=sa \}.
        \end{equation}
    \end{definition}
    A factor is then defined as follows.
    \begin{definition}
        A von Neumann algebra $\MM$ is a \textup{(von Neumann) factor} if it has trivial centre, i.e.
        \begin{equation}
            \MM\cap\MM'=\CC 1_{\MM}.
        \end{equation}
    \end{definition}
    A key feature of von Neumann algebras acting on separable Hilbert spaces is that they can always be decomposed into factors, see e.g.\ \cite[Thm.\ IV.8.21]{Takesaki2001}.
    \begin{theorem}
        Let $\MM$ a von Neumann algebra acting on a separable Hilbert space $\HH$, then there exists a $\sigma$-finite measure space $(X,\mu)$ and a measurable family of separable Hilbert spaces $X\ni x\mapsto \HH(x)$ and of von Neumann algebra factors $X\ni x\mapsto \MM(x)$ such that $\MM(x)$ acts on $\HH(x)$ and there is an isomorphism of Hilbert spaces
        \begin{equation}
            \HH\cong \int_X^\oplus\HH(x)\diff\mu(x),
        \end{equation}
        under which
        \begin{equation}
            \MM\cong \int_X^\oplus\MM(x)\diff\mu(x).
        \end{equation}
    \end{theorem}
    See \cite[Ch.\ IV.8]{Takesaki2001} for details on measurable fields of Hilbert spaces and von Neumann algebras. Given the existence of such decompositions, the classification of factors allows insight into the structure of a rich class of von Neumann algebras. Although it is more conventional and arguably more fundamental to describe the type classification of factors in terms of the structure of their projection lattice, we shall give their characterisation in terms of which traces they admit. Let us first recall the definition of a trace, and the more general notion of a weight, on a von Neumann algebra.
    \begin{definition}
        \label{def:weight_trace}
        Let $\mathcal{M}$ be a von Neumann algebra. A \textup{weight} on a von Neumann algebra is a map $\varphi:\MM^+\to [0,\infty]$ such that for each $x,y\in\MM$, $\lambda\in\RR^+$ we have
        \begin{equation}
            \varphi(x+y)=\varphi(x)+\varphi(y),\;\varphi(\lambda x)=\lambda\varphi(x),
        \end{equation}
        where we use the convention $0\cdot\infty=0$. A \textup{trace} is a weight $\tau:\MM^+\to [0,\infty]$ that furthermore satisfies for each $x\in\mathcal{M}$
        \begin{equation}
            \tau(x^*x)=\tau(xx^*).
        \end{equation}
        A \textup{tracial state} is a state on $\MM$ that restricts to a trace on $\MM^+$, where in particular
        \begin{equation}
            \tau(1_{\MM})=1.
        \end{equation}
        We say a weight is \textup{faithful} if for $x\in\MM$ we have $\varphi(x^*x)=0$ if and only if $x=0$, \textup{semifinite} if for every $x\in\MM^+$ with $x>0$ there is a $y\in\MM^+$ with $0<y\leq x$ for which we have $\varphi(y)<\infty$, \textup{finite} if $\tau(1_\MM)<\infty$ and \textup{normal} if for every bounded increasing net $x_\alpha$ in $\MM$, we have $\sup_{\alpha}\varphi(x_\alpha)=\varphi(\sup_\alpha x_\alpha)$.
    \end{definition}
    One can use the notion of a trace to distinguish classes of von Neumann algebras. For simplicity we shall restrict ourselves here to the case of von Neumann algebras acting on separable Hilbert spaces.
    \begin{definition}
        A von Neumann algebra $\MM\subset B(\HH)$ for $\HH$ separable is \textup{(semi)finite} if it admits a faithful (semi)finite normal trace.
    \end{definition}
    For general von Neumann algebras, the characterisation of finite algebras based on traces is somewhat more involved, see \cite[Ch.~V.2]{Takesaki2001}.
    
    Note that it need in general not be true that any faithful normal trace on a semifinite von Neumann algebra is semifinite. Furthermore, a general semifinite von Neumann algebra may admit many linearly independent faithful semifinite normal traces. For semifinite factors however, the freedom in chosing a semifinite normal trace is significantly restricted, see e.g.\ \cite[Cor.\ V.2.32]{Takesaki2001}.
    \begin{proposition}
        All (nonzero) semifinite normal traces on a semifinite factor are proportional.
    \end{proposition}
    One can use this fact to classify factors into mutually exclusive types in terms of properties of their traces. This classification is equivalent to the more conventional classification in terms of the structure of the projection lattice of a factor, see in particular \cite[Prop.\ 8.5.3-5]{kadison1997fundamentals}.
    \begin{definition}
        Let $\MM$ be a von Neumann factor and let $\MM_p\subset\MM^+$ be the set of all its projections
        \begin{equation}
            \MM_p:=\{a\in\MM:a=a^*=a^2\}.
        \end{equation}
        We distinguish the following cases:
        \begin{itemize}
            \item We say $\MM$ is type $\mathrm{I}$ if:
            \begin{itemize}
                \item $\MM$ admits a faithful normal tracial state $\tau$ with $\tau(\MM_p)=\{0,1/n,2/n,...,1\}$ for some $n\in\mathbb{N}$, in which case we say $\MM$ is type $\mathrm{I}_n$, or
                \item $\MM$ admits a faithful normal semifinite trace $\tau$ with $\tau(\MM_p)=\mathbb{N}_0\cup\{\infty\}$, in which case we say $\MM$ is type $\mathrm{I}_\infty$.
            \end{itemize}
            \item We say $\MM$ is type $\mathrm{II}$ if:
            \begin{itemize}
                \item $\MM$ admits a faithful normal tracial state $\tau$ with $\tau(\MM_p)=[0,1]$, in which case we say $\MM$ is type $\mathrm{II}_1$, or
                \item $\MM$ admits a faithful normal semifinite trace $\tau$ with $\tau(\MM_p)=[0,\infty]$, in which case we say $\MM$ is type $\mathrm{II}_\infty$.
            \end{itemize}
            \item We say $\MM$ is type $\mathrm{III}$ if it is not semifinite.
        \end{itemize}
    \end{definition}
    It can be shown that this classification exhausts all factors. \\
    
    Though type III factors are typically not further distinguished based on properties of a trace defined on them, various finer classifications of type III factors have been introduced. In particular we recall a classification due to \cite{Connes1973}, see also \cite[Ch.\ XII]{takesaki2002}, based on the spectrum of modular operators. We sketched the construction of the modular operator for a faithful positive normal linear functional below Thm.\ \ref{thm:modular_action}. This construction can be generalised for faithful semifinite normal weights, see e.g. \cite[Thm.\ VIII.1.2]{takesaki2002}. For such a weight $\varphi$, we denote $\Delta_\varphi$ its associated modular operator and $Sp(\Delta_\varphi)\subset\RR$ its spectrum. We now recall the following definition.
    \begin{definition}
        For a von Neumann algebra $\MM$, we define its \textup{modular spectrum} by
        \begin{equation}
            S(\MM)=\bigcap_{\substack{\varphi\text{ semifinite}\\\text{normal weight}}}Sp(\Delta_\varphi).
        \end{equation}
        For $\MM$ a type $\mathrm{III}$ factor, we say that
        \begin{itemize}
            \item $\MM$ is type $\mathrm{III}_0$ if $S(\MM)=\{0,1\}$,
            \item $\MM$ is type $\mathrm{III}_\lambda$ for $\lambda\in (0,1)$ if $S(\MM)=\overline{\{\lambda^n:n\in\mathbb{Z}\}}$,
            \item $\MM$ is type $\mathrm{III}_1$ if $S(M)=[0,\infty)$.
        \end{itemize}
    \end{definition}
    This classification exhausts all type III factors, which is most easily seen through an equivalent characterisation, see \cite[Def.\ XII.1.5 and Thm.\ XII.1.6]{takesaki2002}. The factors of type $\mathrm{III}_1$ are the most relevant to us, firstly because their frequent appearance in the context of local von Neumann algebras of quantum field theories, see e.g.\ \cite{fredenhagen1985modular,Verch:1996wv}, but also due to a relation between the crossed product of type $\mathrm{III}_1$ factors with its modular group and type $\mathrm{II}_\infty$ factors, see \cite[Cor. 9.7]{Takesaki1973duality}.
    \begin{proposition}
        \label{prop:type3to2fact}
        Let $\MM$ be a von Neumann algebra and $\sigma$ a modular action (associated with some faithful normal semifinite weight or positive functional), then the crossed product $\MM\rtimes_\sigma\RR$ is a type $\mathrm{II}_\infty$ factor if and only if $\MM$ is a type $\mathrm{III}_1$ factor.
    \end{proposition}
    We note that the result above can be used in a structure theory of type $\mathrm{III}_1$ factors, namely showing that any such factor $\MM$ has a natural decomposition $\MM\cong \mathcal{N}\rtimes_\vartheta\RR$ for some group action $\vartheta:\RR\times\mathcal{N}\to\mathcal{N}$ where $\mathcal{N}$ is a type $\mathrm{II}_\infty$ factor. Similar decompositions, where $\mathcal{N}$ has a nontrivial centre, can also be derived for the other type III factors, see \cite[Ch.\ XII]{takesaki2002} for a detailed treatment.

    \paragraph{Tensor products of semifinite von Neumann algebras}
    An important property of the class of semifinite von Neumann algebras is that it is closed under the (spatial) tensor product, see e.g.\ \cite[Thm.\ V.2.30]{Takesaki2001}. Here we give an alternative proof of this fact by constructing a faithful normal semifinite trace on such a tensor product of two semifinite von Neumann algebras from traces defined on these. This particular construction is of use to us in the proof of Thm.\ \ref{thm:type_change}.
    \begin{proposition}
    \label{prop:trace_product}
    Let $\MM^{(1)},\MM^{(2)}$ be semifinite von Neumann algebras acting on Hilbert spaces $\HH^{(1)},\HH^{(2)}$ with semifinite faithful normal traces $\tau^{(1)},\tau^{(2)}$, then $\MM^{(1)}\otimes\MM^{(2)}$ is semifinite with semifinite faithful normal trace $\tau$, which is characterised by the following: for $x^{(i)}\in \MM^{(i)}$ positive one has
    \begin{equation}
        \tau(x^{(1)}\otimes x^{(2)})=\tau^{(1)}(x^{(1)})\tau^{(2)}(x^{(2)}).
    \end{equation}
    \end{proposition}
\begin{proof}
    Let $\{p_\alpha^{(i)}\in \MM^{(i)}\}$ be monotonically increasing nets of finite trace projections such that
    \begin{equation}
        \sup_{\alpha}p_\alpha^{(i)}=1_{\MM^{(i)}}.
    \end{equation}
    
    Note that for any $\alpha$, the map $\tau^{(i)}_\alpha:\MM^{(i)}\to\CC$, $x\mapsto \tau^{(i)}(p^{(i)}_\alpha xp^{(i)}_\alpha)$ is a positive normal and linear map. Now define the map $\tau_{\alpha,\beta}:\MM^{(1)}\otimes\MM^{(2)}\to\CC$ by
   \begin{equation}
       \tau_{\alpha,\beta}(x)=\left(\tau_\alpha^{(1)}\otimes\tau_\beta^{(2)}\right)(x).
   \end{equation}
   By \cite[Prop.\ 11.2.7.]{kadison1997fundamentals}, this is also a normal positive linear map. Note that for any pairs $x^{(i)},y^{(i)}\in \MM^{(i)}$, one has for $x=x^{(1)}\otimes x^{(2)}, y=y^{(1)}\otimes y^{(2)}$
   \begin{equation}
       \label{eq:pretrace_property}
       \tau_{\alpha,\beta}(x(p^{(1)}_{\alpha'}\otimes p^{(2)}_{\beta'})y)=\tau_{\alpha',\beta'}(y(p^{(1)}_\alpha\otimes p^{(2)}_\beta)x).
   \end{equation}
   By linearity and normality of $\tau_{\alpha,\beta}$, this equality can be extended to arbitrary $x,y\in \MM^{(1)}\otimes\MM^{(2)}$.\\

   Now define $\tau: (\MM^{(1)}\otimes \MM^{(2)})^+\to [0,\infty]$ by
    \begin{equation}
        \tau(x)=\sup_{\alpha,\beta}\tau_{\alpha,\beta}(x).
    \end{equation}
    This is a normal map, in the sense that for each bounded increasing net $x_\gamma$, we have
    \begin{equation}
        \sup_{\gamma}\tau(x_\gamma)=\sup_{\alpha,\beta,\gamma}\tau_{\alpha,\beta}(x_\gamma)=\sup_{\alpha,\beta}\tau_{\alpha,\beta}(\sup_\gamma x_\gamma)=\tau(\sup_\gamma x_\gamma),
    \end{equation}
    where we have used normality of $\tau_{\alpha,\beta}$ and interchangeability of suprema.
    Moreover, using Eq.~\eqref{eq:pretrace_property}, one finds that for general $x\in \MM^{(1)}\otimes\MM^{(2)}$
    \begin{align}
        \tau(x^*x)=\sup_{\alpha,\beta}\tau_{\alpha,\beta}(x^*x)=&\sup_{\alpha,\beta}\sup_{\alpha',\beta'}\tau_{\alpha,\beta}(x^*(p_{\alpha'}^{(1)}\otimes p_{\beta'}^{(2)})x)\nonumber\\
        =&\sup_{\alpha,\beta,\alpha',\beta'}\tau_{\alpha',\beta'}(x(p_\alpha^{(1)}\otimes p_\beta^{(2)})x^*)=\tau(xx^*).
    \end{align}
    Therefore $\tau$ is a trace. We can now show that $\tau$ is semifinite. Let $x^*x\in(\MM^{(1)}\otimes \MM^{(2)})^+$, then for any $\alpha,\beta$ we have $x^*x\geq x^*(p_\alpha^{(1)}\otimes p_\beta^{(2)})x$, for which we have
    \begin{equation}
        \tau(x^*(p_\alpha^{(1)}\otimes p_\beta^{(2)})x)=\tau((p_\alpha^{(1)}\otimes p_\beta^{(2)})xx^*(p_\alpha^{(1)}\otimes p_\beta^{(2)}))=\tau_{\alpha,\beta}(xx^*)<\infty.
    \end{equation}

        Lastly, we show that $\tau$ is faithful. Suppose $\tau(x^*x)=0$, then also $\tau_{\alpha,\beta}(x^*x)=0$, and hence $\tau_{\alpha,\beta}(x^*(p_\alpha^{(1)}\otimes p_\beta^{(2)})x)=0$.
    Note that $\tau^{(i)}$ defines a faithful normal positive bounded linear functional on $\mathcal{N}^{(i)}_\alpha:=p_\alpha^{(i)}\MM^{(i)}p_\alpha^{(i)}$, as by \cite[Ch.\ V.2.\ Eq.\ (2)]{Takesaki2001}
    \begin{equation}
        \vert\tau^{(i)}(p_\alpha^{(i)}x^{(i)}p_\alpha^{(i)})\vert\leq \Vert p_\alpha^{(i)}x^{(i)}p_\alpha^{(i)}\Vert \tau^{(i)}(p_\alpha^{(i)}),
    \end{equation} which define faithful normal representations via the G.N.S. construction $(\pi_\alpha^{(i)},\HH_\alpha^{(i)},\Omega_\alpha^{(i)})$ with $\Omega_\alpha^{(i)}$ a cyclic and separating vector such that
    \begin{equation}
        \langle \Omega_\alpha^{(i)},\pi_\alpha^{(i)}(p_\alpha^{(i)}x^{(i)}p_\alpha^{(i)})\Omega_\alpha^{(i)}\rangle=\tau^{(i)}(x^{(i)}p_\alpha^{(i)}).
    \end{equation} By \cite[Thm.\ IV.5.2]{Takesaki2001} $(\pi_\alpha^{(1)}\otimes \pi_\beta^{(2)},\HH_\alpha^{(1)}\otimes\HH_\beta^{(2)})$ defines a faithful representation of 
    \begin{equation}
        \mathcal{N}_\alpha^{(1)}\otimes \mathcal{N}_\beta^{(2)}=(p_\alpha^{(1)}\otimes p_\beta^{(2)})(\MM^{(1)}\otimes \MM^{(2)})(p_\alpha^{(1)}\otimes p_\beta^{(2)}).
    \end{equation}
    As $\Omega_\alpha^{(i)}$ is separating for $\pi_\alpha^{(i)}(\mathcal{N}_\alpha^{(i)})=\pi_\alpha^{(i)}(\mathcal{N}_\alpha^{(i)})''$, one finds it is cyclic for $\pi_\alpha^{(i)}(\mathcal{N}_\alpha^{(i)})'$, i.e. $\pi_\alpha^{(i)}(\mathcal{N}_\alpha^{(i)})'\Omega_\alpha^{(i)}$ is dense in $\HH_\alpha^{(i)}$. It follows that 
    \begin{equation}
        (\pi_\alpha^{(1)}(\mathcal{N}_\alpha^{(i)})'\otimes \pi_\beta^{(2)}(\mathcal{N}_\beta^{(i)})')(\Omega_\alpha^{(1)}\otimes \Omega_\beta^{(2)})
    \end{equation} 
    is dense in $\HH_\alpha^{(1)}\otimes\HH_\beta^{(2)}$, hence $\Omega_\alpha^{(1)}\otimes \Omega_\beta^{(2)}$ is cyclic for 
    \begin{equation}
        \pi_\alpha^{(1)}(\mathcal{N}_\alpha^{(1)})'\otimes \pi_\beta^{(2)}
    (\mathcal{N}_\beta^{(2)})'=(\pi_\alpha^{(1)}(\mathcal{N}_\alpha^{(1)})\otimes \pi_\beta^{(2)}(\mathcal{N}_\beta^{(2)}))',
    \end{equation} see \cite[Thm.\ IV.5.9]{Takesaki2001}. This in turn implies that $\Omega_\alpha^{(1)}\otimes \Omega_\beta^{(2)}$ is separating for $(\pi_\alpha^{(1)}(\mathcal{N}_\alpha^{(1)})\otimes \pi_\beta^{(2)}(\mathcal{N}_\beta^{(2)}))$, and as
    \begin{equation}
        \tau_{\alpha,\beta}(x(p_\alpha^{(1)}\otimes p_\beta^{(2)}))=\langle \Omega_\alpha^{(1)}\otimes \Omega_\beta^{(2)},(\pi_\alpha^{(1)}\otimes \pi_\beta^{(2)})((p_\alpha^{(1)}\otimes p_\beta^{(2)})x(p_\alpha^{(1)}\otimes p_\beta^{(2)})) (\Omega_\alpha^{(1)}\otimes \Omega_\beta^{(2)})\rangle,
    \end{equation}
    we find that $\tau_{\alpha,\beta}$ is faithful on $(p_\alpha^{(1)}\otimes p_\beta^{(2)})(\MM^{(1)}\otimes \MM^{(2)})(p_\alpha^{(1)}\otimes p_\beta^{(2)})$.

    We now see that $\tau(x^*x)=0$ implies $(p_\alpha^{(1)}\otimes p_\beta^{(2)})x(p_\alpha^{(1)}\otimes p_\beta^{(2)})=0$, and as 
    \begin{equation}
        x=\operatorname{w-lim}_{\alpha,\beta}(p_\alpha^{(1)}\otimes p_\beta^{(2)})x(p_\alpha^{(1)}\otimes p_\beta^{(2)}),
    \end{equation}
    we find $x=0$. Hence, $\tau$ is faithful.
\end{proof}

\section{On thermal properties of quantum reference frames}
\label{appx:QRF_thermal_weight}
In this appendix we prove some results used in the discussion of Sec.\ \ref{sec:thermal_int}, these are all presented as lemmas in this appendix.
\begin{lemma}
\label{lem:pos_atl}
    For a principal QRF $(U_R,E,\HH_R)$ for the group $\RR$, the set of positive $t$-integrable operators $\mathcal{P}$ defined in Def.\ \ref{def:pos_atl} is a hereditary convex subcone of $B(\HH_R)^+$. For $f\in L^1(\RR)\cap L^\infty(\RR)^+$, i.e. a bounded positive $L^1$ function, we have that $E(f):=\int_\RR f(t)\diff E(t)\in \mathcal{P}$. Furthermore, the set $\mathcal{D}$ defined in \ref{def:pos_atl} is a weakly and strongly dense *-subalgebra of $B(\HH_R)$.
\end{lemma}
\begin{proof}
    Let us first note that $C_a=C_{\sqrt{a^*a}}$, for $C_a$ as defined in \eqref{eq:Ca}. Considering $x,y\in\mathcal{P}$, $\lambda\in \RR^+$, and any $\psi\in \HH_R$, we have
    \begin{equation}
        I_{\sqrt{\lambda x},\psi}=\int_\RR \langle\psi, U_R(t)\lambda xU_R(t)^*\psi\rangle\diff t=\lambda\int_\RR \langle\psi, U_R(t)xU_R(t)^*\psi\rangle\diff t=\lambda I_{\sqrt{x},\psi},
    \end{equation}
    and similarly
    \begin{equation}
        I_{\sqrt{x+y},\psi}=\int_\RR \langle\psi, U_R(t)(x+y)U_R(t)^*\psi\rangle\diff t= I_{\sqrt{x},\psi}+I_{\sqrt{y},\psi}.
    \end{equation}
    It follows that $C_{\sqrt{\lambda x}}=\sqrt{\lambda}C_{\sqrt{x}}$ and $C_{\sqrt{x+y}}\leq \sqrt{C_{\sqrt{x}}^2+C_{\sqrt{y}}^2}$, hence $\lambda x,\,x+y\in\mathcal{P}$ and therefore $\mathcal{P}$ is a positive cone. 

    That $\mathcal{P}$ is hereditary, i.e. for every $x\in \mathcal{P}$ and $y\in B(\HH_R)^+$ with $y\leq x$ we have $y\in\mathcal{P}$, follows from the fact that for $0\leq y\leq x$
    \begin{equation}
        \Vert \sqrt{y}U_R(t)^*\psi\Vert\leq \Vert \sqrt{x}U_R(t)^*\psi\Vert,
    \end{equation}
    hence $C_{\sqrt{y}}\leq C_{\sqrt{x}}$.

    To show that $E(f)\in\mathcal{P}$ for $f\in L^1(\RR)\cap L^\infty(\RR)^+$, recall that by Theorem~\ref{thm:qref_embed} there exists a separable Hilbert space $\KK$ and an isometric embedding $W:\HH_R\to L^2(\RR)\otimes\KK$ such that for $t\in \RR$ we have $U_R(t)W^*=W^*(\lambda_\RR(t)\otimes 1_{\KK})$ and for each $f\in L^\infty(\RR)$, $E(f)=W^* T_{f}\otimes 1_{\KK} W$. Consider now $f\in L^1(\RR)^+\cap L^\infty(\RR)$ and $\Psi\in L^2(\RR)\otimes \KK$; we calculate
    \begin{align}
        \langle W^*\Psi,U_R(t)E(f)U_R(t)^*W^*\Psi\rangle =& \langle WW^*\Psi, (\lambda_{\RR}(t)T_{f}\lambda_\RR(t)^*\otimes 1_{\KK}) WW^*\Psi\rangle\nonumber\\
        =&\int_\RR f(s-t)\Vert (WW^*\Psi)(s)\Vert_\KK^2\diff s.
    \end{align}
    Hence, by Tonelli's theorem, we find
    \begin{align}
    \label{eq:Ef_integral}
        I_{\sqrt{E(f)},W^*\Psi}=&\int_\RR\int_\RR f(s-t)\Vert (WW^*\Psi)(s)\Vert_\KK^2\diff s\diff t\nonumber\\
        =&\left(\int_\RR f(t)\diff t\right)\Vert (WW^*\Psi)\Vert^2
        =\Vert f\Vert_{L^1(\RR)}\Vert W^*\Psi\Vert^2.
    \end{align}
    It follows that $C_{\sqrt{E(f)}}=\sqrt{\Vert f\Vert_{L^1(\RR)}}$, and thus $E(f)\in\mathcal{P}$ for $f\in L^1(\RR)^+\cap L^\infty(\RR)$.

    Since $\mathcal{P}$ is a hereditary convex subcone of $B(\HH_R)^+$, it follows from \cite[Lem.\ VII.1.2]{takesaki2002} that $\mathcal{P}$ defines a left ideal of $B(\HH_R)$
\begin{equation}
    \mathcal{N}=\{a\in B(\HH_R):a^*a\in\mathcal{P}\},
\end{equation}
and a hereditary *-subalgebra of $B(\HH_R)$
\begin{equation}
    \mathcal{D}'=\left\{\sum_{i=1}^N a_i^*b_i:N\in\mathbb{N},\,a_i,b_i\in\mathcal{N}\right\}\subset\mathcal{N},
\end{equation}
such that $\mathcal{D}'\cap B(\HH_R)^+=\mathcal{P}$ and each element $x\in\mathcal{D}'$ can be written as a linear combination of four elements in $\mathcal{P}$. Since $\mathcal{D}'$ is in particular closed under linear combinations, we see that $\mathcal{D}'$ is the linear span of $\mathcal{P}$, hence $\mathcal{D}'=\mathcal{D}$.

To show that $\mathcal{D}$ is weakly dense in $B(\HH_R)$, we first show that $\mathcal{N}$ is dense in $B(\HH_R)$. Note that for any $a\in B(\HH_R)$ and for $K\subset\RR$ compact we have $
    aE(K)\in \mathcal{N}$ due to the fact that $\mathcal{N}$ is a left ideal of $B(\HH_R)$. Moreover, for any $\Psi,\Psi'\in L^2(\RR)\otimes\KK$, we have
    \begin{equation}
        \langle W^*\Psi,aE(K)W^*\Psi'\rangle=\int_K\langle (Wa^*W^*\Psi)(t),(WW^*\Psi')(t)\rangle_{\KK}\diff t.
    \end{equation}
    Observe that in the limit $K\uparrow\RR$ this tends to 
    \begin{equation}
        \int_\RR\langle (Wa^*W^*\Psi)(t),(WW^*\Psi')(t)\rangle_{\KK}\diff t=\langle W^*\Psi,aW^*\Psi'\rangle.
    \end{equation}
    hence the net $aE(K)$ converges weakly to $a$ and therefore $\mathcal{N}$ is weakly dense in $B(\HH_R)$. Similarly, for $a\in\mathcal{N}$, one has that $E(K)a\in\mathcal{D}$ for $K\subset\RR$ compact and that this net weakly converges to $a$, hence $\mathcal{D}$ is weakly dense in $\mathcal{N}$ and therefore also in $B(\HH_R)$. Furthermore, since $\mathcal{D}$ is a *-algebra of $B(\HH_R)$, we conclude from \cite[Thm.\ II.2.6]{Takesaki2001} that $\mathcal{D}\cap S$ is strongly dense in $S$ for $S$ the unit ball of $B(\HH_R)$, from which it follows that $\mathcal{D}$ is strongly dense in $B(\HH_R)$.    
\end{proof}
\begin{lemma}
    \label{lem:weight_normal}
     The weight $\varphi_\beta$ for a QRF $(U_R,E,\HH_R)$ as in Def.\ \ref{def:beta_weight} is normal.
\end{lemma}
\begin{proof}
    Let $\{x_\alpha\}\subset B(\HH_R)$ be a monotone increasing bounded net. We have
    \begin{align}
        \sup_\alpha\varphi_\beta(x_\alpha)=&\sup_{K,N}\sum_{n=1}^N\sup_\alpha\int_\RR \langle W^*\Psi_{K,n}^{(\beta)},U_R(t)x_\alpha U_R(t)^*W^*\Psi_{K,n}^{(\beta)}\rangle\diff t\nonumber\\
        =&\sup_{K,K',N}\sum_{n=1}^N\sup_\alpha\int_{K'} \langle W^*\Psi_{K,n}^{(\beta)},U_R(t)x_\alpha U_R(t)^*W^*\Psi_{K,n}^{(\beta)}\rangle\diff t,
    \end{align}
    for $K,K'\subset\RR$ compact. Here we recall that $W:\HH_R\to L^2(\RR)\otimes \KK$ an isometry and $\Psi^{(\beta)}_{K,n}=\mathcal{F}^*f_{K}^{(\beta)}\otimes \psi_n$ for $f_{K}^{(\beta)}(\xi)=\chi_K(\xi)\exp(-\frac{\beta}{2}\xi)$ such that $f_K\in L^2(\RR)$ and $\{\psi_n\}_{n=1}^\infty$ an orthonormal basis of $\KK$. We now note that $K'\ni t\mapsto \langle W^*\Psi_{K,n}^{(\beta)},U_R(t)x_\alpha U_R(t)^*W^*\Psi_{K,n}^{(\beta)}\rangle$ is continuous due to strong continuity of $U_R$. For fixed $K'$, let $\mathbb{N}\ni m\mapsto q_m$ be a bijection of $\mathbb{N}$ and $\mathbb{Q}\cap K'$ and let $\varepsilon>0$. Since for fixed $t\in K'$ one has 
    \begin{equation}
        \sup_\alpha\langle W^*\Psi_{K,n}^{(\beta)},U_R(t)x_\alpha U_R(t)^*W^*\Psi_{K,n}^{(\beta)}\rangle=\langle W^*\Psi_{K,n}^{(\beta)},U_R(t)\sup_\alpha x_\alpha U_R(t)^*W^*\Psi_{K,n}^{(\beta)}\rangle<\infty,
    \end{equation}
    one can for each $M\in\mathbb{N}$ choose an $x_M\in\{x_\alpha\}$ such that for each $m<M$ one has
    \begin{equation}
         0\leq \langle W^*\Psi_{K,n}^{(\beta)},U_R(q_m)(\sup_\alpha x_\alpha-x_M) U_R(q_m)^*W^*\Psi_{K,n}^{(\beta)}\rangle<\varepsilon,
    \end{equation}
    and such that $\{x_M\}_{M\in\mathbb{N}}$ is a (bounded) increasing sequence. For this sequence, we note that $\lim_{M\to\infty} x_M\leq \sup_\alpha x_\alpha$ and that for each $q\in\mathbb{Q}\cap K'$
    \begin{equation}
        \lim_{M\to\infty}\langle\Psi_{K,n}^{(\beta)},U_R(q)x_M U_R(q)^*W^*\Psi_{K,n}^{(\beta)}\rangle>\langle\Psi_{K,n}^{(\beta)},U_R(q)\sup_\alpha x_\alpha U_R(q)^*W^*\Psi_{K,n}^{(\beta)}\rangle-\varepsilon,
    \end{equation}
    Due to continuity, one can extend this lower bound to all of $K'$. Now we see by dominated convergence that 
    \begin{align}
        &\int_{K'} \langle W^*\Psi_{K,n}^{(\beta)},U_R(t)\sup_\alpha x_\alpha U_R(t)^*W^*\Psi_{K,n}^{(\beta)}\rangle\diff t\nonumber\\
        &\geq \sup_\alpha\int_{K'} \langle W^*\Psi_{K,n}^{(\beta)},U_R(t) x_\alpha U_R(t)^*W^*\Psi_{K,n}^{(\beta)}\rangle\diff t\nonumber\\
        &\geq\lim_{M\to\infty}\int_{K'} \langle W^*\Psi_{K,n}^{(\beta)},U_R(t) x_M U_R(t)^*W^*\Psi_{K,n}^{(\beta)}\rangle\diff t\nonumber\\
        &>\int_{K'} \left[\langle W^*\Psi_{K,n}^{(\beta)},U_R(t) \sup_\alpha x_\alpha U_R(t)^*W^*\Psi_{K,n}^{(\beta)}\rangle-\varepsilon\right]\diff t.
    \end{align}
    Since $K'$ is compact and $\varepsilon>0$ was chosen arbitrarily, it follows that
    \begin{align}
        &\sup_\alpha\int_{K'} \langle W^*\Psi_{K,n}^{(\beta)},U_R(t) x_\alpha U_R(t)^*W^*\Psi_{K,n}^{(\beta)}\rangle\diff t\nonumber\\
        &=\int_{K'} \langle W^*\Psi_{K,n}^{(\beta)},U_R(t) \sup_\alpha x_\alpha U_R(t)^*W^*\Psi_{K,n}^{(\beta)}\rangle\diff t.
    \end{align}
    From this one concludes that $\varphi_\beta$ is normal.    
\end{proof}
\begin{lemma}
    \label{lem:weight_type_change}
    For $\varphi_\beta$ the normal weight for a QRF $(U_R,E,\HH_R)$ as in Def.\ \ref{def:beta_weight} and $\mathcal{P}$ the positive $t$-integrable operators on $\HH_R$, we have that Eq.\ \eqref{eq:type_change_condition} holds if and only if 
    \begin{equation}
        \varphi_\beta(\mathcal{P})=\RR^+,
    \end{equation}
    i.e.\ $\varphi_\beta(x)<\infty$ for $x\in\mathcal{P}$.
\end{lemma}
\begin{proof}
    Let $x=a^*a\in\mathcal{P}$ with $C_a<\infty$. Note that by definition, for any vector $\psi\in\HH_R$ we have
    \begin{equation}
        I_{a,\psi}=\int_\RR\Vert aU_R(t)^*\psi\Vert^2\diff t\leq C_a^2\Vert\psi\Vert.
    \end{equation}
    Hence we see that
    \begin{align}
        \varphi_\beta(x)=&\sup_{K,N}\sum_{n=1}^N I_{a,W^*\Psi_{K,n}^{(\beta)}}\nonumber\\
        \leq&C_a^2\sup_{K,N}\sum_{n=1}^N \Vert W^*\Psi_{K,n}^{(\beta)}\Vert^2\nonumber\\
        =&C_a^2\sup_{K,N}\sum_{n=1}^N \int_K \exp(-\beta\xi)\langle\psi_n, \hat{p}(\xi)\psi_n\rangle\diff\xi\nonumber\\
        =&C_a^2\int_\RR \exp(-\beta\xi)m_{U_R}(\xi)\diff\xi,
    \end{align}
    where we have used \eqref{eq:multiplicity function in a basis} and recall that $\hat{p}(\xi)$ is the measurable family of projections defined by $\hat{p}=(\mathcal{F}\otimes 1_{\KK})WW^*(\mathcal{F}^*\otimes 1_{\KK})$ as in the proof of Prop.\ \ref{prop:spect_decomp}, see Appendix \ref{appx:spect_decomp_proof}. It follows that $\varphi_\beta(x)$ is finite provided that Eq.\ \eqref{eq:type_change_condition} holds. 

    To invert the implication, consider now $f\in L^1(\RR)^+\cap L^\infty(\RR)$, for which $E(f)\in\mathcal{P}$ by Lem.\ \ref{lem:pos_atl}. We verify that
    \begin{align}
        \varphi_\beta(E(f))=&\sup_{K,N}\sum_{n=1}^N I_{\sqrt{E(f)},W^*\Psi_{K,n}^{(\beta)}}\nonumber\\
        =&\Vert f\Vert_{L^1(\RR)}\sup_{K,N}\sum_{n=1}^N \Vert W^*\Psi_{K,n}^{(\beta)}\Vert^2\nonumber\\
        =&\Vert f\Vert_{L^1(\RR)}\int_\RR \exp(-\beta\xi)m_{U_R}(\xi)\diff\xi,
    \end{align}
    where we have used Eq.\ \eqref{eq:Ef_integral}. Hence Eq.\ \eqref{eq:type_change_condition} holds if $\varphi_\beta(E(f))<\infty$ for some $f\in L^1(\RR)\cap L^\infty(\RR)^+$ with $\Vert f\Vert_{L^1(\RR)}>0$.
\end{proof}

\begin{lemma}
    \label{lem:time_trace_completeness}
    Let $(U_R,E,\HH_R)$ be a QRF as in Def.\ \ref{def:pos_atl} and $\mathcal{N}=\{a\in B(\HH_R):a^*a\in\mathcal{P}\}$ the ideal given in the proof of Lem.\ \ref{lem:pos_atl} with $\mathcal{D}\subset\mathcal{N}$ the $t$-integrable operators as in Def.\ \ref{def:pos_atl}. Given $W:\HH_R\to L^2(\RR)\otimes\KK$ as in Thm.\ \ref{thm:qref_embed}, let $\{\psi_n\}_{n=1}^\infty$ be an orthonormal basis for $\KK$. Define for $a\in\mathcal{N}$, $\psi\in\HH_R$, $n\in\mathbb{N}$ and $K\subset\RR$ compact, the map $\RR^2\ni(t,s)\mapsto g_{a,\psi,n,K}(t,s)$ by
        \begin{equation}
            g_{a,\psi,n,K}(t,s)=\langle W^*(\mathcal{F}^*\chi_K\otimes \psi_n),U_R(s)a U_R(t)^*\psi\rangle.
        \end{equation}
        Then
        \begin{equation}
            \int_\RR\langle \psi, U_R(t)a^*a U_R(t)^*\psi\rangle\diff t =\frac{1}{2\pi}\sup_{K,N}\sum_{n=1}^N\int_{\RR^2}\vert g_{a,\psi,n,K}(t,s)\vert^2\diff t\diff s.
        \end{equation}
        Additionally, for any $b\in \mathcal{N}$, $\psi'\in \HH_R$, we have
        \begin{equation}
            \label{eq:time_trace_completeness}
            \int_\RR\langle \psi, U_R(t)a^*b U_R(t)^*\psi'\rangle\diff t =\frac{1}{2\pi}\lim_{K\uparrow\RR}\sum_{n=1}^\infty\int_{\RR^2}\overline{g_{a,\psi,n,K}(t,s)}g_{b,\psi',n,K}(t,s)\diff t\diff s.
        \end{equation}
\end{lemma}
\begin{proof}
    Observe that for $a\in\mathcal{N}$
    \begin{align}
        g_{a,\psi,n,K}(t,s)=&\langle (\mathcal{F}^*\chi_K\otimes \psi_n),(\lambda_\RR(s)\otimes 1_\KK)Wa U_R(t)^*\psi\rangle\nonumber\\
        =&\int_K\langle \psi_n,((\mathcal{F}\lambda_\RR(s)\otimes 1_\KK)Wa U_R(t)^*\psi)(\xi)\rangle_\KK\diff\xi\nonumber\\
        =&\int_K\exp(i\xi s)\langle \psi_n,((\mathcal{F}\otimes 1_\KK)Wa U_R(t)^*\psi)(\xi)\rangle_\KK\diff\xi
        =\sqrt{2\pi}(\mathcal{F}^*h_t)(-s),
    \end{align}
    where 
    \begin{equation}
        h_t(\xi)=\chi_K(\xi)\langle \psi_n,((\mathcal{F}\otimes 1_\KK)Wa U_R(t)^*\psi)(\xi)\rangle_\KK.
    \end{equation}
    Note that $h_t$ is a compactly supported $L^2$ function, hence in particular $h_t\in L^2(\RR)\cap L^1(\RR)$.
    Now consider 
    \begin{align}
        \int_{\RR^2}\vert g_{a,\psi,n,K}(t,s)\vert^2\diff s\diff t=&2\pi\int_{\RR^2}\vert(\mathcal{F}^*h_t)(-s)\vert^2\diff s\diff t\nonumber\\
        =&2\pi\int_{\RR^2}\vert(h_t)(\xi)\vert^2\diff \xi\diff t\nonumber\\
        =&2\pi\int_\RR\int_K \vert\langle \psi_n,((\mathcal{F}\otimes 1_\KK)Wa U_R(t)^*\psi)(\xi)\rangle_\KK\vert^2\diff\xi\diff t\nonumber\\
        =&2\pi\int_\RR\langle \psi,U_R(t)a^*p_{K,n}aU_R(t)^*\psi\rangle\diff t,
    \end{align}
    where we have used the Plancherel identity and have denoted 
    \begin{equation}
        p_{K,n}=W^*\left(\mathcal{F}^*T_{\chi_K}\mathcal{F}\otimes \vert\psi_n\rangle\langle\psi_n\vert\right)W.
    \end{equation} From the fact that $\{\psi_n\}$ forms a basis of $\KK$, we can conclude that
    \begin{equation}
        \sup_{N}\sum_{n=1}^N\int_{\RR^2}\vert g_{a,\psi,n,K}(t,s)\vert^2\diff s\diff t=2\pi\int_\RR\Vert(T_{\chi_K}\mathcal{F}\otimes 1_\KK)Wa U_R(t)^*\psi\Vert^2\diff t,
    \end{equation}
    hence
    \begin{align}
         \sup_{K,N}\sum_{n=1}^N\int_{\RR^2}\vert g_{a,\psi,n,K}(t,s)\vert^2\diff s\diff t=&2\pi\int_\RR\Vert(\mathcal{F}\otimes 1_\KK)Wa U_R(t)^*\psi\Vert^2\diff t\nonumber\\=&2\pi\int_\RR\Vert a U_R(t)^*\psi\Vert^2\diff t.
    \end{align}
    Eq.~\eqref{eq:time_trace_completeness} now follows from the polarisation identity.
\end{proof}

\begin{lemma}
    \label{lem:F_int_lim}
    Let $(U_R,E,\HH_R)$ be a QRF as in Def.\ \ref{def:pos_atl} and $\mathcal{D}$ the *-algebra of $t$-integrable operators. Given $W:\HH_R\to L^2(\RR)\otimes\KK$ as in Thm.\ \ref{thm:qref_embed}, let $\{\psi_n\}_{n=1}^\infty$ be an orthonormal basis for $\KK$ and $\Psi_{K,n}^{(z)}=(\mathcal{F}^*f_K^{(z)}\otimes \psi_n)\in L^2(\RR)\otimes\KK$ with $f_K^{(z)}(\xi)=\chi_K(\xi)\exp(-z\xi/2)$. Assume furthermore that Eq.\ \eqref{eq:type_change_condition} is satisfied for some $\beta\in (0,\infty)$. Then for $x,y\in\mathcal{D}$ the following interchange of limits is valid:
\begin{align}
    \label{eq:lim_interchange}
    \lim_{K\uparrow\RR}\sum_{n=1}^\infty \lim_{K'\uparrow\RR}\sum_{n'=1}^\infty\int_{\RR^2}& \langle W^*\Psi_{K,n}^{(2is-it)},yW^*\Psi_{K',n'}^{(\beta+2is'-it)}\rangle\nonumber\\
    &\times\langle W^*\Psi_{K',n'}^{(\beta+2is'+it)},xW^*\Psi_{K,n}^{(it+2is)}\rangle\diff s\diff s'=\nonumber\\
    \lim_{K'\uparrow\RR}\sum_{n'=1}^\infty\lim_{K\uparrow\RR}\sum_{n=1}^\infty\int_{\RR^2}& \langle W^*\Psi_{K,n}^{(2is-it)},yW^*\Psi_{K',n'}^{(\beta+2is'-it)}\rangle\nonumber\\
    &\times\langle W^*\Psi_{K',n'}^{(\beta+2is'+it)},xW^*\Psi_{K,n}^{(it+2is)}\rangle\diff s\diff s'.
\end{align}
\end{lemma}
\begin{proof}
To show that the relevant limiting operations may indeed be interchanged, we once again make use of the polarisation identity to reduce the problem of proving Eq.~\eqref{eq:lim_interchange} to showing that
\begin{align}
    &\lim_{K\uparrow\RR}\sum_{n=1}^\infty \lim_{K'\uparrow\RR}\sum_{n'=1}^\infty\int_{\RR^2} \langle W^*\Psi_{K,n}^{(2is)},a^*W^*\Psi_{K',n'}^{(\beta+2is')}\rangle\langle W^*\Psi_{K',n'}^{(\beta+2is')},aW^*\Psi_{K,n}^{(2is)}\rangle\diff s\diff s'\nonumber\\
    &\quad=\lim_{K\uparrow\RR}\sum_{n=1}^\infty \lim_{K'\uparrow\RR}\sum_{n'=1}^\infty\int_{\RR^2} \vert\langle W^*\Psi_{K',n'}^{(\beta+2is')},aW^*\Psi_{K,n}^{(2is)}\rangle\vert^2\diff s\diff s',
\end{align}
for any $a\in\mathcal{D}$.
Following the proof  of Lem.\ \ref{lem:time_trace_completeness}, we use the projection
 \begin{equation}
        p_{K,n}=W^*\left(\mathcal{F}^*T_{\chi_K}\mathcal{F}\otimes \vert\psi_n\rangle\langle\psi_n\vert\right)W,
\end{equation}
to write
\begin{align}
    &\int_{\RR^2} \vert\langle W^*\Psi_{K',n'}^{(\beta+2is')},aW^*\Psi_{K,n}^{(2is)}\rangle\vert^2\diff s\diff s'\nonumber\\
    &\quad=2\pi\int_{\RR} \langle W^*\Psi_{K',n'}^{(\beta+2is')},a^*p_{K,n}aW^*\Psi_{K',n'}^{(\beta+2is')}\rangle\diff s',
\end{align}
but also
\begin{align}
    &\int_{\RR^2} \vert\langle W^*\Psi_{K',n'}^{(\beta+2is')},aW^*\Psi_{K,n}^{(2is)}\rangle\vert^2\diff s\diff s'\nonumber\\
    &\quad=2\pi\int_{\RR} \langle W^*\Psi_{K,n}^{(2is)},a^*\tilde{p}^{(\beta)}_{K',n'}aW^*\Psi_{K,n}^{(2is)}\rangle\diff s,
\end{align}
 where 
\begin{equation}
        \tilde{p}^{(\beta)}_{K,n}=W^*\left(\mathcal{F}^*T_{\chi_K\exp(-\beta.)}\mathcal{F}\otimes \vert\psi_n\rangle\langle\psi_n\vert\right)W
\end{equation}
is a positive operator (but not a projection).
From these expressions, it is clear that the integrals 
\begin{equation}
    \int_{\RR^2} \vert\langle W^*\Psi_{K',n'}^{(\beta+2is')},aW^*\Psi_{K,n}^{(2is)}\rangle\vert^2\diff s\diff s'
\end{equation} are increasing functions of $K,K'\subset \RR$ compact (for fixed $n,n'\in\mathbb{N}$). Hence we find 
\begin{align}
    &\lim_{K\uparrow\RR}\sum_{n=1}^\infty \lim_{K'\uparrow\RR}\sum_{n'=1}^\infty\int_{\RR^2} \vert\langle W^*\Psi_{K',n'}^{(\beta+2is')},aW^*\Psi_{K,n}^{(2is)}\rangle\vert^2\diff s\diff s'\nonumber\\
    &\quad=\sup_{K,K',N,N'}\sum_{n=1}^N \sum_{n'=1}^{N'}\int_{\RR^2} \vert\langle W^*\Psi_{K',n'}^{(\beta+2is')},aW^*\Psi_{K,n}^{(2is)}\rangle\vert^2\diff s\diff s',
\end{align}
and it is clear that the order of the relevant limiting operations are interchangeable.
\end{proof}
\begin{lemma}
    \label{lem:ana_pos_atl}
    For $(U_R,E,\HH_R)$ a QRF as in Def.\ \ref{def:pos_atl} and $\mathcal{D}_1\subset B(\HH_R)$ as in Def.~\ref{def:ana_pos_atl}, one has
    \begin{enumerate}
        \item $\mathcal{D}_1\subset \mathcal{D}$ and the former is strongly and weakly dense in the latter.
        \item For $f\in L^1(\RR)$ a Gaussian, $E(f)\in\mathcal{D}_1$.
        \item For $x\in\mathcal{D}_1$, the function $t\mapsto \alpha(t,x):=U_R(t)x U_R(t)^*\in\mathcal{D}_1$ for $t\in\RR$ extends to a function $z\mapsto \alpha(z,x)\in\mathcal{D}_1$ for $z\in\CC$ that is weakly holomorphic, i.e.\ for each $\psi,\psi'\in\HH_R$ the function $z\mapsto \langle \psi,\alpha(z,x)\psi'\rangle$ is holomorphic.
    \end{enumerate}
\end{lemma}
\begin{proof}
    1. Consider $a\in B(\HH_R)$ such that $C_a<\infty$, i.e. $a^*a\in\mathcal{P}$. If $f:\RR\to\RR^+$ a positive $L^1$ function, then via Eq.~\eqref{eq:op_smear} $(a^*a)(f)\geq 0$ and
        \begin{align}
            \int_\RR\langle \psi,U_R(t)(a^*a)(f)U_R(t)^*\psi\rangle\diff t=&\int_\RR \int_\RR  f(s) \langle\psi,U_R(t+s)a^*aU_R(t+s)^*\psi\rangle\diff s\diff t\\
            \leq&\Vert f\Vert_{L^1(\RR)}C_a^2,
        \end{align}
    hence $(a^*a)(f)\in\mathcal{P}$. As one can write each element of $\mathcal{D}$ as a finite linear combination of elements of $\mathcal{P}$, and each Gaussian as a finite linear combination of bounded positive $L^1$ functions, concretely 
    \begin{align}
        f(t)=&\vert\max(\Re(f(t)),0)\vert-\vert\min(\Re(f(t)),0)\vert\nonumber\\
        &+i(\vert\max(\Im(f(t)),0)\vert-\vert\min(\Im(f(t)),0)\vert),
    \end{align} it follows that $x(f)\in\mathcal{D}$ for $x\in\mathcal{D}$ and Gaussian $f$. Since these elements generate $\mathcal{D}_1$, it follows that $\mathcal{D}_1\subset\mathcal{D}$. To show weak density, we use that $U_R$ is strongly continuous, which implies for each $\psi,\psi'$ that the function $t\mapsto\langle \psi,U_R(t)xU_R(t)^*\psi'\rangle$ is continuous. It follows that for a sequence of Gaussians $f_n$ approaching the $\delta$ distribution (pointwise when integrated against bounded continuous functions), the sequence $\langle\psi,x(f_n)\psi'\rangle$ approaches $\langle\psi,x\psi'\rangle$. Since $\mathcal{D}_1$ is a *-subalgebra of $B(\HH_R)$, the weak and strong closure of $\mathcal{D}_1$ coincide by similar arguments as in the proof of Lem.\ \ref{lem:pos_atl}.

    2. Consider $f_1,f_2\in L^1(\RR)$ Gaussians with $f_i(t)=\alpha_i\exp(-\gamma_i(t-z_i)^2)$. Note that $E(f_1)\subset\mathcal{D}$, so we have $E(f_1)(f_2)\in\mathcal{D}_1$, but by covariance of $E$ we have
    \begin{align}
        E(f_1)(f_2)=&\int_\RR f_2(t)U_R(t)E(f_1)U_R(t)^*\diff t\nonumber\\
        =&\int_\RR f_2(t)E(\lambda_\RR(t)f_1\lambda_\RR(t)^*)\diff t=E(f_1*f_2),
    \end{align}
    where $\lambda_\RR(t)f_1\lambda_\RR(t)^*\in L^1(\RR)\cap L^\infty(\RR)$ a Gaussian with with $(\lambda_\RR(t)f_1\lambda_\RR(t)^*)(s)=f_1(s-t)$ and $f_1*f_2$ the convolution given by
    \begin{equation}
        (f_1*f_2)(s)=\int_\RR f_1(s-t)f_2(t)\diff t=\gamma_1\gamma_2\sqrt{\frac{\pi}{\alpha_1+\alpha_2}}\exp\left(-\left(\frac{\alpha_1\alpha_2}{\alpha_1+\alpha_2}\right)(s-z_1-z_2)^2\right).
    \end{equation}
    One can always tune the parameters $\alpha_i,\gamma_i,z_i$ to construct arbitrary Gaussians $f$ via this convolution, hence $E(f)\in\mathcal{D}_1$.

    3. For $x\in \mathcal{D}$ and $f\in L^1(\RR)$ a Gaussian, the function $t\mapsto \alpha(t,x(f))$ is given by
    \begin{equation}
        \alpha(t,x(f))=\int_\RR f(s)U_R(s+t)xU_R(s+t)^*\diff s=\int_\RR f(s-t)U_R(s)xU_R(s)^*\diff s.
    \end{equation}
    Noting that a Gaussian $f$ can be analytically continued to a holomorphic function $f:\CC\to \CC$ and that the map $\RR\ni s\mapsto f(s-z)$ is still a Gaussian for any $z\in\CC$, it follows that the map $\CC\ni z\mapsto \alpha(z,x(f)):=\int_\RR f(s-z)U_R(s)xU_R(s)^*\diff s\in \mathcal{D}_1$ is well defined. Furthermore, for $z_0\in \CC$ some sufficiently small neighbourhood $U\subset\CC$, there are $c_1,c_2,c_3>0$ such that $\vert \partial_zf(t-z)\vert \leq c_1(\vert t\vert +c_2)\exp(c_3\vert t\vert)f(t)$ for every $t\in\RR$, $z\in\CC$, where as a function of $t$ this upper bound is again in $L^1(\RR)$. As for any $\psi,\psi'\in\HH_R$,  $t\mapsto \langle\psi,U_R(s)xU_R(s)^*\psi\rangle$ is a continuous (hence measurable) bounded function, i.e.\ an element of $L^\infty(\RR)$, we have
    \begin{equation}
        \partial_z\langle\psi,\alpha(z,x(f))\psi'\rangle=\int_\RR \partial_zf(s-z)\langle\psi,\alpha(s,x)\psi'\rangle\diff s,
    \end{equation}
    which is again well-defined. It follows that $\alpha(z,x(f))$ is a (weakly) holomorphic function of operators. We note that this argument directly generalises to products $x_1(f_1)...x_n(f_n)$ for $x_1,...,x_n\in\mathcal{D}$, $f_1,...,f_n\in L^1(\RR)\cap L^2(\RR)$ Gaussians, as we have
    \begin{multline}
        U_R(t)x_1(f_1)...x_n(f_n)U_R(t)^*=\\
        \int_{\RR}\left(\prod_{i}f_i(s_i-t)\right)U_R(s_1)x_1U_R(s_1)^*...U_R(s_n)x_nU_R(s_n)^*\diff^n s,
    \end{multline}
   and once again for $\psi,\psi'\in\HH_R$, the function 
    \begin{equation}
        (s_1,...,s_2)\mapsto\langle\psi,U_R(s_1)x_1U_R(s_1)^*...U_R(s_n)x_nU_R(s_n)^*\psi'\rangle,
    \end{equation}
    is an $L^\infty(\RR^n)$ continuous function. It follows that $\alpha(z,x_1(f_1)...x_n(f_n))$ is holomorphic.
\end{proof}

\begin{lemma}
    \label{lem:KMS_F_int}
    Let $(U_R,E,\HH_R)$ be a QRF as in Def.\ \ref{def:pos_atl} satisfying Eq.\ \eqref{eq:type_change_condition} for some $\beta\in(0,\infty)$, $\varphi_\beta$ the weight defined in Def.\ \ref{def:beta_weight} and $\mathcal{D}_1$ the *-algebra defined in Def.\ \ref{def:ana_pos_atl}. Let for $x,y\in\mathcal{D}_1$, the function $F_{x,y}:S_\beta\to \CC$ be given by
    \begin{equation}
        F_{x,y}(z):=\varphi_\beta(y\alpha(z,x)),
    \end{equation}
    with $\alpha(z,x)$ as in Lem.\ \ref{lem:ana_pos_atl}. Then $F_{x,y}$ is continuous and bounded, and analytic on the interior on $S_\beta$. Furthermore, it satisfies the boundary conditions
\begin{equation}
    F_{x,y}(t)=\varphi_\beta(y\alpha(t,x)),\, F_{x,y}(t+i\beta)=\varphi_\beta(\alpha(t,x)y).
\end{equation}
\end{lemma}
\begin{proof}
For $x,y\in\mathcal{D}_1$, $z\in \CC$, $y\alpha(z,x)\in\mathcal{D}_1\subset\mathcal{D}$ by Lem.\ \ref{lem:ana_pos_atl}, hence $F_{x,y}$ is well-defined. Let us assume $x=x_1(f_1)...x_n(f_n)$ with $x_1,...,x_n\in\mathcal{D}$, $f_1,...,f_n\in L^\infty(\RR)$ Gaussian, then 
\begin{align}
    F_{x,y}(z)=&\varphi_\beta\left(\int_{\RR^n} f_1(t_1-z)...f_n(t_n-z)y\alpha(t_1,x_1)...\alpha(t_n,x_n)\diff^n t\right)\nonumber\\
    =&\int_{\RR^n} f_1(t_1-z)...f_n(t_n-z)\varphi_\beta\left(y\alpha(t_1,x_1)...\alpha(t_n,x_n)\right)\diff^n t,
\end{align}
where we have used the Fubini-Tonelli theorem to interchange the integral in the definition of $\varphi_\beta$ (Def.~\ref{def:beta_weight}) with the integral of $\RR^n$ above and use Lebesgue's dominated convergence theorem to show to interchange the remaining limiting operations with the integral over $\RR^n$. This can be applied as 
\begin{equation}
    \vert\varphi_\beta\left(y\alpha(t_1,x_1)...\alpha(t_n,x_n)\right)\vert\leq\Vert y\alpha(t_1,x_1)...\alpha(t_n,x_n)\Vert_\mathcal{D}\int_{\RR}\exp(-\beta\xi)m_{U_R}(\xi)\diff\xi,
\end{equation}
with 
\begin{equation}
    \Vert y\alpha(t_1,x_1)...\alpha(t_n,x_n)\Vert_\mathcal{D}\leq C_{y^*}C_{x_n}\Vert x_1\Vert...\Vert x_{n-1}\Vert.
\end{equation}
Hence in particular $(t_1,...,t_n)\mapsto\varphi_\beta\left(y\alpha(t_1,x_1)...\alpha(t_n,x_n)\right)$ is a bounded function on $\RR^n$. It is also measurable, as we can construct this function from a limit of a sequence of continuous functions. Similarly as in the proof of property 3 of Lem.\ \ref{lem:ana_pos_atl} this implies that $F_{x,y}(z)$ is holomorphic. Since 
\begin{equation}
    \gamma\mapsto \int_\RR \vert f_i(t+i\gamma)\vert\diff t\text{ and }\gamma\mapsto \int_\RR \vert f_i'(t+i\gamma)\vert\diff t
\end{equation} are bounded on $\gamma\in[0,\beta]$, it follows that $F_{x,y}:S_\beta\to\CC$ is a bounded continuous function that is holomorphic on the interior of $S_\beta$. By linearity, these results now generalise to arbitrary $x,y\in\mathcal{D}_1$.

To now show that $F_{x,y}$ satisfies the KMS boundary conditions on $\mathcal{D}_1$, it is sufficient to show that for $x,y\in\mathcal{D}_1$, we have $\varphi_\beta(y\alpha(i\beta,x))=\varphi_\beta(xy)$. Here we use that as for $t\in\RR$ we have $\varphi_\beta(y\alpha(t,x))=\varphi_\beta(\alpha(-\frac{t}{2},y)\alpha(\frac{t}{2},x))$, it follows by analytic continuation that  \begin{equation}
    \varphi_\beta(y\alpha(z,x))=\varphi_\beta\left(\alpha\left(-\frac{z}{2},y\right)\alpha\left(\frac{z}{2},x\right)\right),
\end{equation}
for $z\in \CC$ and thus in particular $\varphi_\beta(y\alpha(i\beta,x))=\varphi_\beta(\alpha(-\frac{i\beta}{2},y)\alpha(\frac{i\beta}{2},x))$. We can now use Lem.\ \ref{lem:time_trace_completeness} to write
\begin{align}
    \label{eq:KMS_partway}
    \varphi_\beta(y\alpha(i\beta,x))=\frac{1}{2\pi}\lim_{K\uparrow\RR}\sum_{n=1}^\infty \lim_{K'\uparrow\RR}\sum_{n'=1}^\infty\int_{\RR^2}& \langle W^*\Psi_{K,n}^{(\beta+2is)},\alpha\left(-\frac{i\beta}{2},y\right)W^*\Psi_{K',n'}^{(2is')}\rangle\nonumber\\
    &\times\langle W^*\Psi_{K',n'}^{(2is')},\alpha\left(\frac{i\beta}{2},x\right)W^*\Psi_{K,n}^{(\beta+2is)}\rangle\diff s\diff s'
\end{align}
To show that the expression above equals $\varphi_\beta(xy)$, we first note that for $a\in\mathcal{B}(\HH_R)$, $n,n'\in\mathbb{N}$ and $K,K'\subset\RR$ compact, $z_0\in\CC$
the function \begin{equation}
    z\mapsto \langle W^*\Psi_{K,n}^{(z_0+i\overline{z})}, a W^*\Psi_{K',n'}^{(iz)}\rangle,
\end{equation} is holomorphic. This we show as follows. The (complex) derivative of this function is given by
\begin{align}
    \lim_{z'\to 0}&\frac{\langle W^*\Psi_{K,n}^{(z_0+i(\overline{z+z'}))}, a W^*\Psi_{K',n'}^{(i(z+z'))}\rangle-\langle W^*\Psi_{K,n}^{(z_0+i\overline{z})}, a W^*\Psi_{K',n'}^{(iz)}\rangle}{z'}\nonumber\\
    =\lim_{z'\to 0}&\left\langle W^*\left(\frac{\Psi_{K,n}^{(z_0+i(\overline{z+z')})}-\Psi_{K,n}^{(z_0+i\overline{z})}}{\overline{z'}}\right), a W^*\Psi_{K',n'}^{(z+z')}\right\rangle\nonumber\\
    &+\left\langle W^*\Psi_{K,n}^{(z_0+i\overline{z}))}, a W^*\left(\frac{\Psi_{K',n'}^{(i(z+z'))}-\Psi_{K',n'}^{(iz)}}{z'}\right)\right\rangle
\end{align}
Recall that
\begin{equation}
    \Psi_{K,n}^{(z)}=\mathcal{F}^*f_K^{(z)}\otimes\psi_n,
\end{equation}
with $f_K^{(z)}(\xi)=\chi_K(\xi)\exp(-\frac{z\xi}{2})$.
Clearly pointwise \begin{equation}
    \lim_{z'\to 0}\frac{1}{z'}\left(f_K^{(z+z')}(\xi)-f_K^{(z)}(\xi)\right)=-\frac{\xi}{2}f_K^{(z)}(\xi),
\end{equation} but as a limit of functions of $\xi$, this limit also holds in the $L^2$ topology, as for $\vert z'\vert<1$
\begin{align}
    \int_K \left\vert\frac{1}{z'}\left(\exp\left(-\frac{(z+z')\xi}{2}\right)-\exp\left(-\frac{z\xi}{2}\right)\right)+\frac{\xi}{2}\exp\left(-\frac{z\xi}{2}\right)\right\vert^2\diff\xi\\
    \leq\left(\sup_{\xi\in K}\left\vert\exp\left(-\frac{z\xi}{2}\right)\right\vert^2\right)\int_K \left\vert\frac{1}{z'}\left(\exp\left(-\frac{z'\xi}{2}\right)-1\right)+\frac{\xi}{2}\right\vert^2\diff\xi\nonumber\\
    \leq\frac{\vert z'\vert^2}{(1-\vert z'\vert)^2}\left(\sup_{\xi\in K}\left\vert\exp\left(-\frac{z\xi}{2}\right)\right\vert^2\right)\int_K \exp(\vert\xi\vert)\diff\xi,
\end{align}
where we have used the Cauchy estimates on series expansion of $\exp\left(-\frac{z'\xi}{2}\right)$ in terms of $z'$. It follows that 
\begin{align}
    \lim_{z'\to 0}&\frac{\langle W^*\Psi_{K,n}^{(z_0+i(\overline{z+z'}))}, a W^*\Psi_{K',n'}^{(i(z+z'))}\rangle-\langle W^*\Psi_{K,n}^{(z_0+i\overline{z})}, a W^*\Psi_{K',n'}^{(z)}\rangle}{z'}\nonumber\\
    =&\langle W^*\partial_{\overline{z}}\Psi_{K,n}^{(z_0+i\overline{z})}, a W^*\Psi_{K',n'}^{(iz)}\rangle+\langle W^*\Psi_{K,n}^{(z_0+i\overline{z})}, a W^*\partial_{z}\Psi_{K',n'}^{(iz)}\rangle,
\end{align}
with 
\begin{equation}
    \partial_z\Psi_{K,n}^{(z)}=\mathcal{F}^*\partial_zf_K^{(z)}\otimes\psi_n,
\end{equation}
where $\partial_zf_K^{(z)}\in L^2(\RR)$. Hence $z\mapsto \langle W^*\Psi_{K,n}^{(z_0+i\overline{z})}, a W^*\Psi_{K',n'}^{(iz)}\rangle$ is complex differentiable at every $z\in \CC$, i.e.\ holomorphic on $\CC$. Such functions are uniquely specified by their value on any continuous nonconstant curve $\RR\to \CC$, and since 
\begin{equation}
    \langle W^*\Psi_{K,n}^{(z_0+2it)}, a W^*\Psi_{K',n'}^{(2it)}\rangle=\langle W^*\Psi_{K,n}^{(z_0)}, \alpha(t,a) W^*\Psi_{K',n'}^{(0)}\rangle,
\end{equation}
we have for $x\in\mathcal{D}_1$
\begin{equation}
   \langle W^*\Psi_{K,n}^{(z_0)}, \alpha(z,x) W^*\Psi_{K',n'}^{(0)}\rangle=\langle W^*\Psi_{K,n}^{(z_0+2i\overline{z})},x W^*\Psi_{K',n'}^{(2iz)}\rangle.
\end{equation}
One can use this to rewrite Eq.~\eqref{eq:KMS_partway}
\begin{align}
    \varphi_\beta(y\alpha(i\beta,x))=\frac{1}{2\pi}\lim_{K\uparrow\RR}\sum_{n=1}^\infty \lim_{K'\uparrow\RR}\sum_{n'=1}^\infty\int_{\RR^2}& \langle W^*\Psi_{K,n}^{(2is)},yW^*\Psi_{K',n'}^{(\beta+2is')}\rangle\nonumber\\
    &\times\langle W^*\Psi_{K',n'}^{(\beta+2is')},xW^*\Psi_{K,n}^{(2is)}\rangle\diff s\diff s',
\end{align}
which by Lem.\ \ref{lem:F_int_lim} yields
\begin{equation}
    \varphi_\beta(y\alpha(i\beta,x))=\varphi_\beta(xy),
\end{equation}
verifying the KMS condition.
\end{proof}






\bibliography{bib_springer}

\end{document}